\documentclass[11pt,oneside]{article}

\usepackage[top=1in, bottom=1in, left=1in, right=1in]{geometry}
\usepackage{pdflscape}

\usepackage[T1]{fontenc}
\usepackage[utf8]{inputenc}
\usepackage{newtxtext}
\usepackage{booktabs}

\usepackage{amsmath,amsthm,amssymb,mathtools}
\usepackage{bbm}
\usepackage{dsfont}
\usepackage{dutchcal}
\usepackage[mathscr]{euscript}
\usepackage{subcaption}
\usepackage{natbib}
\usepackage{verbatim}
\usepackage{titlesec}
\usepackage{newtxtext}
\usepackage{ushort}

\usepackage{graphicx}
\usepackage{rotating}
\usepackage{float}
\usepackage[percent]{overpic}

\usepackage{tikz}
\usetikzlibrary{chains,shapes.multipart,shapes,calc,fit}
\usetikzlibrary{automata,positioning}

\usepackage{array}
\usepackage{booktabs}
\usepackage{longtable}
\usepackage{tabularx}
\usepackage{multirow}
\usepackage[para,online,flushleft]{threeparttable}

\usepackage{enumerate}
\usepackage[inline]{enumitem}

\usepackage[font=small]{caption}

\usepackage{color}

\usepackage{natbib}
\setcitestyle{round}
\usepackage[hidelinks]{hyperref}
\hypersetup{
    colorlinks=true,
    linkcolor=blue,        % internal links (sections, equations, figures)
    citecolor=blue,        % \cite, \citet, \citep
    urlcolor=blue,         % \url and \href
    filecolor=blue
}
\usepackage{cleveref}
\usepackage{refcount}

\usepackage{algorithm}
\usepackage{algpseudocode}

\usepackage{appendix}
\usepackage{apptools}
\AtAppendix{\renewcommand{\thesection}{Appendix \Alph{section}}}

\usepackage[pagewise]{lineno}

\usepackage{setspace}
\usepackage{changepage}
\usepackage[bottom]{footmisc}
\usepackage[all]{xy}
\usepackage{romannum}
\usepackage{authblk}

\newtheoremstyle{ModifiedStyle}
	{\topsep}{3pt}{}{}{\bfseries}{.}{.5em}{}
\theoremstyle{ModifiedStyle}

\newtheorem{proposition}{Proposition}

\newtheorem{remark}{Remark}

\newtheorem{to-do}{To-Do}

\newcommand{\BSmall}[1]{\mkern-1.7mu\raisebox{-1.2pt}{\scalebox{0.9}{$\scriptscriptstyle B$}}}
\newcommand{\NSmall}[1]{\mkern-1.7mu\raisebox{-1.2pt}{\scalebox{0.9}{$\scriptscriptstyle N$}}}

\makeatletter
\newcommand{\Biggg}{\bBigg@{2.5}}
\newcommand{\vast}{\bBigg@{3}}
\newcommand{\Vast}{\bBigg@{3.5}}
\newcommand{\massive}{\bBigg@{4.5}}
\newcommand{\Massive}{\bBigg@{6}}
\makeatother

\newcommand{\nquad}{\kern-1em}
\newcommand{\thickhline}{%
	\noalign{\ifnum 0=`}\fi \hrule height 1pt
	\futurelet \reserved@a \@xhline}
\newcolumntype{"}{@{\hskip\tabcolsep\vrule width 1pt\hskip\tabcolsep}}
\newcommand{\PreserveBackslash}[1]{\let\temp=\\#1\let\\=\temp}
\newcolumntype{C}[1]{>{\PreserveBackslash\centering}p{#1}}
\newcolumntype{R}[1]{>{\PreserveBackslash\raggedleft}p{#1}}
\newcolumntype{L}[1]{>{\PreserveBackslash\raggedright}p{#1}}

\DeclareMathAlphabet{\mathpzc}{OT1}{pzc}{m}{it}

\makeatletter

\newcommand{\Rom}[1]{\expandafter\@slowromancap\romannumeral #1@}
\makeatother
\newcommand{\customappendixref}[1]{\hyperref[#1]{\Alph{section}}}

\makeatletter
\let\LN@align\align
\let\LN@endalign\endalign
\renewcommand{\align}{\linenomath\LN@align}
\renewcommand{\endalign}{\LN@endalign\endlinenomath}
\let\LN@gather\gather
\let\LN@endgather\endgather
\renewcommand{\gather}{\linenomath\LN@gather}
\renewcommand{\endgather}{\LN@endgather\endlinenomath}
\makeatother

\title{\textbf{Dynamic Scheduling of a Parallel-Server Queueing System: A Computational Method for High-Dimensional Problems}}
\author{Barış Ata\thanks{\texttt{baris.ata@chicagobooth.edu}} \ and Ebru Kaşıkaralar\thanks{\texttt{ebrukasikaralar@gmail.com}}}
\date{{\small \vspace{-10mm} \today}}
\begin{document}

\vspace*{-25mm}
{\let\newpage\relax\maketitle}

%\linenumbers %Uncomment to add line numbers

\begin{abstract}
\textbf{Problem definition:} A key operational challenge for call centers is to decide, in real time, which waiting customer should be served by which available agent. This is known as skill-based routing, and the decision becomes especially difficult in large systems with many customer classes, where standard dynamic programming methods can be computationally intractable.

\textbf{Methodology/results:} Focusing on the Halfin--Whitt heavy-traffic regime and an infinite-horizon discounted cost criterion, we develop a computational method that scales to high-dimensional settings with many customer classes. Our approach begins by deriving an approximating diffusion control problem in the heavy traffic limiting regime. Building on earlier work by \citet{han2018solving}, we develop a simulation-based method to solve this problem, relying heavily on deep neural network techniques. Using this framework, we construct a policy for the original (prelimit) call center scheduling problem. To evaluate performance, we adopt a data-driven approach. Using call center data from a large U.S. bank, we calibrate the model and construct realistic test instances. We then compare the resulting policy with benchmark policies drawn from the literature. Across all test problems considered so far, our policy performs at least as well as or better than the best benchmark identified. Moreover, the method remains computationally feasible in dimensions up to 100, corresponding to call centers with 100 or more distinct customer classes.

\textbf{Managerial implications:} The proposed approach gives managers a scalable way to improve real-time routing decisions in complex service systems. It can help call centers improve service quality and evaluate routing policies before implementing them. More broadly, the paper provides decision makers with a computational framework for studying and managing high-dimensional service systems that are too large for standard exact methods.
\end{abstract}

\pagenumbering{arabic}

\setlength{\abovedisplayskip}{8pt}
\setlength{\belowdisplayskip}{8pt}
\setlength{\abovedisplayshortskip}{8pt}
\setlength{\belowdisplayshortskip}{8pt}
\vspace{-7mm}
% ============================================================
%  SECTION 1: INTRODUCTION
% ============================================================
\section{Introduction}\label{sect:intro}
 Motivated by the skill-based routing problem commonly encountered in call centers, this paper considers the dynamic control of a parallel-server queueing system. Focusing attention on the Halfin--Whitt heavy traffic asymptotic regime and on the infinite-horizon discounted cost criterion, we develop an effective computational method that scales to high-dimensional settings involving many customer classes. To evaluate its performance, we adopt a data-driven approach: using call center data from a large U.S. bank, we calibrate our model and construct realistic test problems. (Our data set is provided by the Service Enterprise Engineering (SEE) Lab
at the Technion, and it is publicly available at the
\href{https://see-center.iem.technion.ac.il/databases/USBank/}{SEE Center U.S. Bank database}.
Accessed on January 17, 2024.) We then compare the policy generated by our method against benchmark policies drawn from the literature.

Modeling and analysis of call center operations have been a vibrant area of research in both operations management and applied probability; see, for example, the surveys by \citet{gans2003telephone}, \citet{aksin2007modern} and \citet{koole2023practice}. A prominent stream within this literature approximates discrete-flow queueing control problems by Brownian control problems under heavy-traffic assumptions, following the seminal work of \citet{halfin1981heavy}; see \citet{harrison2004dynamic} and \citet{atar2004scheduling} for early applications of this approach. Proceeding in the same spirit, we formulate an approximating Brownian control problem in the Halfin--Whitt asymptotic regime for our parallel-server system.

The approximating Brownian control problem we derive is a drift-rate control problem. To compute its solution, we analyze the associated Hamilton-Jacobi-Bellman (HJB) equation, which, in our setting, takes the form of a semilinear partial differential equation (PDE); see \citet{fleming2006controlled}. To numerically solve this PDE, we adapt the deep learning-based approach developed by \citet{han2018solving} for semilinear PDEs. Crucially, we tailor their method to our control problem in three ways: we derive a stochastic identity whose structure follows directly from the HJB equation of our problem, we use domain knowledge to generate 
sample paths that are representative of the system's behavior, and we 
impose shape constraints on the value function and its gradient to improve solution quality.

To be specific, we first derive a stochastic equation that is equivalent to the HJB equation in the sense that a function solves the HJB equation if and only if it satisfies the stochastic equation. We then propose neural network approximations for the value function and its gradient and formulate a loss function based on the resulting discrepancy in the stochastic equation. The training procedure involves minimizing this loss over the neural network parameters, yielding an approximate solution to the HJB equation. This solution is used to construct a policy for the original discrete-flow control problem; see Section \ref{sect:computational_method}.

The primary contribution of this paper is to develop an effective computational method for solving dynamic scheduling problems in high-dimensional parallel-server queueing systems. We evaluate our proposed policy against benchmarks on five test problems. The U.S. Bank call center data naturally leads to a 13-dimensional test problem, but this problem is computationally intractable because the associated Markov decision process (MDP) suffers from the curse of dimensionality. We also consider a variant of this problem, as well as a 100-dimensional test problem. For these problems, the optimal policy is not available. Therefore, we also consider two low-dimensional test problems, whose optimal policy can be computed via standard MDP techniques. Our computational method yields policies whose performance is on par with the optimal policy for these test problems. Similarly, the policies derived via our method outperform the best available benchmark for the high-dimensional test problems. Our computational results also yield useful insights regarding the so-called basic versus nonbasic activities as well as the joint work conservation property; see Section \ref{sect:concluding_remarks}.

The rest of the paper is structured as follows: Section \ref{sect:literature} reviews the related literature. Section \ref{sect:model} introduces the underlying queueing model. Section \ref{sect:brownian_control} derives the diffusion approximation in the Halfin--Whitt regime. Section \ref{sect:equivalent_characterization} derives the key stochastic identity that motivates the loss function for our computational method. Section \ref{sect:computational_method} describes our computational method. Section \ref{sect:data} reviews the U.S. Bank call center data, introduces our test problems and the benchmark policies we use to assess the effectiveness of our method. Section \ref{sect:results} presents the computational results, comparing the performance of the proposed policy against that of the benchmark policies. Section \ref{sect:concluding_remarks}  discusses the role of basic and nonbasic activities and the joint work conservation property. The appendix contains supplementary data tables, details of the offline approximation used to evaluate the auxiliary function $F(x,v)$ and the derivation of the proposed prelimit policy. Finally, an online supplement provides a worked example of the approach we used to construct the 100-dimensional test problem, the MDP formulation and solution approach for the low-dimensional problems, and the implementation details of our method.
\vspace{-5mm}

% ============================================================
%  SECTION 2: LITERATURE REVIEW
% ============================================================
\section{Literature Review}\label{sect:literature}
As noted above, call centers have been studied extensively over the past three decades, driven by the operational complexity of managing large-scale service systems. One key feature of many modern call centers is the presence of caller and agent heterogeneity. Callers are often categorized into distinct classes based on their service needs, while agents are grouped according to their skill sets to serve different types of customers. This many-to-many mapping between customers and agents gives rise to the problem of skill-based routing, where the objective is to dynamically match incoming calls to the qualified agents in real time. 

Skill-based routing has been studied using both applied and theoretical methodologies. On the applied side, simulation is a common approach for evaluating routing policies.  For example, \citet{koole2003routing} analyze the performance of various routing heuristics in both blocking and delay systems using simulation. Similarly, \citet{mehrotra2012routing} conduct a simulation study using real-world customer service data to assess the effectiveness of different routing schemes.

On the theoretical side, much of the literature is grounded in asymptotic analysis. Two asymptotic regimes are commonly considered: the conventional heavy-traffic regime and the many-server Halfin–Whitt regime. In the conventional heavy-traffic regime, \citet{harrison1999heavy} is one of the first skill-based routing models that capture the parallel-server structure inherent in skill-based routing. They propose a discrete-review scheduling policy and conjecture its asymptotic optimality, which is later proved by \citet{ata2005heavy}. In a related stream of research, \citet{bell2001dynamic} study a two-class, two-server $N$-network model and show that a threshold-based continuous review model is asymptotically optimal; also see \citet{pesic2016dynamic}. \citet{ghamami2013dynamic} later extend this $N$-network model to incorporate customer abandonments and show that a two-threshold control policy is asymptotically optimal; also see \citet{rubino2009dynamic}, who study a more general parallel-server network model with abandonments in the heavy-traffic limit.

In the Halfin--Whitt regime, \citet{armony2005dynamic} is one of the first to study skill-based routing in call centers. Armony considers an inverted-V model with a single customer class served by multiple server pools and proposes a Fastest-Server-First (FSF) routing policy, showing that it asymptotically minimizes the steady-state queue length and virtual waiting time. \textcolor{blue}{Atar} (\citeyear{atar2005diffusion},\,\citeyear{atar2005scheduling}) studies parallel-server systems with multiple customer classes and multiple server pools. Focusing on tree-like networks, Atar establishes a joint work conservation property that holds asymptotically in the Halfin--Whitt regime. He also derives asymptotically optimal policies that assume knowledge of the value function and its gradient for the associated control problem. Our work complements Atar's by providing a method to compute the value function and its gradient. \citet{gurvich2009queue} introduce the Queue-and-Idleness-Ratio (QIR) control policy, which dynamically routes arrivals to the server pool with the greatest idleness imbalance and assigns idle servers to the class with the largest queue imbalance. In a companion paper, \citet{gurvich2009scheduling} prove that the QIR policy is asymptotically optimal under convex holding costs and pool-dependent service rates. A recent work by \citet{zhao2024hierarchical} studies the global stability of multiclass queueing networks under the family of QIR policies. In related work, \citet{tezcan2010dynamic} study the $N$-system under the Halfin–Whitt regime with linear holding costs and pool-dependent service rates. They propose a static priority rule in which each idle server selects the queue with the higher holding cost, and each arriving customer is routed to the faster available server. The authors establish that this policy is asymptotically optimal.

The heavy-traffic approach to solving dynamic control problems such as skill-based routing typically involves deriving a limiting Brownian control problem (BCP) and using the solution to the associated HJB equation to propose policies for the original (prelimit) system; see, for example, \textcolor{blue}{Atar} (\citeyear{atar2005diffusion},\,\citeyear{atar2005scheduling}). In this work, we formally derive the BCP for a general parallel-server model and solve the resulting HJB equation using neural network approximations. An important antecedent of our paper is  \citet{atar2005diffusion} who considers a similar parallel-server model under additional restrictions. Crucially, Atar assumes all activities are basic, and under certain conditions, the author restricts attention to policies that are jointly work conserving in the heavy traffic limit. This simplifies his analysis. \textcolor{blue}{Atar} (\citeyear{atar2005diffusion},\,\citeyear{atar2005scheduling}) jointly establish that under suitable conditions the HJB equation has a smooth solution and, using that solution, the author proposes a control policy for a parallel-server queueing system that is asymptotically optimal.  However, implementing this policy in practice requires solving the HJB equation, which is computationally intractable in high dimensions using traditional methods such as finite-element methods. Consequently, prior literature has been limited to low-dimensional examples; for instance, \citet{harrison2004dynamic} study a two-dimensional system, while \citet{kumar2004numerical} and \citet{ata2020dynamic} explore similarly tractable settings. This is where our contribution lies. Namely, we develop a computational method for general parallel-server systems that is effective in high dimensions. We then use the solution to the high-dimensional HJB equation to propose a scheduling policy for the original parallel-server queueing model. While we focus on call center applications, we expect our method to have broader applicability in improving the performance of other service systems, including healthcare delivery operations.

As mentioned above, our computational method builds on the seminal work of \citet{han2018solving}, who develop a neural network-based method to solve semilinear parabolic PDEs. Their work is part of a growing literature that uses deep learning to solve high-dimensional PDEs. Two major methodological frameworks that have emerged in this area are Physics-Informed Neural Networks (PINNs) and methods based on Backward Stochastic Differential Equations (BSDEs).

PINNs, introduced by \citet{raissi2019physics}, are mesh-free methods that are designed to learn both from the training data and the underlying physical laws governing the data, which are often represented by PDEs. These methods approximate the solution, its gradients, and Hessians using automatic differentiation. A single neural network is trained to satisfy the differential operator as well as initial and boundary conditions by minimizing a loss function evaluated at randomly sampled points within a domain. One key challenge of this methodology lies in the high computational cost of computing Hessians through automatic differentiation, particularly in high-dimensional settings. Nevertheless, several techniques have been proposed to accelerate these computations; see, for example, \citet{he2023learning} and \citet{hu2024hutchinson}.

The other class of methods is based on reformulating PDEs as BSDEs, a framework first developed by \citet{pardoux1990adapted}. The Deep BSDE method, pioneered by \citet{han2018solving} for solving semilinear PDEs, approximates the value function $V$ at the initial condition and its gradient $\nabla_{x}V$ at each time step by minimizing a global loss function that strives to match the simulated value function $V$ at the terminal time $T$ with a terminal condition. The solution is represented by a feedforward neural network at each time step, with parameters optimized as part of this global minimization problem. For reviews of this literature, see \citet{weinan2021algorithms}, \cite{beck2023overview}, \citet{chessari2023numerical} and \citet{han2025brief}.

Building on this literature, recent work applies neural network-based PDE solvers to dynamic control problems in operations management. Within this literature, \citet{ata2025dynamic} consider a scheduling problem for a multiclass queueing system with a single-server pool under a finite-horizon total cost criterion. The parallel-server queueing model we consider in this paper represents a significant generalization, both methodologically and practically. Methodologically, the parallel-server model necessitates incorporating nonbasic activities in general; see Section \ref{sect:brownian_control} below for a definition of a nonbasic activity. Also for the single-server pool system studied in \citet{ata2025dynamic}, one can ensure all servers are busy whenever there are more jobs than servers in the system. In addition to efficiently utilizing the servers, this also simplifies the analysis. The analogue of this property for parallel-server systems is called the joint work conservation, whereby one ensures no server in the entire system is idle if the total number of jobs in the system exceeds the total number of servers; see \citet{atar2005diffusion}. Unfortunately, this property cannot be ensured for general parallel-server systems. Thus, we do not restrict attention to policies that satisfy joint work conservation. Indeed, we allow for non-work-conserving policies. The parallel-server system also presents additional challenges, because the loss function used for neural network training is defined through an optimization problem. This leads to a significantly more involved training process than that of \citet{ata2025dynamic}; see Section \ref{sect:computational_method} for details. For other applications of similar neural network-based methods in operations management, see \citet{ata2024drift}, \citet{ata2024singular}, \citet{ata2025maketoorder}, \citet{ata2025analysis}, \citet{ata2025matching} and \citet{ata2025computational}.

Finally, we show the effectiveness of our proposed neural network-based policy through a simulation study, comparing its performance to the best available benchmarks. For the low-dimensional problems, we use the optimal policy derived from the MDP solution as the benchmark. For higher-dimensional problems where solving the MDP is computationally infeasible, we compare our policy against benchmark policies drawn from the existing literature. These include the classical $c\mu$ rule proposed by \citet{cox1961queues}, which prioritizes classes based on the product of holding cost and service rate; the generalized $c\mu$ rule for the parallel-server systems of \citet{mandelbaum2004scheduling}, which assigns servers based on marginal cost reductions under convex delay costs (also see \citet{van1995dynamic}, \citet{ata2013scheduling}); the $c\mu/\theta$ rule introduced by \citet{atar2010cmu}, which incorporates abandonment rates and the FSF policy studied by \citet{armony2005dynamic}. While these benchmarks provide useful points of comparison, it is important to note that they are not known to be optimal for general parallel-server systems, which is the setting we focus on in this work.
\vspace{-3mm}
% ============================================================
%  SECTION 3: MODEL
% ============================================================
\section{Model}\label{sect:model}
We consider a parallel-server queueing model of a telephone call center with $K$ customer classes and $J$ service stations, where each station consists of many servers with identical capabilities. Class $k$ callers arrive to the system according to a Poisson process with rate $\lambda_{k} > 0$. Each caller requires only one service before they exit the system, and only a certain subset of service stations can serve each class. The classes are indexed by $1, \ldots, K$ and the service stations are indexed by $1, \ldots, J$. We let $\mathcal{K}$ denote the set of classes and let $\mathcal{J}$ denote the set of service stations, i.e., $\mathcal{K} = \{1,\ldots,K\}$ and $\mathcal{J} = \{1,\ldots,J\}$.

Additionally, we define the bipartite graph $\mathcal{G} = (\mathcal{V}, \mathcal{E})$ with set of vertices $\mathcal{V} = \mathcal{K} \cup \mathcal{J}$ and set of edges $\mathcal{E}$ that are defined as follows: An edge exists between vertices $k \in \mathcal{K}$ and $j \in \mathcal{J}$ if service station $j$ can serve class $k$ callers. We denote that edge by $(k,j)$, and also refer to it as \textit{activity} $(k,j)$. To facilitate the analysis, we let $\mathcal{K}(j)$ denote the set of customer classes that can be served by station $j \in \mathcal{J}$, i.e., $\mathcal{K}(j) = \{k \in \mathcal{K}: (k,j) \in \mathcal{E}\}$. Similarly, we let $\mathcal{J}(k) = \{j \in \mathcal{J}: (k,j) \in \mathcal{E}\}$, i.e., the set of service stations that can serve callers from class $k \in \mathcal{K}$.

Service times for activity $(k,j) \in \mathcal{E}$ form an i.i.d. sequence of exponential random variables with mean $1/\mu_{kj} > 0$; $\mu_{kj}$ is the corresponding service rate. Also, callers can abandon while waiting for service. The abandonment times of class $k$ callers are modeled as i.i.d. exponential random variables with mean $1/\theta_{k}$, where $\theta_{k}$ is the corresponding abandonment rate. We assume the service times, abandonment times, and the arrival processes are mutually independent. Figure \ref{figure_model} provides a schematic description of the model for $K = 4$ and $J = 3$.
\begin{figure}[!htb]
    \centering
    \includegraphics[scale=0.18]{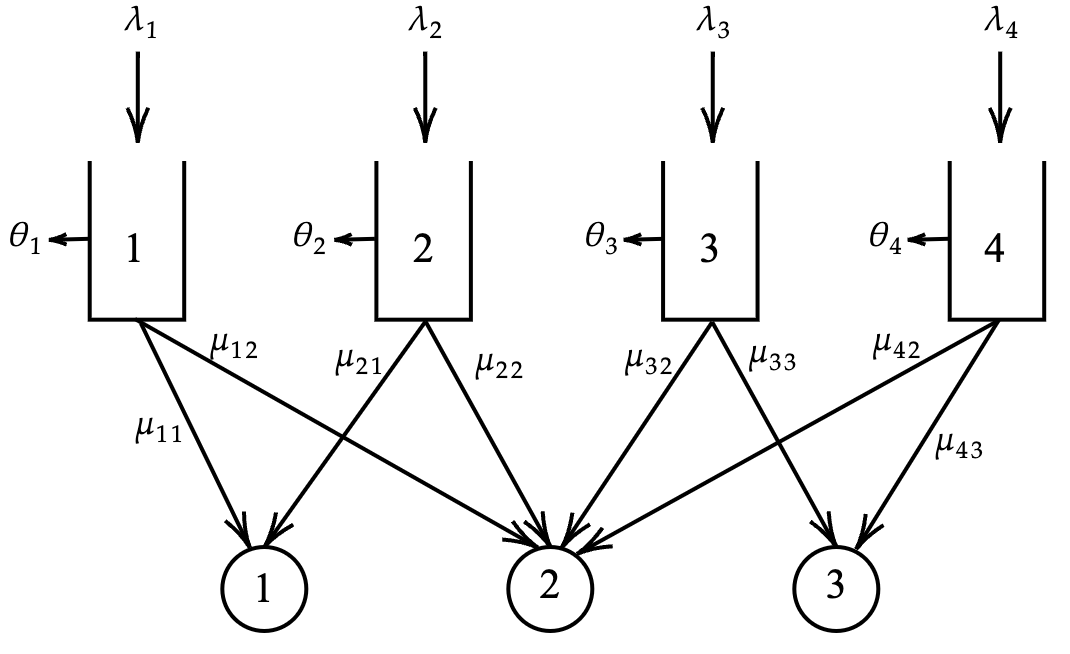}
    \caption{A schematic model of a parallel-server system.}
    \label{figure_model}
\end{figure}

\vspace{-5mm}
The system state is denoted by $X(t) = \left(X_{1}(t), \ldots, X_{K}(t)\right)$, where $X_{k}(t)$ denotes the number of class $k$ callers in the system at time $t$. For $j \in \mathcal{J}$, we let $N_{j}$ denote the number of agents working at service station $j$. The system manager assigns servers to callers in various classes dynamically over time by choosing how many callers to serve using each activity $(k,j) \in \mathcal{E}$. For simplicity, we assume services can be interrupted at any time and resumed later without efficiency loss. Thus, the system manager's control is an $|\mathcal{E}|$-dimensional process, denoted by $\psi = \{\psi_{kj}(t): (k,j) \in \mathcal{E},\,\, t \geq 0\}$, where $\psi_{kj}(t)$ denotes the number of class $k$ jobs assigned to service station $j$ at time $t \geq 0$. Given system state $X(t) = x$, the action $\psi(t)$ at time $t$ must belong to the set $A(x)$ of feasible actions, where 
$$A(x) = \{a \in \mathbb{R}_{+}^{|\mathcal{E}|}: \sum_{j \in \mathcal{J}(k)} a_{kj} \leq x_{k} \,\, \text{for }k \in \mathcal{K}, \text{and}  \sum_{k \in \mathcal{K}(j)} a_{kj} \leq N_{j} \,\, \text{ for } j \in \mathcal{J}\}.$$
The first constraint, $\sum_{j \in \mathcal{J}(k)} a_{kj} \leq x_{k}$, ensures that the number of class $k$ callers in service does not exceed that in the system ($k \in \mathcal{K}$). Similarly, the second constraint, $\sum_{k \in \mathcal{K}(j)} a_{kj} \leq N_{j}$, states that the number of busy servers at station $j$ cannot exceed the total number of servers $N_{j}$ at station $j$ for $j \in \mathcal{J}$. 

In addition, we use $Y_{k}(t)$ to denote the number of class $k$ callers waiting in the queue and $Z_{j}(t)$ to denote the number of idle servers at station $j$, both at time $t$. These processes jointly satisfy the following: For $k \in \mathcal{K},\, j \in \mathcal{J}$ and $t \geq 0$: 
\begin{align}
    &Y_{k}(t) + \sum_{j \in \mathcal{J}(k)} \psi_{kj}(t) = X_{k}(t), \label{eq:prelimit_relation1}\\
    &Z_{j}(t) + \sum_{k \in \mathcal{K}(j)} \psi_{kj}(t) = N_{j}, \label{eq:prelimit_relation2}\\
    & Y_{k}(t) \geq 0, \quad \, Z_{j}(t) \geq 0, \quad X_{k}(t) \geq 0, \quad \psi_{kj}(t) \geq 0. \label{eq:prelimit_relation3}
\end{align}
Next, we define the arrival, service, and abandonment processes formally. Denoting the cumulative number of class $k$ arrivals until time $t$ by $A_{k}(t)$, we set
\begin{align}
    A_{k}(t) = N_{k}^{a}\left(\lambda_{k}\,t\right), \quad k \in \mathcal{K}, \quad t \geq 0,
\end{align}
where $N_{k}^{a}(\cdot)$ is a rate-one Poisson process. Similarly, given a control $\psi$, the cumulative number of class $k$ callers served at station $j$ up to time $t$, denoted by $S_{kj}(t)$, is given as follows: 
\begin{align}
    S_{kj}(t) = N_{kj}^{s}\left(\int_{0}^{t} \mu_{kj} \psi_{kj}(s)ds\right), \quad (k,j) \in \mathcal{E}, \quad t \geq 0,
\end{align}
where $N_{kj}^{s}(\cdot)$ is a rate-one Poisson process. Lastly, we let $R_{k}(t)$ denote the cumulative number of class $k$ callers who abandon until $t$ and model it as follows: 
\begin{align}
    R_{k}(t) = N_{k}^{b}\left(\int_{0}^{t} \theta_k \,Y_{k}(s)ds\right), \quad k \in \mathcal{K}, \quad t \geq 0,
\end{align}
where $N_{k}^{b}(\cdot)$ is a rate-one Poisson process; and $N_{k}^{a}, N_{kj}^{s}$ and $N_{k}^{b}$ for $k \in \mathcal{K}$ and $j \in \mathcal{J}$ are mutually independent. Then one can describe the system dynamics as follows:
\begin{align}
    X_{k}(t) = X_{k}(0) + A_{k}(t) - \sum_{j \in \mathcal{J}(k)} S_{kj}(t) - R_{k}(t), \quad k \in \mathcal{K},\,\, t \geq 0.  \label{eq:state_process}
\end{align}
The economic primitives of our model are the holding and abandonment cost parameters. For class $k$, the holding cost rate is $h_{k}$ per caller per unit of time $(k = 1,\ldots,K)$. Similarly, the abandonment cost is $p_{k}$ per class $k$ caller who abandons. We define the effective cost rate for class $k$, denoted by $c_{k}$, as $c_{k} = h_{k} + \theta_{k}\, p_{k} > 0,$ for $k \in \mathcal{K}$.
%\begin{equation}
%c_{k} = h_{k} + \theta_{k}\, p_{k} > 0, \quad k \in \mathcal{K}.   \label{eqn:defn:cost:param}  
%\end{equation}

Given a control $\psi = \{\psi(t): t\geq 0\}$, and the resulting state, queue-length, and idleness processes, $X(t)$, $Y(t)$ and $Z(t)$, respectively, the instantaneous cost rate is $c\cdot Y(t)$. Thus, conditional on $X(0) = x$, the expected present value of the total costs under control $\psi$, denoted by $J(x;\psi)$, is given as follows:
\vspace{-1mm}
\begin{equation}
    J(x;\psi) = \mathbb{E}_{x}^{\psi}\left\{\int_{0}^{\infty} e^{-\alpha s}\, c \cdot Y(s)ds \right\}, \label{eqn_cost_of_policy}
\end{equation}
where $\alpha > 0$ is the interest rate for discounting and $\mathbb{E}_{x}^{\psi}$ denotes the conditional expectation starting in state $x$ under policy $\psi$. The problem described in this section is an MDP. Although standard dynamic programming techniques can be used to solve it, this approach becomes computationally infeasible for problems with high-dimensional state vectors due to the curse of dimensionality. To overcome this challenge, we adopt a novel computational approach based on deep neural networks. Specifically, we first derive a diffusion approximation of the original control problem and then study that problem formally using deep neural network approximations.
\vspace{-4mm}
% ============================================================
%  SECTION 4: BROWNIAN CONTROL PROBLEM
% ============================================================
\section{The approximating Brownian control problem}\label{sect:brownian_control}
To derive a tractable approximation to the original control problem described in Section \ref{sect:model}, we consider a sequence of systems indexed by $r = 1,2,\ldots,$ each having the structure described in Section \ref{sect:model}. A superscript of $r$ is attached to various stochastic processes to emphasize their dependence on it. We assume that the arrival, service, and abandonment rates vary with $r$ as follows: For $k \in \mathcal{K}$ and $(k,j) \in \mathcal{E}$: 
\vspace{-5mm}
\begin{align}
    &\lambda_{k}^{r} = r \lambda_{k} + \sqrt{r} \zeta_{k} + o(\sqrt{r}), \quad \mu_{kj}^{r} = \mu_{kj} \,\, \text{ and } \,\, \theta_{k}^{r} = \theta_{k}, \label{eq:def:scaled_lambda}
\end{align}
where $\zeta_{k}$ is a given constant (that can be estimated from the data). Similarly, the number of agents varies with $r$ as  $N_{j}^{r} = r \nu_{j}$ for $j \in \mathcal{J}$.
%\begin{equation}
%    N_{j}^{r} = r \nu_{j}, \quad j \in \mathcal{J}. \label{eqn:def:scaled_agent_no} 
%\end{equation}
Our approximation assumes a critically loaded sequence of systems. To describe such a sequence of systems, we introduce a static planning problem {(cf. \citet{harrison1999heavy}, \citet{harrison2000brownian}, \citet{atar2005scheduling}): Choose $\rho$ and $\xi = \{\xi_{kj}: (k,j) \in \mathcal{E}\}$ so as to
\vspace{-5mm}
\begin{align}
&\text{Minimize } \quad \rho  \label{lp_obj}\\
&\text{subject to} \nonumber\\
&\sum_{j \in \mathcal{J}(k)} \nu_{j}\mu_{kj}\xi_{kj} = \lambda_{k}, \quad  k \in \mathcal{K}, \label{lp_constraint1}\\
&\sum_{k \in \mathcal{K}(j)} \xi_{kj} \leq \rho, \quad \quad \quad \,\,\,\, j \in \mathcal{J}, \label{lp_constraint2}\\
&\xi_{kj} \geq 0, \quad \quad \quad \quad \quad \quad \,\, (k,j) \in \mathcal{E}. \label{lp_constraint3}
\end{align}
Here, $\xi_{kj}$ denotes the nominal fraction of server station $j$'s service capacity that is allocated to class $k$ in the long run, $(k,j) \in \mathcal{E}$. Crucially, we assume the system is a balanced, high-volume system, and that the model primitives satisfy the following assumption.\\
\\
\textbf{Heavy Traffic Assumption.} There is a unique optimal solution $(\xi^{*}, \rho^{*})$ to the static planning problem that satisfies $\rho^{*} = 1$ and 
\vspace{-5mm}
\begin{equation}
    \sum_{k \in \mathcal{K}(j)} \xi_{kj}^{*} = 1, \quad j \in \mathcal{J}. \label{assumption:heavy_traffic}
\end{equation}
Under the foregoing assumptions, on the fluid scale, we define the nominal number of class $k$ callers in the system as follows:
\vspace{-5mm}
\begin{equation}
    x_{k}^{*} = \sum_{j \in \mathcal{J}(k)} \xi_{kj}^{*} \nu_{j}, \quad k \in \mathcal{K}. \label{eq:nominal_no_system}
\end{equation}
Similarly, the nominal number of class $k$ callers at service station $j$ is given as 
\begin{equation}
    \psi_{kj}^{*} = \xi_{kj}^{*} \nu_{j}, \quad k \in \mathcal{K}, \, j \in \mathcal{J}. \label{eq:nominal_no_service} 
\end{equation}

Following \citet{harrison1999heavy}, we call an activity $(k,j) \in \mathcal{E}$ \textit{basic} if $\xi_{kj}^{*} > 0$, and \textit{nonbasic} if $\xi_{kj}^{*} = 0$. To facilitate the analysis, we partition the set $\mathcal{E}$ of activities into two sets $\mathscr{B}$ and $\mathscr{N}$ (mnemonic for basic and nonbasic, respectively), where $\mathscr{B} = \{(k,j) \in \mathcal{E}: \xi_{kj}^{*} > 0 \}$ and $\mathscr{N} = \{(k,j) \in \mathcal{E}: \xi_{kj}^{*} = 0 \}$.

Next, we introduce the scaled state, control, queue length, and idleness processes $\hat{X}^{r}, \hat{\psi}^{r}, \hat{Y}^{r}, \hat{Z}^{r}$, respectively, as follows: For $t \geq 0$ and $r \geq 1$, let
\begin{align}
    \hat{X}_{k}^{r}(t) 
    &= \frac{X_{k}^{r}(t) - r x_{k}^{*}}{\sqrt{r}},
    \quad \quad k \in \mathcal{K}, \label{eq:scaled_state}\\
    \hat{\psi}_{kj}^{r}(t) 
    &= \frac{\psi_{kj}^{r}(t) - r \psi_{kj}^{*}}{\sqrt{r}},
    \quad \, (k,j) \in \mathcal{E}, \label{eq:scaled_control}\\
    \hat{Y}_{k}^{r}(t) 
    &= \frac{Y_{k}^{r}(t)}{\sqrt{r}},
    \quad k \in \mathcal{K},
    \qquad
    \hat{Z}_{j}^{r}(t) 
    = \frac{Z_{j}^{r}(t)}{\sqrt{r}},
    \quad j \in \mathcal{J}. \label{eq:scaled_queue_idleness}
\end{align}
Using the scaled processes, one can derive \eqref{eq:const:scaled1}--\eqref{eq:const:scaled3} below from Equations \eqref{eq:prelimit_relation1}--\eqref{eq:prelimit_relation3}. For $t \geq 0$, $k \in \mathcal{K}$ and $j \in \mathcal{J}$, we have that
\begin{align}
    &\hat{Y}_{k}^{r}(t) + \sum_{j \in \mathcal{J}(k)} \hat{\psi}_{kj}^{r}(t) = \hat{X}_{k}^{r}(t), \label{eq:const:scaled1}\\
    &\hat{Z}_{j}^{r}(t) + \sum_{k \in \mathcal{K}(j)} \hat{\psi}_{kj}^{r}(t) = 0, \label{eq:const:scaled2} \\
    &\hat{Y}_{k}^{r}(t) \geq 0,\quad \hat{Z}_{j}^{r}(t) \geq 0. \label{eq:const:scaled3}
\end{align}
Equations \eqref{eq:const:scaled1}--\eqref{eq:const:scaled3} imply that the scaled control $\hat{\psi}^{r}$ satisfies the following for $k \in \mathcal{K}$, $j \in \mathcal{J}$ and $t \geq 0$:
\begin{align}
    &\sum_{j \in \mathcal{J}(k)} \hat{\psi}^{r}_{kj}(t) \leq \hat{X}^{r}_{k}(t), \label{feasible_constraint1}\\
    &\sum_{k \in \mathcal{K}(j)} \hat{\psi}^{r}_{kj}(t) \leq 0. \label{feasible_constraint2}
\end{align}
In what follows, we restrict attention to control policies that satisfy the following for $t \geq 0$ (see Equation (\ref{eq:scaled_control}) and note from the sets $\mathscr{N}$ and $\mathscr{B}$ that $\psi_{kj}^{*}$ = 0 for $(k,j) \in \mathscr{N}$ and $\psi_{kj}^{*} > 0$ for $(k,j) \in \mathscr{B}$): 
\begin{align}
    &\psi_{kj}^{r}(t) = r\psi^{*}_{kj} + \sqrt{r}\hat{\psi}^{r}_{kj}(t) + o(\sqrt{r}), \quad  (k,j) \in \mathscr{B}, \label{eq:service_basic1} \\
    &\psi_{kj}^{r}(t) = \sqrt{r}\hat{\psi}^{r}_{kj}(t), \quad \quad \quad \quad \quad \quad \quad \quad \, (k,j) \in \mathscr{N}, \label{eq:service_nonbasic1} \\ 
    &\hat{\psi}^{r}_{kj}(t) \in \mathbb{R}, \quad \quad \quad \quad \quad \quad \quad \quad \quad \quad \quad \,\,\, (k,j) \in \mathscr{B},  \label{eq:service_basic}\\
    &\hat{\psi}^{r}_{kj}(t) \geq 0, \quad \quad \quad \quad \quad \quad \quad \quad \quad \quad \quad \,\,\,\, (k,j) \in \mathscr{N}.
    \label{eq:service_non_basic}
\end{align}

For such policies, we now derive the infinitesimal drift and covariance of the scaled state process $\hat{X}^{r}$ to facilitate the formal derivation of its diffusion limit. In particular, for $t \geq 0$, $k \in \mathcal{K}$, $l\neq k$ and small $h > 0$, the following holds:
\begin{align}  &\mathbb{E}\left[\hat{X}_{k}^{r}(t + h) - x_{k} \mid \hat{X}^{r}(t) = x\right] = \left[\zeta_{k} - \sum_{j \in \mathcal{J}(k)} \mu_{kj} \hat{\psi}_{kj}^{r}(t) - \theta_{k} \hat{Y}_{k}^{r}(t)\right]h + o(h), \label{eq:infinitesimal_drift}\\
&\mathbb{E}\left[(\hat{X}_{k}^{r}(t + h) - x_{k})^{2} \mid \hat{X}^{r}(t) = x \right] = 2\lambda_{k}h + o(h), \label{eq:infinitesimal_variance}\\
& \mathbb{E}\left[(\hat{X}_{k}^{r}(t + h) - x_{k})(\hat{X}_{l}^{r}(t + h) - x_{l}) \mid \hat{X}^{r}(t) = x\right] = o(h). \label{eq:infinitesimal_covariance}
\end{align}

Taking the formal limit as $r \rightarrow \infty$, denoting the weak limit of $(\hat{X}^{r}, \hat{\psi}^{r}, \hat{Y}^{r}, \hat{Z}^{r})$ by $(\hat{X}, \hat{\psi}, \hat{Y}, \hat{Z})$ and using Equation \eqref{eq:const:scaled1}, we deduce from Equations \eqref{eq:infinitesimal_drift}--\eqref{eq:infinitesimal_covariance} that the limiting state process $\hat{X}$ satisfies the following: For $k \in \mathcal{K}$ and $t \geq 0$, 
\begin{equation}
    d \hat{X}_{k}(t) = \left[\zeta_{k} - \theta_{k}\hat{X}_{k}(t) + \sum_{j \in \mathcal{J}(k)} (\theta_{k} - \mu_{kj})\,\hat{\psi}_{kj}(t)\right] dt + \sqrt{2\lambda_{k}}\,dB_{k}(t), \label{eq:state_process_limit}
\end{equation}
where $B(t) = \left(B_{k}(t)\right)$ is a $K$-dimensional standard Brownian motion. 

Also, it follows from Equations \eqref{eq:const:scaled1}--\eqref{eq:const:scaled3} that, the limiting control, queue-length, and idleness process, $\hat{\psi}$, $\hat{Y}$, and $\hat{Z}$, respectively, satisfy the following: For $k \in \mathcal{K}$ and $j \in \mathcal{J}$,
\begin{align}
    &\hat{Y}_{k}(t) = \hat{X}_{k}(t) - \sum_{j \in \mathcal{J}(k)}\hat{\psi}_{kj}(t), \quad t \geq 0,\\
    &\hat{Z}_{j}(t) = - \sum_{k \in \mathcal{K}(j)} \hat{\psi}_{kj}(t), \qquad t \geq 0,\\
    &\hat{Y}(t) \geq 0, \quad \hat{Z}(t) \geq 0.
\end{align}
To minimize technical complexity, we restrict attention to stationary Markov controls, i.e., $\hat{\psi}(t) = \psi(\hat{X}(t))$ for $t \geq 0,$
%\begin{align}
%    \hat{\psi}(t) = \psi(\hat{X}(t)) \quad \text{for} \quad t \geq 0,
%\end{align}
where $\psi : \mathbb{R}^{K} \to \mathbb{R}^{|\mathcal{E}|}$ is a measurable function that we refer to as the control, hereafter. For a control $\psi$ to be admissible, it must satisfy certain restrictions. More specifically, we require that $\psi(x) \in \Psi(x)$ for $x \in \mathbb{R}^{K}$,
%\begin{equation*}
%    \psi(x) \in \Psi(x) \quad \text{for} \quad x \in \mathbb{R}^{K},
%\end{equation*}
where the set $\Psi(x)$ is defined as 
\begin{equation}
\label{eq:set_feasible_policies}
\Psi(x) = \left\{ \psi \in \mathbb{R}^{|\mathcal{E}|} : 
\begin{array}{l}
\sum_{j \in \mathcal{J}(k)} \psi_{kj} \leq x_{k},\, \quad \forall k \in \mathcal{K}, \\
\sum_{k \in \mathcal{K}(j)} \psi_{kj} \leq 0, \,\,\,\, \quad \forall j \in \mathcal{J}, \\
\psi_{kj} \geq 0, \quad \quad \quad \quad \,\, (k, j) \in \mathscr{N}. 
\end{array}
\right\}, \quad x \in \mathbb{R}^{K}.
\end{equation}
Note that the restrictions defining the set $\Psi(x)$ follow from Equations \eqref{feasible_constraint1}--\eqref{feasible_constraint2} and \eqref{eq:service_basic}--\eqref{eq:service_non_basic}. 

For notational simplicity, we let $\sigma_{k} = \sqrt{2\lambda_{k}}$ for $k \in \mathcal{K}$. Also, we define the drift function $b: \mathbb{R}^{K} \times \mathbb{R}^{|\mathcal{E}|} \rightarrow \mathbb{R}^{K}$, where 
\begin{align}
    b_{k}\,(x, a) = \zeta_{k} - \sum_{j \in \mathcal{J}(k)} \mu_{kj} a_{kj} - \theta_{k}\Big(x_{k} - \sum_{j \in \mathcal{J}(k)} a_{kj}\Big), \label{eq:drift_function}
\end{align}
for $(x,a) \in \mathbb{R}^{K} \times \mathbb{R}^{|\mathcal{E}|}$ and $k = 1,\ldots,K$. Then letting $\sigma = \operatorname{diag}(\sigma_{1},\ldots,\sigma_{K})$ denote the covariance matrix, the evolution of the state process $\hat{X}$ under an admissible control $\psi$ can be succinctly described as follows:
\vspace{-5mm}
\begin{equation}
    d \hat{X}(t) = b\,(\hat{X}(t),\hat{\psi}(t))\,dt + \sigma dB(t), \quad t \geq 0.
\end{equation}
\vspace{-2mm}
\noindent Additionally, we define the cost-rate function $c:\mathbb{R}^{K} \times \mathbb{R}^{|\mathcal{E}|} \rightarrow \mathbb{R}$ as
\begin{equation}
    c\,(x,a) = \sum_{k \in \mathcal{K}} c_{k}\Big(x_{k} - \sum_{j \in \mathcal{J}(k)} a_{kj}\Big), \label{eq:instantaneous_cost}
\end{equation}
where $x_{k} - \sum_{j \in \mathcal{J}(k)}a_{kj}$ corresponds to (scaled) class $k$ queue length when the system state is $x$ and action $a$ is taken.
Then, given an admissible control $\psi$ and the limiting system state $\hat{X}(t) = x$, the instantaneous cost rate is given as $c\,(x, \psi(x))$. Therefore, the expected present value of the total discounted cost under an admissible policy $\psi$, given the initial state $\hat{X}(0) = x$, denoted by $\hat{J}(x;\psi)$, is given as follows:
\begin{equation}
    \hat{J}(x;\psi) = \mathbb{E}_{x}^{\psi} \left\{\int_{0}^{\infty} e^{- \alpha s}\,c\left(\hat{X}(s), \psi(s)\right)ds\right\},  \label{eqn:j_definition}
\end{equation}
\color{black}
where $\alpha > 0$ is the interest rate and $\mathbb{E}_{x}^{\psi}$ denotes the conditional expectation starting in state $x$ under control $\psi$. We now define the optimal value function as 
\begin{equation}
    V(x) = \inf \hat{J}(x;\psi), \label{eqn:value_function_definition}
\end{equation}
where the infimum is taken over the class of admissible policies. 

Next, we formally write the associated HJB equation to characterize an optimal policy; see \citet{fleming2006controlled}. To this end, we define the differential operator $\mathcal{L}$ and the function $\mathcal{H}$ as follows:
\begin{align}
    &\mathcal{L} = \sum_{k = 1}^{K} \lambda_{k}\frac{\partial^{2}}{\partial x_{k}^{2}} \quad \text{and} \quad \mathcal{H}(x,p) = \inf_{a \in \Psi(x)}\left[b\,(x,a)\cdot p + c\,(x,a)\right], \label{eq:pde_derivation}
\end{align}
where $b\,(x, a)$ is given in Equation \eqref{eq:drift_function}. Then the HJB equation involves finding a sufficiently smooth function $V(x)$ that solves the following PDE: 
\begin{align*}
    &\mathcal{L} V(x) + \mathcal{H}(x, \nabla V(x)) - \alpha V(x)= 0, \quad x \in \mathbb{R}^{K}.
\end{align*}
Specifically, the HJB equation is considered on $\mathbb{R}^{K}$ with a polynomial growth condition. We let $C^{2}(\mathbb{R}^{K})$ denote the class of functions that are twice continuously differentiable over $\mathbb{R}^{K}$ and set
\begin{equation}
    C_{\text{pol}} := \left\{f \in C^{2}\left(\mathbb{R}^{K}\right): \exists \,\alpha, \beta > 0\, \text{ such that } \, |f(x)| \leq \alpha \left(1+|x|^{\beta}\right) \,\, \forall x \in \mathbb{R}^{K}\right\}.
\end{equation}
That is, $C_{\text{pol}}$ denotes the class of $C^2$ functions with polynomial growth. One seeks a smooth solution to the HJB equation that belongs to this class. Questions of existence and uniqueness for such solutions have been studied rigorously in \citet{atar2005diffusion} under additional assumptions; see also \citet{gilbarg2001elliptic}. These theoretical issues are beyond the scope of this paper. Accordingly, we assume that a solution $V \in C_{\text{pol}}$ to the HJB equation exists and focus instead on its computation in the high-dimensional setting using deep learning. This approach ultimately yields an effective control policy for the parallel-server queueing network; see Section \ref{sect:computational_method}.

To facilitate the analysis that follows, we write the HJB equation in a more explicit form. Combining equations \eqref{eq:drift_function} and \eqref{eq:pde_derivation}, we obtain the following: For $x \in \mathbb{R}^{K}$,
\begin{align}
    \begin{split}
        \sum_{k =1}^{K} \lambda_{k} \frac{\partial^{2}V(x)}{\partial x_{k}^{2}} &+ \sum_{k=1}^{K} c_{k}x_{k} + \sum_{k =1}^{K}(\zeta_{k} - \theta_{k}x_{k})\frac{\partial V(x)}{\partial x_{k}} \\
        & - \sup_{\psi \in \Psi(x)} \left[\sum_{k = 1}^{K} \sum_{j \in \mathcal{J}(k)} \left(c_{k} + (\mu_{kj} - \theta_{k})\frac{\partial V(x)}{\partial x_{k}}\right)\psi_{kj}\right] - \alpha V(x) = 0. \label{eqn:HJB_PDE}
    \end{split}
\end{align}
%Similarly, to solve for the expected discounted cost $\hat{J}(x;\psi)$ under an arbitrary policy $\psi$, one needs to consider the following PDE: For $x \in \mathbb{R}^{K}$,
%\begin{equation}
%    \sum_{k = 1}^{K}\lambda_{k}\frac{\partial^{2}\hat{J}(x;\psi)}{\partial x_{k}^{2}}  + \sum_{k = 1}^{K} b_{k}(x,\psi)\frac{\partial \hat{J}(x;\psi)}{\partial x_{k}} + c(x,\psi) - \alpha \hat{J}(x;\psi) = 0. \label{eqn:hjb_arbitrary}
%\end{equation}

% ============================================================
%  SECTION 5: EQUIVALENT CHARACTERIZATION
% ============================================================
\section{An equivalent characterization of the value function} \label{sect:equivalent_characterization}
In this section, we prove a key identity, Equation (\ref{eq:key_identity}) below, which serves as an alternative characterization of the value function. Section \ref{sect:computational_method} uses this identity to formulate the loss function of our computational method which relies heavily on neural networks. Our method begins by specifying a \textit{reference policy}, denoted by $\tilde{\psi}$. This is a nominal policy, set at the beginning but subject to modification based on computational insights, which we use to generate the sample paths of the system state. Intuitively, our goal is to choose a reference policy whose paths tend to occupy parts of the state space that we expect the optimal policy to visit frequently.  In our computational study, we consider the following reference policies motivated by the literature: the $c\mu$ rule \citep{cox1961queues}, the $c\mu/\theta$ rule \citep{atar2010cmu}, and the fastest-server-first (FSF) rule \citep{armony2005dynamic}; see Section \ref{sect:computational_method} for details.

The corresponding \textit{reference process}, denoted by $\tilde{X}$, satisfies the following: For $k \in \mathcal{K}$ and $t \geq 0$,
\begin{equation}
    d \tilde{X}_{k}(t) = \Big(\zeta_{k} - \theta_{k}\tilde{X}_{k}(t) + \sum_{j \in \mathcal{J}(k)} (\theta_{k} - \mu_{kj})\,\tilde{\psi}_{kj}(\tilde{X}(t))\Big) dt + \sqrt{2\lambda_{k}}\,dB_{k}(t). \label{eq:training_sde}
\end{equation}
As a preliminary to introducing the key identity, we next define an auxiliary function $F(\cdot,\cdot):\mathbb{R}^{K}\times\mathbb{R}^{K}\rightarrow\mathbb{R}$, where
\begin{equation}
    F(x,v) =  \sum_{k = 1}^{K}\sum_{j \in \mathcal{J}(k)}(\theta_{k} - \mu_{kj})\,\tilde{\psi}_{kj}(x)\,v_{k} - \sum_{k = 1}^{K} c_{k} x_{k} + \sup_{\psi \in \Psi(x)}\Big[\sum_{k = 1}^{K}\sum_{j \in \mathcal{J}(k)}\Big(c_{k} + (\mu_{kj} - \theta_{k})\,v_{k}\Big)\psi_{kj}\Big]. \label{eq:f_function} 
\end{equation}
The following proposition derives the key identity; see \ref{app:proof_proposition} for its proof.

\begin{proposition}\label{prop_key_identity}
If $V(\cdot)$ satisfies the HJB equation (\ref{eqn:HJB_PDE}), then it satisfies the following identity almost surely for any $T > 0$:
\begin{equation}
    e^{-\alpha T}V(\tilde{X}(T)) - V(\tilde{X}(0)) = \int_{0}^{T} e^{-\alpha t}\nabla V(\tilde{X}(t))\,\cdot \sigma dB(t) + \int_{0}^{T} e^{-\alpha t} F\left(\tilde{X}(t), \nabla V(\tilde{X}(t))\right)dt. \label{eq:key_identity}
\end{equation}
\end{proposition}

Proposition \ref{prop_key_identity} motivates the loss function of our computational method (see Section \ref{sect:computational_method}). The following result can be viewed as its converse. Its proof is similar to that of Proposition 3 of \citet{ata2024drift} and requires only minor modifications. As such, it is omitted.

\begin{proposition} \label{prop_key_identity_converse}
Suppose that $V: \mathbb{R}^{K} \rightarrow \mathbb{R}$ is a $C^{2}$ function, $G: \mathbb{R}^{K} \rightarrow \mathbb{R}^{K}$ is continuous and $V, \nabla V,$ and $G$ all have polynomial growth. Also, assume that the following identity holds almost surely for some fixed $T > 0$ and every $\tilde{X}(0) = x \in \mathbb{R}^{K}$:
\begin{equation}
    e^{-\alpha T}V(\tilde{X}(T)) - V(\tilde{X}(0)) = \int_{0}^{T} e^{-\alpha t} G (\tilde{X}(t))\,\cdot\, \sigma \,dB(t) + \int_{0}^{T} e^{-\alpha t} F(\tilde{X}(t), G(\tilde{X}(t)))\,dt. \label{eq:key_identity2}
\end{equation}
Then $G(\cdot) = \nabla V(\cdot)$ and $V$ satisfies the HJB equation (\ref{eqn:HJB_PDE}).
\end{proposition}

Computing $F(\cdot,\cdot)$ exactly requires solving a linear program at each iteration of the computation for each function call, which is time consuming given that $F$ is evaluated repeatedly during neural network training. We therefore train two neural networks offline, $\hat{H}$ and $\hat{D}$, that approximate the supremum and reference-policy terms of $F$, respectively, and combine them into an approximation $\hat{F}(\cdot,\cdot)$; see ~\ref{sect:offline_approximation_auxiliary_function} for details.

\color{black}
\vspace{-5mm}
% ============================================================
%  SECTION 6: COMPUTATIONAL METHOD
% ============================================================
\section{Computational method}\label{sect:computational_method}

We build on the computational method introduced by \citet{han2018solving} for solving semilinear parabolic partial differential equations (PDEs). Their approach defines the loss function using a backward stochastic differential equation (BSDE) representation of the target PDE. Similarly, we formulate our loss function using the identity in Equation~\eqref{eq:key_identity}. To begin, we set a reference policy $\tilde{\psi}$ and then simulate the discretized sample paths of the reference process $\tilde{X}$ based on this reference policy on a fixed and finite time domain $[0,T]$. We then fix a partition $0 = t_{0} < t_{1} < \ldots < t_{N} = T$ of the time horizon $[0,T]$, and simulate the discretized sample paths of the reference process at times $t_{0},t_{1}, \ldots, t_{N}$; see Subroutine \ref{subroutine:euler}. 

% ── Subroutine: Euler discretization ──
\begin{algorithm}[!htb]
    \floatname{algorithm}{Subroutine}
    \caption{Euler discretization scheme.}
    \label{subroutine:euler}
    \begin{algorithmic}[1]
        \Statex \textbf{Input:} The trained neural network $\hat{D}(\cdot)$ 
            from Subroutine~\ref{subroutine:d_approx} in \ref{sect:offline_approximation_auxiliary_function}, the variance term 
            $\sigma^{2}$, the time horizon $T$, the number of intervals $N$, 
            a discretization step-size $\Delta t_{n} \triangleq T/N$, and a 
            random initial state $x_{0} \sim \Gamma_{0}$. The initial distribution $\Gamma_{0} = \text{Uniform}([-10,10]^{K})$ is chosen based on the observed range of the simulated sample paths $\tilde{X}$.
        \Statex \textbf{Output:} Discretized reference process 
            $\tilde{X}(t_{n})$ for $n = 1,\ldots, N$, and the Brownian 
            increments $\Delta B(t_{n})$ for $n = 0,\ldots, N-1$.
        \Function{Discretize}{$T, \Delta t_{n}, x_0$}
            \State Construct the partition 
                $0 = t_{0} < t_1 < \ldots < t_{N} = T$ with 
                $\Delta t_{n} = t_{n+1} - t_{n}$ for $n = 0,\ldots,N-1$.
            \State Generate $N$ i.i.d.\ $K$-dimensional Gaussian random 
                vectors $\Delta B(t_{n}) = (\Delta B_{k}(t_{n}))_{k=1}^K$ 
                with mean zero and covariance matrix $\Delta t_{n}\,I$ for 
                $n = 0,\ldots, N-1$.
            \State Set $\tilde{X}(t_0) \leftarrow x_0$.
            \For{$n = 0, \ldots, N-1$}
                \For{$k = 1, \ldots, K$}
                    \State $\tilde{X}_{k}(t_{n+1}) \leftarrow 
                        \tilde{X}_{k}(t_{n}) + \Big(\zeta_{k} - \theta_{k}\,
                        \tilde{X}_{k}(t_{n}) + \hat{D}_{k}\big(
                        \tilde{X}(t_{n})\big)\Big)\,\Delta t_{n} + 
                        \sigma_{k}\,\Delta B_{k}(t_{n})$
                \EndFor
            \EndFor
            \State \Return $\tilde{X}(t_{n})$ for $n = 1,\ldots, N$ and 
                $\Delta B(t_{n})$ for $n = 0,\ldots, N-1$.
        \EndFunction
    \end{algorithmic}
\end{algorithm}

We approximate the value function $V(\cdot)$ using a deep neural network $V^{\omega}(\cdot)$ with associated parameter vector $\omega$. Similarly, we approximate the gradient function $\nabla_{x}V(\cdot)$ using a deep neural network $G^{\nu}(\cdot)$ with parameter vector $\nu$. We adopt a discretized approximation of the identity \eqref{eq:key_identity} to define our loss function, denoted by $\ell(\omega,\nu)$, as a function of the neural network parameters $(\omega,\nu)$ as follows:
\vspace{-2mm}
\begin{equation}
\begin{aligned}
\ell(\omega, \nu)= &\mathbb{E}\Bigg[\Bigg(e^{-\alpha T}V^{\omega}(\tilde{X}(T)) - V^{\omega}(\tilde{X}(0))-\sum_{n=0}^{N-1} e^{-\alpha  t_{n}} G^{\nu}(\tilde{X}(t_n)) \cdot \sigma\,\Delta B(t_{n}) \\
& \qquad \qquad \qquad \qquad \quad \quad-\sum_{n=0}^{N-1}  e^{-\alpha t_{n}} F\left(\tilde{X}(t_{n}), G^{\nu}(\tilde{X}(t_{n}))\right) \Delta t_{n}\Bigg)^{2}\Bigg],\label{eq:loss_function_first}
\end{aligned}
\end{equation}
where $\Delta t_{n}  = t_{n+1} - t_{n}$ and $\Delta B(t_{n})$ is a Gaussian random vector with zero mean and covariance matrix $\Delta t_{n}I$ for $n = 0, 1, \ldots, N-1$. Here, we approximate the expectation, summing over the sample paths of the reference process $\tilde{X}$. Our method determines the optimal neural network parameters $(\omega^{*}, \nu^{*})$ that minimize the loss function \eqref{eq:loss_function_first} using stochastic gradient descent; see Algorithm \ref{main_algorithm}. Given those parameters, we use the learned gradient function $G^{\nu^{*}}(\cdot)$ to propose a policy for the prelimit system. We describe the proposed policy next.   

% ── Main algorithm ──
\begin{algorithm}[!htb]
    \caption{}
    \label{main_algorithm}
    \begin{algorithmic}[1]
        \Statex \textbf{Input:} The trained networks $\hat{H}(\cdot,\cdot)$ 
            from Subroutine~\ref{subroutine:h_approx} in \ref{sect:offline_approximation_auxiliary_function} and 
            $\hat{D}(\cdot)$ from 
            Subroutine~\ref{subroutine:d_approx} in \ref{sect:offline_approximation_auxiliary_function}, the number of 
            training iterations $M$, a batch size $S$, a learning rate schedule $(\text{lr}_0, \gamma, \text{milestones})$, neural network architecture hyperparameters (number of layers, neurons per layer, activation function), the time horizon $T$, the number of time 
            intervals $N$, a discretization step-size 
            $\Delta t_n \triangleq T/N$, an initial distribution 
            $\Gamma_0$, and an optimization solver 
            (Adam, SGD, RMSProp, etc.).
        \Statex \textbf{Output:} The approximate value function 
            $V^\omega(\cdot)$ and the approximate gradient function 
            $G^\nu(\cdot)$.
        \State Define the approximate auxiliary function:
            \begin{equation*}
                \hat{F}(x, v) = \hat{D}(x) \cdot v 
                - \sum_{k = 1}^{K} c_k\, x_k 
                + \hat{H}(x, v).
            \end{equation*}
        \State Initialize the neural networks $V^\omega(\cdot)$ and 
            $G^\nu(\cdot)$.
        \For{$m = 0, \ldots, M - 1$}
            \State Sample $x_0^{(s)} \sim \Gamma_0$ for 
                $s = 1, \ldots, S$.
            \State Simulate $S$ discretized sample paths by invoking 
                \textsc{Discretize}$(T, \Delta t_n, x_0^{(s)})$ 
                (Subroutine~\ref{subroutine:euler}) to obtain 
                $\{\tilde{X}^{(s)}(t_n), 
                \Delta B^{(s)}(t_n)\}$ for $s = 1, \ldots, S$.
            \State Compute the empirical loss:
                \begin{align*}
                    \ell(\omega, \nu) = \frac{1}{S}\sum_{s=1}^{S}
                    \Bigg(&e^{-\alpha T} V^\omega(\tilde{X}^{(s)}(T)) 
                    - V^\omega(\tilde{X}^{(s)}(0)) 
                    - \sum_{n=0}^{N-1} e^{-\alpha t_n}\, 
                    G^\nu(\tilde{X}^{(s)}(t_n)) \cdot \sigma\,
                    \Delta B^{(s)}(t_n) \\
                    &- \sum_{n=0}^{N-1} e^{-\alpha t_n}\, 
                    \hat{F}\Big(\tilde{X}^{(s)}(t_n),\, 
                    G^\nu(\tilde{X}^{(s)}(t_n))\Big)\, 
                    \Delta t_n \Bigg)^{\!2}.
                \end{align*}
            \State Update $(\omega, \nu)$ by computing 
                $\nabla_{(\omega,\nu)} \ell$ and applying the chosen 
                optimizer.
            \State Update learning rate according to schedule.
        \EndFor
        \State \Return $V^\omega(\cdot)$ and $G^\nu(\cdot)$.
    \end{algorithmic}
\end{algorithm}
\vspace{-5mm}
% ── Proposed policy ──
\paragraph{Proposed policy for the prelimit system.}
Given the trained gradient network $G^{\nu^{*}}(\cdot)$ that approximates the gradient $\nabla V(x)$ of the optimal value function, we propose the following policy for the prelimit system, i.e., the $r^{\text{th}}$ system: Upon observing the system state $X^{r}(t)$ at time $t\geq 0$, choose the server assignments $\psi_{kj}^{r}(t)$ for $(k,j) \in \mathcal{E}$ by solving the following linear program, where $\hat{X}^{r}(t)$ is the scaled system state defined via Equation \eqref{eq:scaled_state}:
\begin{align}
    &\operatorname{Maximize} \sum_{k = 1}^{K} \sum_{j \in \mathcal{J}(k)} \Big(c_k + (\mu_{kj} - \theta_k)\, G_k^{\nu^*}(\hat{X}^{r}(t))\Big)\, \psi_{kj}^r(t) \label{eq:prelimit_obj} \\
    &\text{subject to} \nonumber \\
    &\sum_{j \in \mathcal{J}(k)} \psi_{kj}^r(t) \leq X_k^r(t), \quad k = 1, \ldots, K, \label{eq:prelimit_c1} \\
    &\sum_{k \in \mathcal{K}(j)} \psi_{kj}^r(t) \leq N_j^r, \quad\;\;\; j  = 1,\ldots, J, \label{eq:prelimit_c2} \\
    &\psi_{kj}^r(t) \geq 0, \quad\quad\quad\quad\;\;\;\; (k,j) \in \mathcal{E}. \label{eq:prelimit_c3}
\end{align}
\begin{remark}
This policy is motivated by the supremum term in the HJB equation \eqref{eqn:HJB_PDE} because the maximizer of that term yields the optimal  policy $\hat{\psi}$ for the limiting problem. More specifically, for the Brownian control problem, at each state $x$, one chooses $\hat{\psi}$ so as to 
\begin{equation*}
    \operatorname{Maximize} \sum_{k = 1}^{K} \sum_{j \in \mathcal{J}(k)} \Big(c_k + (\mu_{kj} - \theta_k)\, G_k^{\nu^*}(\hat{X}^{r}(t))\Big)\, \hat{\psi}_{kj} \label{eq:limit_obj} \quad \text{subject to } \eqref{eq:lp1_off_lp}-\eqref{eq:lp4_off_lp}
\end{equation*}
To translate this solution to a policy recommendation for the prelimit system, one can set $\psi^{r}$ using $\hat{\psi}$ and the scaling relation \eqref{eq:scaled_control}. Also using the scaling relations, one can show that constraints \eqref{eq:lp1_off_lp}--\eqref{eq:lp4_off_lp} correspond to constraints \eqref{eq:prelimit_c1}--\eqref{eq:prelimit_c3}; see \ref{sec:derivation_proposed_policy} for details.
\end{remark}
\vspace{-5mm}
% ============================================================
%  SECTION 7: DATA, TEST PROBLEMS, AND BENCHMARKS
% ============================================================
\section{Data, test problems, and benchmark policies}\label{sect:data}

% ────────────────────────────────────────────────────────────
\subsection{Data}
\label{sect:data_section}
We use the dataset from the call center of a large U.S. bank, made publicly available by the Service Enterprise Engineering Lab at the Technion. (Available at the \href{https://see-center.iem.technion.ac.il/databases/USBank/}{SEE Center database}. Accessed on January 17, 2024.) The data spans from March 2001 to October 2003 and includes detailed records of agent activities and individual calls.  The call center operates 24/7 and is distributed across four locations: New York, Pennsylvania, Rhode Island, and Massachusetts. The workforce is divided into groups that are referred to as nodes, which do not correspond directly to physical sites. Agents from different locations may belong to the same node, and agents from the same location may be assigned to different nodes. The center handles between 330,000 and 350,000 calls on weekdays, and 170,000 to 190,000 on weekends, routing calls based on agent skill sets. Staffing levels adjust accordingly, with more than a thousand agents working on weekdays and several hundred on weekends, unevenly distributed across the nodes.

Customers first interact with a voice response unit (VRU), an automated system that enables them to complete transactions independently. While the majority exit the system after using the VRU, approximately 55,000 to 65,000 calls per day, around 20\% of total volume, are transferred to live agents. Our analysis focuses exclusively on these calls. The VRU directs each of these calls to agents with the appropriate skill set based on the requested service. Over time, the call center modified its service offerings, phasing out some services available in 2001 and introducing new ones after November 2002. To ensure consistency in arrival patterns, service types, and agent groups, we restrict our analysis to calls received between July and October 2002. 
Each call in the dataset is split into one or more subcalls, tracing its path through the call center from entry to exit. Our analysis focuses on the first subcall, which begins when the customer enters the queue to speak with an agent and ends upon completion of the first service. To remove outliers, we limit our analysis to calls with outcomes of normal termination, transfer, short abandonment, or abandonment, accounting for over 99\% of all observations. To further eliminate outliers, we restrict our analysis to calls with service times under 30 minutes and waiting times under 15 minutes. 

Figure \ref{Arrivals_CI} shows the average weekday call arrival pattern. To focus on the busier and more stable period of the day, we restrict our analysis to the peak period, i.e., to the calls received between 10 AM and 4 PM on weekdays. During this period, Retail class calls are routed to three service nodes, labeled 1, 2, and 3. Since Retail class alone accounts for approximately 68\% of total call volume, we divide it into three distinct subclasses based on the node where each call is handled. For each customer class, we compute the average hourly arrival rate by counting the number of calls during this six-hour window on each weekday, and then averaging across all weekdays from July to October 2002. After interacting with the VRU, customers who request agent assistance are placed in a queue. While waiting, some customers may abandon the queue before being served. However, since fewer than 3\% of calls end in abandonment, abandonment times are heavily censored, which can lead to biased estimates. To address this, we apply the bias-corrected Kaplan-Meier estimator proposed by \citet{stute1994jackknife} to estimate average abandonment rates for each customer class; also see \citet{akcsin2013structural} for an application of this method to the U.S. Bank call center data. Table \ref{table:arrival_abandonment_rates_2002} presents the resulting arrival percentage, hourly arrival rate and abandonment rate for each customer class.

\begin{figure}[!htb]
\centering
 \hspace*{-1.5cm}
    \hspace{20mm} \includegraphics[scale = 0.45]{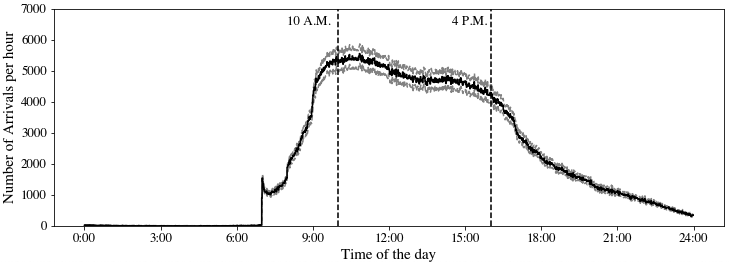}
    \caption{Hourly arrival rate of callers on weekdays during July--October 2002. The resolution of the horizontal axis is one minute. That is, arrival rates are calculated over one-minute intervals. The solid line shows the hourly average arrival rate across all days. The two dashed lines that enclose it are computed by taking the average plus/minus 2 times the standard deviation of the rates for each one-minute interval across all days.}
    \label{Arrivals_CI}
\end{figure}

Calls are routed to agents assigned to different groups according to their skill set. The dataset includes a group code for each agent, assigning them to an agent group that shares a common set of skills. In addition, the dataset identifies the main service type for each agent group, which is defined as the service type that accounts for the largest share of calls handled by agents in that group. During July--October 2002, there are 14 agent group codes: 1, 5, 9, 15/16, 19/20, 26, 28, 30, 31, 33, 34, 38, 39, and 40. The number of agents varies across agent groups. The notation 15/16 and 19/20 reflects that these are the same agents assigned different group codes on different days. Specifically, if group code 15 is active, group code 16 is not, and vice versa; the same applies to codes 19 and 20. These pairs never operate simultaneously and handle the same set of service types, with the code switching from one day to the next.

To estimate staffing levels, we focus on the 10~AM to 4~PM window,
which we divide into one-minute intervals. For each interval and each
agent group, we count the number of agents who are logged in and not
on a break, based on the second-by-second status codes in the dataset. Since our
analysis focuses on first subcalls, we scale each count by the
fraction of each group's total service time that is devoted to first
subcalls. We then average these adjusted counts over all intervals in
the six-hour window and round up to the nearest integer to obtain a
single staffing estimate per group code.
Table~\ref{table:group_code_main_service} in \ref{app:supply_data} reports the resulting
staffing levels along with the main service type handled by each
group, as inferred from the data.

For each group code, we compute the percentage (\%) of each service type among all calls handled by that group; see Table~\ref{table:num_obs_2002} in \ref{app:supply_data}. To eliminate outliers and filter out ad hoc and rare assignments due to unforeseen increases in the workload, we exclude assignments that constitute less than 0.1\% of the total call volume handled by a given group. For the rest of our analysis, we do not directly use the agent group codes provided in the raw dataset to define service stations. Instead, we reorganize the agents into service stations (or skill-based server pools) based on the service types they handle, ensuring that each server pool includes a substantial number of agents. 

We begin by merging group codes 26, 28, and 30 into a single station, as all three primarily serve the Business customer class. Similarly, we merge group codes 15/16, 38, 39, and 40 into a single station, as these groups collectively handle Retail, Case Quality, and Priority Service calls. The remaining group codes each form their own service station. Table \ref{table:code_assignments} in \ref{app:supply_data} summarizes the resulting mapping from the group codes to the service stations, along with the corresponding number of agents and the service types handled by each station. Table~\ref{table:service_rates_2002} in \ref{data_main_test} shows the hourly service rates for each service station and customer class. Consequently, for our analysis, the underlying network we consider has 22 nodes, 13 corresponding to customer classes and 9 to service stations, and 40 edges, representing the feasible service activities; see Figure \ref{fig:underlying_network}.
\begin{figure}[!htb]
    \centering
    \includegraphics[width=0.71\linewidth]{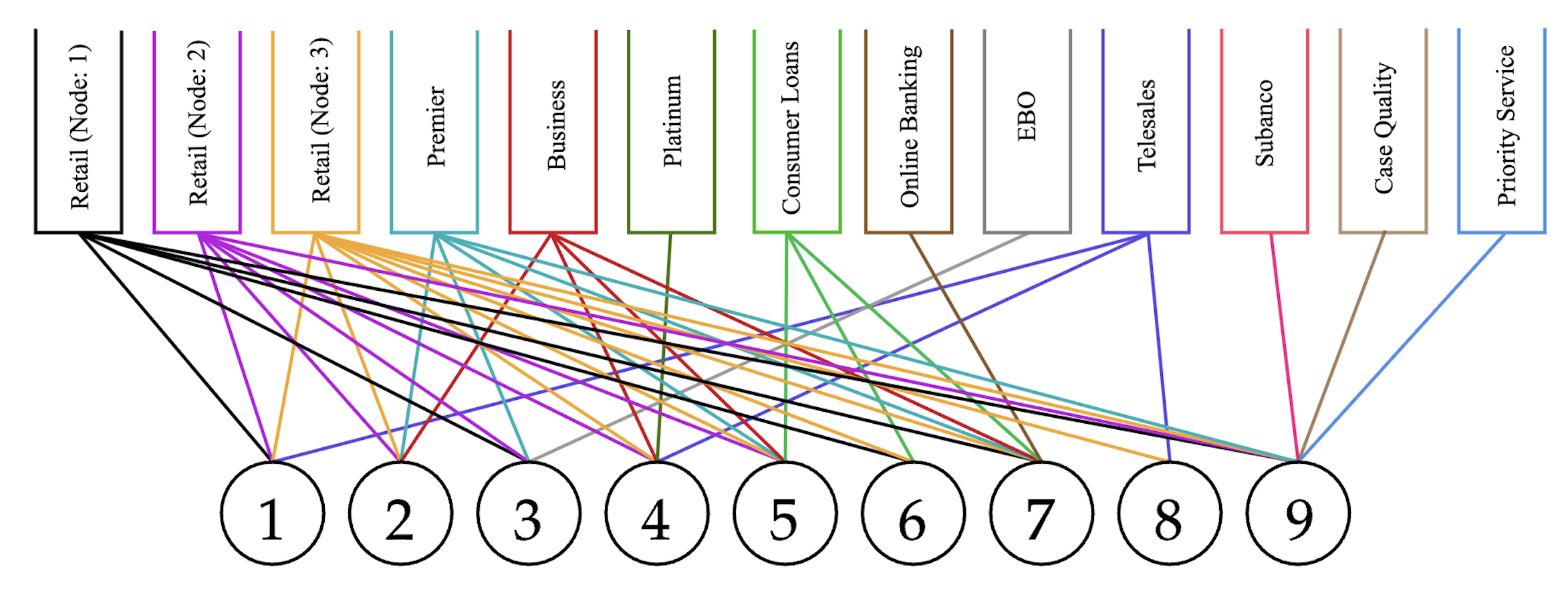}
    \caption{The underlying network of the multiclass, multi-server pool queueing system considered in our analysis, consisting of 13 customer classes and 9 service pools. Colors are used to distinguish customer classes and the service activities (edges) used to serve them.}
    \label{fig:underlying_network}
\end{figure}
In summary, we focus on 2,501,409 calls across 11 service types, 
received on weekdays during July--October 2002 (excluding holidays), 
between 10 AM and 4 PM. Since we split the Retail class into three 
subclasses based on the service node, our analysis involves 13 
customer classes. We restrict attention to calls that were routed to 
an agent queue and consider only their first subcall. The summary 
statistics for these calls are provided in 
Table~\ref{table:arrival_abandonment_rates_2002}.

\vspace{2mm}
\begin{table}[!htb]
	\centering
	\setlength\tabcolsep{4pt}
	{\small 
             \scalebox{0.78}{
		\begin{tabular}{lccccccc}
			\toprule
			Class && Arrival (\%) &&  Arrival Rate  && Abandonment Rate \\
			\rule{0pt}{3.5ex} &&  && (per hr) && (per hr)\\
			\midrule
			Retail (Node: 1) && 23.80 && 1216.22 && 7.01\\
                Retail (Node: 2) && 26.68 && 1363.70 && 7.74\\
                Retail (Node: 3) && 17.94 && 916.82 && 7.74\\
			Premier && 3.11 && 158.85 && 36.26\\			
                Business && 6.86 && 350.79 && 5.73\\
			Platinum && 0.60 && 30.50 && 6.12\\
			Consumer Loans && 7.76 && 396.38 && 4.57\\
			Online Banking && 3.17 && 161.94 && 8.25\\
			EBO && 0.86 && 43.84 && 7.38\\
			Telesales && 7.27 && 371.33 && 9.78\\
			Subanco && 0.71 && 36.48 && 7.62\\
			Case Quality && 0.53 && 27.05 && 13.37\\
			Priority Service && 0.73 && 37.21 && 17.54\\
			\bottomrule 
		\end{tabular}
	}
 }
        \caption{Arrival percentages, hourly arrival rates, and abandonment rates for each customer class.}
	\label{table:arrival_abandonment_rates_2002}
\end{table}

We complement the U.S. Bank data by estimating holding and abandonment costs using an additional data source; see the fifth and sixth columns of Table \ref{stats_main} in \ref{data_main_test}. To approximate holding costs, we assume that the opportunity cost of waiting is equal to the caller's foregone hourly wage. According to the
\href{https://www.bls.gov/news.release/empsit.t19.htm}{U.S. Bureau of Labor Statistics} (accessed on April 2, 2025), the average hourly wage in the retail industry is \$25, which we use as the holding cost rate for the Retail class. 

We group the remaining classes based on their perceived priority relative to the Retail class. Specifically, Premier, Business, Platinum, and Priority Service are treated as higher-priority classes and are assigned holding cost rates ranging from \$27 to \$33 per hour. The remaining lower-priority classes are further subdivided by call volume. For classes whose arrival rates constitute less than 1\% of total call volume, we assign the lowest holding cost rate of \$20 per hour. For the other classes in this group, we assign a holding cost rate of \$23 per hour. This tiered assignment maintains a weighted average holding cost rate close to \$25 per hour used for the Retail class. To estimate abandonment penalties, we approximate the caller's value for service using the value of their time spent in service. Following this logic, the abandonment penalty is set equal to the hourly holding cost rate divided by 15, assuming servers handle approximately 15 calls per hour. While this is a simplified approach, it aligns the abandonment penalties with the priority-based structure of the holding cost rates.

% ────────────────────────────────────────────────────────────
\subsection{Main test problem and its variant}\label{section:main_test}
We calibrate our main test problem using data from the U.S. Bank call center. The call center offers 11 distinct service types. In addition, we divide the Retail class into three separate subclasses based on the nodes where these calls are handled as mentioned above. As a result, we set the number of customer classes to $K = 13$. Furthermore, after combining agent group codes in the raw dataset as shown in Table \ref{table:code_assignments} in \ref{app:supply_data}, we identify 9 distinct service stations (server pools) and set $J = 9$.

We estimate the prelimit arrival rate of class $k$, denoted by $\lambda_{k}^{r}$, from the raw data (see the third column of Table \ref{table:arrival_abandonment_rates_2002}). Similarly, abandonment rates $\theta_{k}$ and service rates $\mu_{kj}$ are estimated directly from the data, as shown in Table \ref{table:arrival_abandonment_rates_2002} and Table \ref{table:service_rates_2002} in \ref{data_main_test}, respectively. These estimates are also used to compute the limiting service and abandonment rates; see Equation \eqref{eq:def:scaled_lambda}. We also estimate the number of agents at each service station $j$, denoted by $N_{j}^{r}$, from the data (see the third column of Table~\ref{table:code_assignments} in \ref{app:supply_data}). We set the scaling parameter to $r = 100$, which reflects the scale of the staffing levels observed in the data. We use this value to calculate the limiting staffing levels as $\nu_{j} = N_{j}^{r}/r$ for $j = 1,\ldots, J$.
%\begin{equation}
%    \nu_{j} = N_{j}^{r}/r, \quad j = 1, \ldots, J. %\label{eq:limiting_agents}
%\end{equation}

Next, we consider the static allocation problem \eqref{lp_obj}--\eqref{lp_constraint3} by putting $\lambda^{r}/r$ in place of $\lambda$. The solution to this problem yields system utilization of $\rho^{*} \approx 80\%$. This suggests that the call center is overstaffed currently, rendering congestion concerns less important. Therefore, in order to focus on a setting where congestion concerns play a more important role, we proportionally increase the arrival rates so that the resulting system utilization is 95\%. This leads to a more interesting and challenging problem from a scheduling perspective. It constitutes our main test problem.
\begin{figure}[!htb]
    \centering
    \includegraphics[scale=0.25]{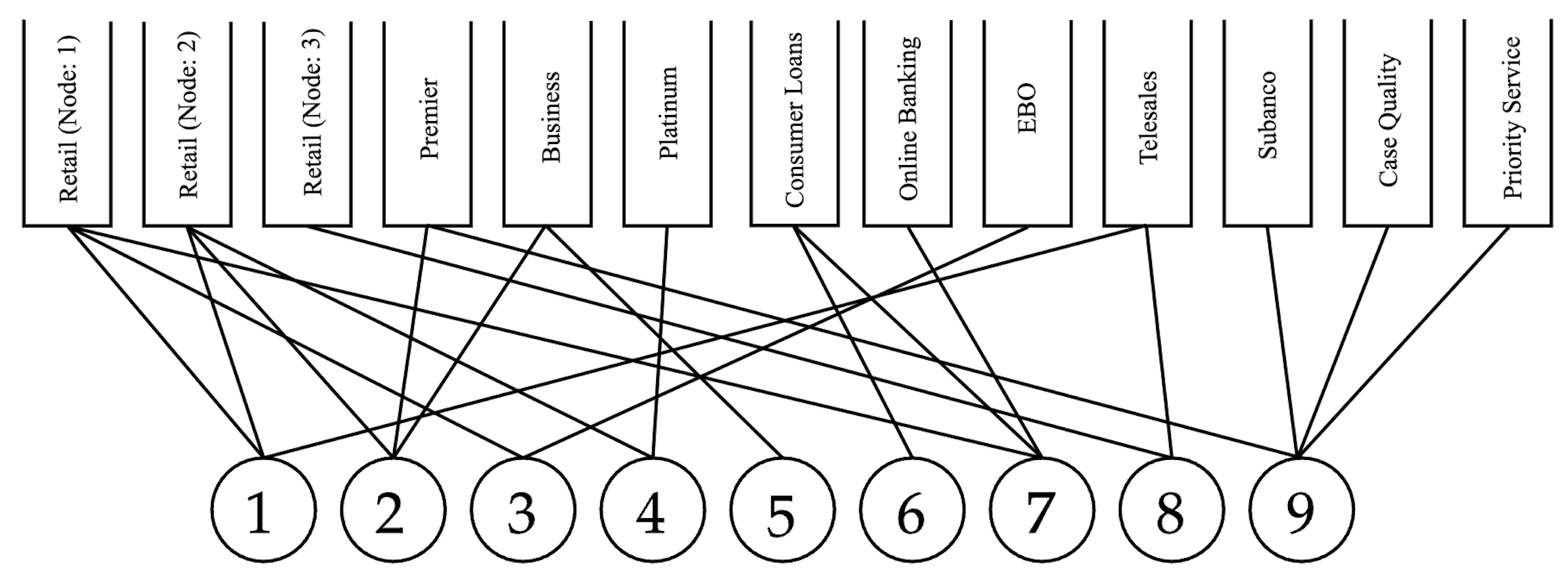}
    \caption{The underlying network corresponding to the system studied in the main test problem. We show only the edges that correspond to basic activities ($\xi_{kj}^{*} > 0$). }
    \label{fig:optimal_solution_2002_basic}
\end{figure}

To be more specific, we first scale up the arrival rate vector $(\lambda^{r}/r)$ (from our first solution of the static planning problem) proportionally to a value $\lambda$ so that the solution of the static planning problem satisfies the heavy traffic assumption; see Equation \eqref{optimal_solution_matrix} in \ref{appendix:optimal_solution} for that solution. Figure \ref{fig:optimal_solution_2002_basic}
shows the basic activities (i.e., those with $\xi_{kj}^{*} > 0$) in the system. While not immediately obvious from the figure, the underlying network structure is a tree with 22 nodes (13 customer classes and 9 service stations) and 21 edges. Building on this solution, we set $\tilde{\lambda}^{r} = 0.95r\lambda$ so that the $r^{\text{th}}$ system has 95\% utilization, corresponding to our main test problem. Then, to characterize the deviation from the fluid limit, we define the second-order terms $\zeta_{k}$ using Equation \eqref{eq:def:scaled_lambda}:
\begin{equation} 
\zeta_{k} = \frac{1}{\sqrt{r}} \left( \tilde{\lambda}_{k}^{r} - r\lambda_{k} \right), \quad k = 1, \ldots, K. \label{eq:limiting_zeta} 
\end{equation}
These $\zeta_k$ values capture the second-order deviation from the fluid limit in the diffusion scaling given in Equation \eqref{eq:def:scaled_lambda}. Lastly, we interpret the discount rate $\alpha$ in Equation~\eqref{eqn:j_definition} as the opportunity cost of capital. Accordingly, we use \href{https://www.federalreserve.gov/releases/h15/}{short-term U.S. Treasury bill rates} as a benchmark for the interest rate $\alpha$. Based on the prevailing rates, we set $\alpha = 4\%$ per year. Given these primitives of the Brownian control problem, we compute the proposed policy using our method and compare its performance against the benchmark policies introduced in Section~\ref{sect:benchmarks}. We observe that the FSF rule stands out among the benchmarks we considered for the main test problem. 

To understand why FSF performs well, we further investigate the system parameters of the main test problem. In doing so, we partition the classes into groups based on their service rates. Classes with the highest service rates form the high-priority group $\mathcal{G}_{H}$, whose queues are negligible under all benchmark policies we consider. The remaining low-priority classes are split into two subgroups, $\mathcal{G}_{L}^{\eta}$ and $\mathcal{G}_{L}^{\beta}$, as shown in Table~\ref{division}. Classes with similar service rates are placed in the same subgroup.
\vspace{-1mm}
\begin{table}[H]
    \centering
    \setlength\tabcolsep{4pt}
    {\footnotesize 
    \scalebox{0.9}{
    \begin{tabular}{llllcccccc}
        \toprule
        Group &    Classes \\
        \midrule
        $\mathcal{G}_H$ & Retail (Node: 3), Retail (Node: 2), Retail (Node: 1)\\
        \midrule
        $\mathcal{G}^{\eta}_{L}$ & Premier, Business, Telesales, Consumer Loans\\
        $\mathcal{G}^{\beta}_{L}$ & Platinum, Case Quality, Priority Service, Online Banking, Subanco, EBO \\
        \bottomrule
    \end{tabular}
    }
    }
    \caption{High-priority group and low-priority subgroups. Within each subgroup, classes are listed from left to right in descending FSF order.}
    \label{division}
\end{table}
\vspace{-5mm}
Table~\ref{updated_stats} reports the system parameters for the classes in each group. FSF orders the three groups as $\mathcal{G}_{H} \succ \mathcal{G}_{L}^{\eta} \succ \mathcal{G}_{L}^{\beta}$. That is, classes in $\mathcal{G}_{H}$ have the highest priority, followed by those in $\mathcal{G}_{L}^{\eta}$, whereas the classes in $\mathcal{G}_{L}^{\beta}$ have the lowest priority. So the queues of the classes in $\mathcal{G}_{H}$ are virtually empty. Most waiting customers belong to classes in $\mathcal{G}_{L}^{\beta}$ whereas customers in classes $\mathcal{G}_{L}^{\eta}$ may also experience queueing delays. In the main test problem, costs and service rates are positively correlated across the two low-priority subgroups: $\mathcal{G}_{L}^{\eta}$ has both higher service rates ($\sum_j \mu_{\cdot,j} \in [42.7, 66.2]$) and a higher average cost rate (\$43.2), while $\mathcal{G}_{L}^{\beta}$ has lower service rates ($\sum_j \mu_{\cdot,j} \in [9.1, 15.4]$) and a lower average cost rate (\$39.4). FSF's throughput-based ordering therefore aligns with the cost-based ordering rule $c\mu$, which partly explains its strong performance in lowering queueing costs.

To generate test problems in which this alignment breaks down, we scale the holding cost rate $h$ and abandonment penalty $p$ by $\eta \in (0,1]$ for classes in $\mathcal{G}_{L}^{\eta}$ and by $\beta \in [1,\infty)$ for classes in $\mathcal{G}_{L}^{\beta}$, leaving service rates unchanged. As $\eta$ decreases and $\beta$ increases, the correlation reverses: $\mathcal{G}_{L}^{\eta}$ becomes the low-cost, high-service-rate group and $\mathcal{G}_{L}^{\beta}$ the high-cost, low-service-rate group. FSF then prioritizes low-cost classes over high-cost ones, which is the condition under which one expects its performance to deteriorate. Table~\ref{updated_stats} summarizes these adjustments.

\begin{table}[!htb]
	\centering
	\setlength\tabcolsep{4pt}
	{\small 
             \scalebox{0.73}{
		\begin{tabular}{lccccccc}
			\toprule
			Class & Arrival & $\theta$ & $h$ & $p$ & $c$ & $\sum_{j} \mu_{\cdot,j}$\\
			\rule{0pt}{3.5ex} & percentage (\%) & (per hr.) & (per hr.) & (per job) & (per hr.) & (per hr.) \\
			\midrule
			Retail (Node: 3) & 17.94 & 7.74 & \$25 & \$1.667 & \$37.90 & 143.15\\
            Retail (Node: 2) & 26.68 & 7.74 & \$25 & \$1.667 & \$37.90 & 97.95\\
            Retail (Node: 1) & 23.80 & 7.01 & \$25 & \$1.667 & \$36.69 & 81.67\\
            \midrule
            Premier & 3.11 & 36.26 & \$27\,\textcolor{orange}{$\eta$} & \$1.800\,\textcolor{orange}{$\eta$} & \$92.27\,\textcolor{orange}{$\eta$} & 66.22\\
			Business & 6.86 & 5.73 & \$30\textcolor{orange}{$\eta$} & \$2.000\,\textcolor{orange}{$\eta$} & \$41.46\textcolor{orange}{$\eta$} & 62.80\\
			Telesales & 7.27 & 9.78 & \$23\textcolor{orange}{$\eta$} & \$1.533\textcolor{orange}{$\eta$} & \$37.99\textcolor{orange}{$\eta$} & 43.07\\
            Consumer Loans & 7.76 & 4.57 & \$23\textcolor{orange}{$\eta$} & \$1.533\textcolor{orange}{$\eta$} & \$30.01\textcolor{orange}{$\eta$} & 42.73\\
            \midrule
		    Platinum & 0.60 & 6.12 & \$33\textcolor{blue}{$\beta$} & \$2.200\textcolor{blue}{$\beta$} & \$46.46\textcolor{blue}{$\beta$} & 15.37\\
            Case Quality & 0.53 & 13.37 & \$20\textcolor{blue}{$\beta$} & \$1.333\textcolor{blue}{$\beta$} & \$37.82\textcolor{blue}{$\beta$} & 11.36\\
            Priority Service & 0.73 & 17.53 & \$33\textcolor{blue}{$\beta$} & \$2.200\textcolor{blue}{$\beta$} & \$71.59\textcolor{blue}{$\beta$} & 11.21\\
            Online Banking & 3.17 & 8.25 & \$23\textcolor{blue}{$\beta$} & \$1.533\textcolor{blue}{$\beta$} & \$35.65\textcolor{blue}{$\beta$} & 10.86\\
            Subanco & 0.71 & 7.62 & \$20\textcolor{blue}{$\beta$} & \$1.333\textcolor{blue}{$\beta$} & \$30.16\textcolor{blue}{$\beta$} & 10.86\\
            EBO & 0.86 & 7.38 & \$20\textcolor{blue}{$\beta$} & \$1.333\textcolor{blue}{$\beta$} & \$29.84\textcolor{blue}{$\beta$} & 9.13\\
			\bottomrule 
		\end{tabular}
	}
 }
        \caption{System parameters adjusted with scaling coefficients $\eta \in (0,1]$ and $\beta \in [1,\infty)$ for low-priority subgroups $\mathcal{G}_{L}^{\eta}$ and $\mathcal{G}_{L}^{\beta}$, respectively. Parameters for $\mathcal{G}_H$ are held fixed.}
	\label{updated_stats}
\end{table}
\vspace{-3mm}
Table~\ref{results} reports the average discounted queueing cost across 10{,}000 simulation replications for the FSF and $c\mu$ rules over a range of $(\eta,\beta)$ pairs. As $\eta$ decreases and $\beta$ increases, FSF's performance deteriorates relative to $c\mu$: the gap widens from $-10\%$ at $(\eta, \beta) = (1, 1)$ to $+73\%$ at $(0.3, 1.7)$. The cost adjustments make the classes in $\mathcal{G}_{L}^{\beta}$ more expensive, while their service rates remain the lowest in the system. The FSF rule continues to deprioritize these classes according to their service rates, causing expensive customers to accumulate in the queue.

\noindent\textbf{A variant of the main test problem.} We scale the holding cost rates $h_{k}$ and abandonment penalties $p_{k}$ by $\eta = 0.7$ for $k \in \mathcal{G}_{L}^{\eta}$ and by $\beta = 1.3$ for $k \in \mathcal{G}_{L}^{\beta}$. For larger $(\eta, \beta)$ gaps, $c\mu$ dominates by a wide margin, leaving little room for improvement. When $(\eta, \beta) = (0.7, 1.3)$, no static policy is dominant and coming up with a near optimal policy appears challenging, which makes it an interesting case to consider.
\vspace{-4mm}
\begin{table}[!htb]
    \centering
    \setlength\tabcolsep{4pt}
    {\small 
    \scalebox{0.8}{
    \begin{tabular}{lccccccc}
        \toprule
        $(\eta, \beta)$ &  FSF rule & $c\mu$ rule & Gap: FSF vs.\ $c\mu$ \\
        \midrule
        (1.0, 1.0) & 21,479,438 & 23,755,812 & $-$9.58\%\\
        (0.9, 1.1) & 23,440,071 & 23,561,278 & $-$0.51\%\\
        (0.8, 1.2) & 25,398,081 & 23,833,193 & 6.57\%\\
        (0.7, 1.3) & 27,358,860 & 23,487,455 & 16.48\%\\
        (0.6, 1.4) & 29,316,849 & 22,524,191 & 30.16\%\\
        (0.5, 1.5) & 31,277,527 & 21,767,266 & 43.69\%\\
        (0.4, 1.6) & 33,234,868 & 21,005,053 & 58.22\%\\
        (0.3, 1.7) & 35,195,781 & 20,319,722 & 73.21\%\\
                \bottomrule 
    \end{tabular}
    }
    }
    \caption{Average discounted queueing cost for each policy under different $(\eta, \beta)$ pairs. The gap column reports the percentage by which the FSF rule exceeds the $c\mu$ rule; negative values indicate that FSF outperforms $c\mu$.}
    \label{results}
\end{table}
\vspace{-8mm}
% ────────────────────────────────────────────────────────────
\subsection{Low-dimensional test problems}
\label{sect:low_dimensional}

In this section, we introduce two 2-dimensional test problems for which standard dynamic programming techniques are computationally tractable, allowing us to compute their optimal policies. These optimal policies serve as natural benchmarks for evaluating the performance of our proposed policy. Standard MDP techniques become computationally intractable as the problem dimension grows. For example, 3-dimensional instances are significantly harder to solve and can require several days of computation, while 4-dimensional instances are not feasible with the computing resources available to us.

To design the first 2-dimensional test problem, we partition the 13 classes of the main test problem into $K = 2$ groups, and combine the classes within each group into a single new class; see Table~\ref{class_division_2dim} in \ref{app:first_two_dim_data}. Similarly, we partition the 9 service stations of the main test problem into $J = 2$ groups, and combine the service stations within each group into a single new service station; see Table~\ref{table:code_assignments_2dim} in \ref{app:first_two_dim_data}. In this instance, however, the performance gap between the best and worst benchmark policies is less than 2\% (see the first column of Table~\ref{results_low_dim}), leaving limited room to demonstrate the value of accurately approximating the gradient of the value function $\nabla V(x)$. We therefore design a second 2-dimensional test problem based on the structural properties of the $N$-network analyzed in \citet{ghamami2013dynamic}. Those authors have proved the asymptotic optimality of a two-threshold policy under certain assumptions. For the second low-dimensional test problem, there is a larger performance gap between benchmarks. As such, it provides a setting in which the quality of the gradient approximation plays a decisive role in policy performance.
\vspace{-2mm}
\subsubsection{The first two-dimensional test problem}
We have $K = 2$ customer classes and $J = 2$ service stations. To design this test problem, we aggregate the Retail (Node: 1, 2, 3) classes into a single large class, while the remaining ten classes are combined into a second, smaller class. The resulting arrival rates for the two aggregated classes, derived directly from the dataset, are shown in the third column of Table~\ref{table_2dim} in \ref{app:first_two_dim_data}. The mean abandonment rates are set as weighted averages of the abandonment rates of individual classes within each group. Weights are proportional to each class' share of arrivals within its respective group, based on the 13-class main test example. Similarly, we set the hourly holding cost rates $h_k$, and abandonment penalties $p_k$ for each group by taking weighted averages of the corresponding cost parameters from the original classes.

To define the two service stations, we group agents according to their main service type. Agents whose main skill is Retail (see Table~\ref{table:group_code_main_service} in \ref{app:supply_data}) are assigned to service station~1, and all remaining agents are assigned to service station~2. Table~\ref{table:code_assignments_2dim} in \ref{app:first_two_dim_data} reports the agent-group codes aggregated into each station, together with the total number of agents. Table~\ref{table:service_rates_2dim} in \ref{app:first_two_dim_data} reports the service rates for each class--station pair. The resulting network is an $X$-network, as shown in Figure~\ref{fig:two_dim_x_network}. 

For this test problem, we set the scaling parameter to $r = 100$ to reflect the order of magnitude of staffing levels in each station as done earlier. This example is further studied in \ref{appendix:optimal_solution}, where the corresponding optimal nominal routing fractions $\xi_{kj}^{*}$ are shown in Equation~(\ref{optimal_solution_matrix_2dim}). All limiting quantities are computed following the same procedure described in Section \ref{section:main_test}.

\begin{figure}[!htb]
    \centering
    \begin{subfigure}[t]{0.25\textwidth}
        \centering
        \includegraphics[height=0.16\textheight]{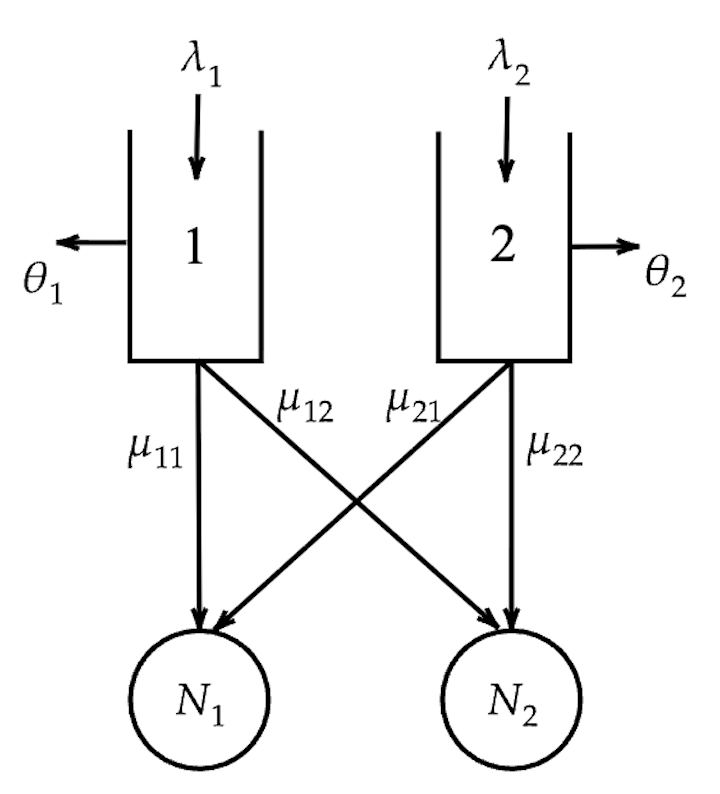}
        \caption{X-network.}
        \label{fig:two_dim_x_network}
    \end{subfigure}
    \hspace{10mm}
    \begin{subfigure}[t]{0.25\textwidth}
        \centering
        \includegraphics[height=0.16\textheight]{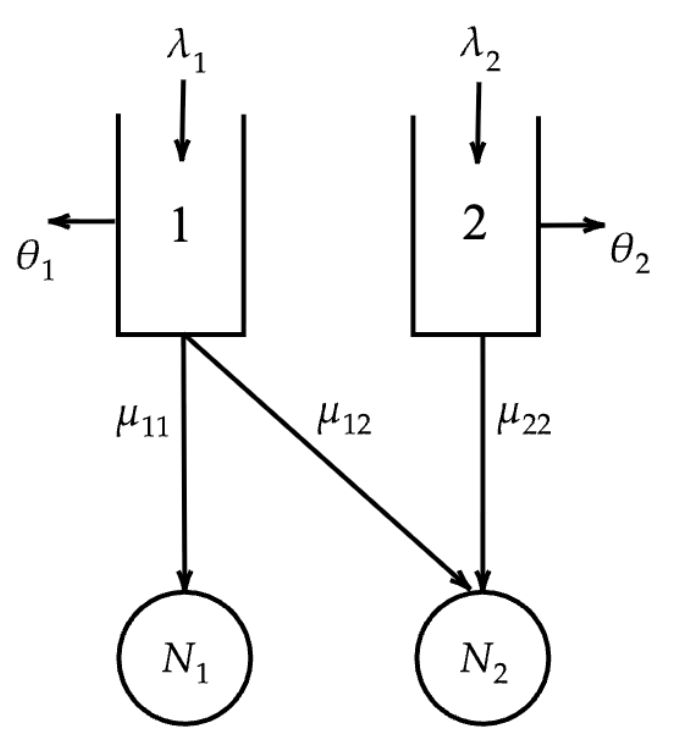}
        \caption{N-network.}
        \label{fig:two_dim_n_network}
    \end{subfigure}
    \caption{Network structures for the two low-dimensional test problems:
    (a) the first two-dimensional test problem and
    (b) the second two-dimensional test problem.}
    \label{fig:two_dim_networks}
\end{figure}
\vspace{-5mm}
% ────────────────────────────────────────────────────────────
\subsubsection{The second two-dimensional test problem}
\label{subsec:2dim_variant}

For the second 2-dimensional test problem, we design an instance whose structure follows the $N$-network studied by \citet{ghamami2013dynamic}, where server pool~1 is dedicated to class~1 and server pool~2 can serve both classes; see Figure~\ref{fig:two_dim_n_network}. \citet{ghamami2013dynamic} study this model in the conventional heavy traffic regime and derive asymptotically optimal policies. Here, we focus on the following particular parameter regime considered in \citet{ghamami2013dynamic}:
\begin{equation}
\label{eq:conditions}
c_1\mu_{12} > c_2\mu_{22}
\quad \text{and} \quad 
\frac{\theta_1+\alpha}{\theta_2+\alpha}
> \frac{c_1\mu_{12}}{c_2\mu_{22}}.
\end{equation}
The first inequality in~\eqref{eq:conditions} means that class~1 is more expensive in the $c\mu$ sense, so the cost structure favors giving server pool~2 to class~1. The second inequality assumes that the expensive class also abandons sufficiently faster than the cheaper class. Together, these conditions generate a trade-off. When the workload is low, server pool~2 prioritizes the more expensive class~1 directly, which minimizes cost. When the workload is high, however, continuing to prioritize class~1 would force class~2 jobs to accumulate in the queue; because class~2 abandons slowly, this backlog would persist and create costly congestion. In this parameter regime, \citet{ghamami2013dynamic} show that the following state-dependent two-threshold policy is asymptotically optimal: Let $L$ be a threshold on the number of class $1$ jobs and let $M$ denote a threshold on the workload process $\{W(t), \,\,t\geq0\}$ defined as
\begin{equation}\label{eq:workload}
W(t) = X_1(t) + \frac{\mu_{12}}{\mu_{22}}\,X_2(t), \quad t \geq 0.
\end{equation}
The policy operates as follows. Server pool~1 works whenever class~1 jobs are present. Server pool~2 prioritizes class~1 when $X_1(t)>L$ and $W(t)\leq M$, and prioritizes class~2 when $X_1(t)>L$ and $W(t)>M$. When $X_1(t)\leq L$, server pool~2 does not serve class~1; instead, it serves class~2 whenever buffer~2 is nonempty and otherwise idles.

The assumptions underlying \citet{ghamami2013dynamic} include $\lambda_{1} > \mu_{11}$, so that server pool~1 alone cannot handle class 1 demand and server pool~2 must help, together with conditions~\eqref{eq:conditions}. Because we are modeling call center operations, we focus attention on the many-server setting and make the following assumptions: (i)~$N_1 = N_2$, so the two server pools are equally staffed; (ii)~$\lambda_1 > \lambda_2$, so class~1 carries the larger arrival volume; and (iii)~$\mu_{12} < \mu_{22}$, so cross-trained service is slower than primary service, reflecting the reduced efficiency when a server pool~2 associate handles a class~1 call outside their specialty. 

More specifically, we set $r = 100$ and take $N_{1} = N_{2} = r$, so that the staffing scale matches the main test problem. Associates working on their primary class serve at rate $15$ calls per hour, matching the average service rate in the main test problem; cross-trained service is slower, with $\mu_{12} = 10 < 15 = \mu_{22}$. The long-run service fractions $(\xi_{12}, \xi_{22})$ for server pool~2 are chosen proportional to the service rates $(\mu_{12}, \mu_{22})$, giving $(\xi_{12}, \xi_{22}) = (0.4, 0.6)$. In the fluid limit, the balance equations yield the arrival rates $\lambda_1 = \nu_1\,\mu_{11}\xi_{11} + \nu_2\,\mu_{12}\,\xi_{12} = 19$ and $\lambda_2 = \nu_2\,\mu_{22}\,\xi_{22} = 9$.
%\[
%\lambda_1 = \nu_1\,\mu_{11}\xi_{11} + \nu_2\,\mu_{12}\,\xi_{12} = 19, \qquad
%\lambda_2 = \nu_2\,\mu_{22}\,\xi_{22} = 9.
%\]

To construct the corresponding prelimit system with target utilization $\rho = 0.95$, the prelimit hourly arrival rates are $\tilde{\lambda}_1^{r} = 0.95\,r\,\lambda_1 = 1805$ and $\tilde{\lambda}_2^{r} = 0.95\,r\,\lambda_2 =  855$. The cost and abandonment parameters are chosen such that conditions~\eqref{eq:conditions} hold. Table~\ref{table_2dim_variant} in \ref{app:second_two_dim_data} summarizes the arrival, cost, and abandonment parameters, and Table~\ref{table:service_rates_2dim_variant} in \ref{app:second_two_dim_data} reports the hourly service rates. The remaining quantities are computed following the same procedure as in Section~\ref{section:main_test}.

% ────────────────────────────────────────────────────────────
\subsection{A high-dimensional test problem}\label{sec:high_dim}

To evaluate the scalability of our method, we construct a test problem with
$\tilde{K} = 100$ customer classes and $\tilde{J} = 70$ service stations.
We set the total staffing to $\tilde{N}_{\text{total}} = 2{,}500$, matching
the maximum number of concurrent active tasks per instance supported by
Amazon Connect, a cloud-based contact center platform; see \citet{aws_connect_limits}. This represents a realistic upper bound on the scale of
operations in large call centers. 

We construct the system by growing the tree of basic activities of the main test problem in Figure~\ref{fig:optimal_solution_2002_basic} into a larger tree, attaching new customer classes and service stations sequentially as leaves. Each new node inherits its service and abandonment parameters from a template drawn uniformly at random from the original classes or stations. Long-run service fractions on the resulting tree are drawn from a symmetric Dirichlet distribution, arrival rates are set so that every station is fully utilized, and service rates on nonbasic edges are adjusted when necessary to preserve strict complementary slackness so that they indeed correspond to nonbasic activities. The system parameter $\tilde{r}$ and the target utilization $\tilde{\rho}^{\tilde{r}}$ are calibrated so that the limiting staffing levels and drift terms are of the same order of magnitude as those of the main test problem. By construction, the resulting system satisfies the heavy traffic assumption, the static planning problem has a unique optimal solution, and $\rho^{*} = 1$; see Proposition~\ref{prop:scaling} in \ref{app:high-dim-construction}, where the full construction, the underlying algorithm, and the cost parameters are provided.

% ────────────────────────────────────────────────────────────
\subsection{Benchmark policies} \label{sect:benchmarks}
For the two low-dimensional test problems introduced in Section~\ref{sect:low_dimensional}, it is computationally feasible to obtain optimal policies using standard dynamic programming techniques. These serve as natural benchmarks for comparison; see the online supplement for details. For the second two-dimensional test problem, we also consider the two-threshold policy of \citet{ghamami2013dynamic}, 
described in Section~\ref{sect:low_dimensional}.

Since an exhaustive search over all possible static priority policies is infeasible for the high-dimensional test problems, we focus on policies that have been well-studied in the literature for multiclass queueing systems: the $c\mu/\theta$ rule proposed by \citet{atar2010cmu}, the $c\mu$ rule proposed by \citet{cox1961queues}, and the fastest-server-first (FSF) rule proposed by \citet{armony2005dynamic}. The $c\mu/\theta$, $c\mu$ and FSF rules are well-studied policies and (asymptotically) optimal for different models; see \citet{atar2010cmu}, \citet{cox1961queues} and \citet{armony2005dynamic}, respectively. We also consider the $G$-$c\mu$ 
rule, a dynamic priority rule studied for a parallel-server network by \citet{mandelbaum2004scheduling}, focusing on quadratic 
queueing costs. 
\vspace{-5mm}
% ============================================================
%  SECTION 8: COMPUTATIONAL RESULTS
% ============================================================
\section{Computational results}
\label{sect:results}
This section compares our proposed policy, derived using the computational method in Section~\ref{sect:computational_method}, with the benchmark policies introduced in Section~\ref{sect:data}. For the low-dimensional problems, where the optimal policy is known from the MDP solution, our policy performs on par with the best benchmark. For the main and high-dimensional test problems, where the optimal policy is unknown, our policy outperforms all benchmarks considered. We use the same random seed for each simulation study with 10,000 replications. All the performance figures reported are subject to simulation and discretization errors.

\vspace{-3mm}

% ────────────────────────────────────────────────────────────
\subsection{Computational results for the low-dimensional test problems}
For low-dimensional test problems, the main benchmark policy is the optimal policy computed using standard dynamic programming techniques. Table \ref{results_low_dim} reports the average infinite-horizon discounted costs obtained in a simulation study, along with the percentage optimality gap of our proposed policy.
\begin{table}[H]
    \centering
    \renewcommand{\arraystretch}{1.2}
    \setlength\tabcolsep{4pt}
	{\small 
             \scalebox{0.8}{
    \begin{tabular}{lcccccccccccccccc}
        \toprule
          Method &&&&& First 2-Dimensional &&&&& Second 2-Dimensional \\
        \midrule
          Our Policy &&&&& 16,288,762  $\pm$ 138,854   &&&&&  20,552,308  $\pm$  151,233   \\
          $c\mu/\theta$ rule &&&&&   16,527,617  $\pm$  141,121   &&&&& 21,850,520 $\pm$ 145,678     \\
          $c\mu$ rule &&&&&  16,272,075  $\pm$ 140,896  &&&&&  21,175,098 $\pm$ 171,373\\
          FSF rule &&&&&  16,272,075 $\pm$ 140,896  &&&&&  21,850,520  $\pm$  145,678  \\
          G-$c\mu$ rule &&&&& 16,425,255 $\pm$ 139,698   &&&&&  21,413,156  $\pm$  153,544  \\
          Two threshold rule &&&&& NA &&&&&  20,399,333 $\pm$  154,244 \\
          Optimal Policy &&&&& 16,173,792 $\pm$   138,907   &&&&&  20,435,237 $\pm$  152,131  \\
          \midrule
          Optimality Gap &&&&& 0.71\% $\pm$ 1.22\% &&&&& 0.75\% $\pm$ 1.06\%\\
        \bottomrule
    \end{tabular}
    }
    \caption{The rows report the total cost $\pm$ the half-length of the 99\% confidence interval for each policy in the two low-dimensional test problems. The last row reports the percentage optimality gap $\pm$ the half-length of the 99\% confidence interval.}
    \label{results_low_dim}
    }
\end{table}
\vspace{-6mm}
\begin{figure}[!htb]
    \centering
    \includegraphics[width=0.45\linewidth]{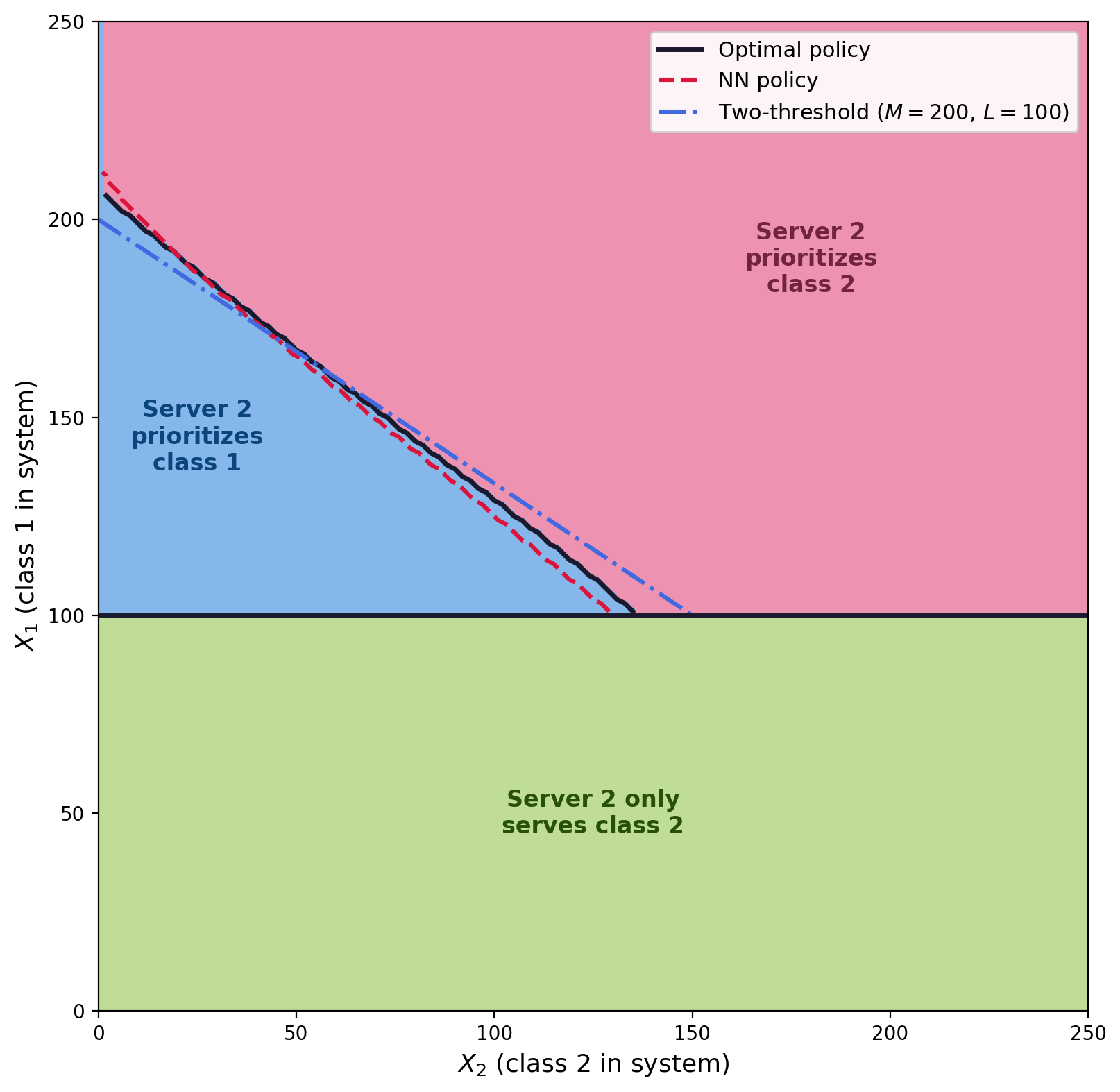}
    \caption{Priority-switching boundaries for the second two-dimensional test problem under the optimal policy, the proposed NN policy, and the two-threshold policy of \citet{ghamami2013dynamic}.}
    \label{fig:all_three_boundaries_2dim_variant}
\end{figure}
\vspace{-2mm}
Figure~\ref{fig:all_three_boundaries_2dim_variant} shows the priority-switching
boundaries of the optimal policy (derived via standard MDP techniques), the neural network-based policy, and the
two-threshold policy of \citet{ghamami2013dynamic} in the $(X_1, X_2)$ state
space. Because of the simple structure of the $N$-network, the only nontrivial control decision is which class server pool~2 should
prioritize, making the switching boundary a natural object to visualize and
compare across methods. All three boundaries agree closely in the region
$X_1 \in [120, 170]$ and $X_2 \in [50, 100]$, which is where the system spends most of its time under the stationary distribution. This explains why the three policies achieve statistically equivalent costs in Table~\ref{results_low_dim}. The two-threshold policy is
parametrized by $(M,L)$, which we calibrate by minimizing the simulated
cost over the grid $M \in \{150, 151, \ldots, 250\}$ and $L \in \{50, 51, \ldots, 150\}$, using $10{,}000$ replications per
grid point. This search required approximately 10 hours, and its cost
grows with the range of candidate values considered. By contrast, the
CTMC policy and the NN policy are each computed in less than one hour.
\color{black}
% ────────────────────────────────────────────────────────────
\subsection{Computational results for the main test problem, its variant, and the high-dimensional test problem}
\label{sect:computational_results_main_high_dim}

For the main test problem, its variant, and the 100-dimensional test problem, the optimal policy is not known. Therefore, we compare our proposed policy against the benchmark policies described in Section~\ref{sect:benchmarks}, using the best-performing benchmark in each case as the reference. Table~\ref{results_main_high_dim_test} reports the average infinite-horizon discounted costs from the simulation study, together with the percentage performance gap between our proposed policy and the best benchmark.

\begin{table}[H]
    \centering
    \renewcommand{\arraystretch}{1.2}
    \setlength\tabcolsep{4pt}
	{\small 
             \scalebox{0.8}{
    \begin{tabular}{lcccccccccccccccccccc}
        \toprule
          Method &&&&& Main (13-Dimensional) &&&&& Variant &&&&& 100-Dimensional\\
        \midrule
          Our Policy &&&&&  21,030,365 $\pm$  106,867  &&&&&  23,086,249  $\pm$ 126,375  &&&&& 20,325,378 $\pm$ 237,656 \\
          $c\mu/\theta$ rule &&&&& 25,084,619  $\pm$ 116,244  &&&&&   24,708,425 $\pm$ 112,632 &&&&& 46,400,824 $\pm$ 283,271 \\
          $c\mu$ rule &&&&& 23,755,812 $\pm$ 113,759  &&&&& 23,487,455  $\pm$ 125,318         &&&&& 83,229,246 $\pm$ 285,878 \\
          FSF rule &&&&&  21,479,438 $\pm$ 108,336  &&&&& 27,358,860  $\pm$ 137,719   &&&&& 21,933,477 $\pm$ 232,431 \\
          G-$c\mu$ rule &&&&& 28,567,826 $\pm$ 155,208  &&&&&  27,418,743  $\pm$ 141,917       &&&&& 46,595,574 $\pm$ 315,015  \\
          \midrule
          Performance Gap &&&&& -2.09\% $\pm$ 0.70\% &&&&& -1.71\% $\pm$ 0.75\% &&&&& -7.33\% $\pm$ 1.46\% \\
        \bottomrule
    \end{tabular}
    }
    \caption{The rows report the total cost $\pm$ the half-length of the 99\% confidence interval for each policy in the main, variant, and high-dimensional test problems. The last row reports the percentage performance gap between the proposed policy and the best benchmark, together with the half-length of the 99\% confidence interval.}
    \label{results_main_high_dim_test}
    }
\end{table}
\vspace{-4mm}
For the main test problem, the best-performing benchmark is the FSF rule, whereas for the cost-scaled variant it is the $c\mu$ rule. In the 100-dimensional test problem, the FSF rule is again the best benchmark. Our proposed policy outperforms the best benchmark in all three cases, improving performance by about two percentage points in the two 13-dimensional test problems and by seven percentage points in the 100-dimensional test problem.

% ────────────────────────────────────────────────────────────
\section{Concluding Remarks}
\label{sect:concluding_remarks}
\vspace{-2mm}
We close by highlighting some implications of generality in our formulation. Prior work in this regime, for example \citet{atar2005diffusion}, restricts attention to settings in which all activities are basic and policies satisfy the joint work conservation (JWC) assumption in the heavy traffic limit. JWC is the assumption that no agent is idle whenever any customer class has a nonempty queue. A weaker condition is local work conservation, which requires that no agent idles whenever there is a waiting customer from a class that the agent is eligible to serve. Our formulation allows both basic and nonbasic activities, and we permit policies that need not be jointly work conserving. Although our formulation does not impose local work conservation either, we observe that it holds for the policies we evaluate. 

In our test problems, we observe for both our proposed policy and the best benchmark policy that the nonbasic activities are used a significant fraction of the time. We also observe that the JWC condition can be violated. As will be illustrated below, this appears to depend on the network topology, especially on its sparsity. These observations point to the need for further theoretical research that incorporates these features more generally than the prior literature. Next, we illustrate our findings in the context of the main test problem and the 100-dimensional test problem.

\begin{figure}[!htb]
    \centering
    \includegraphics[width=0.99\textwidth]{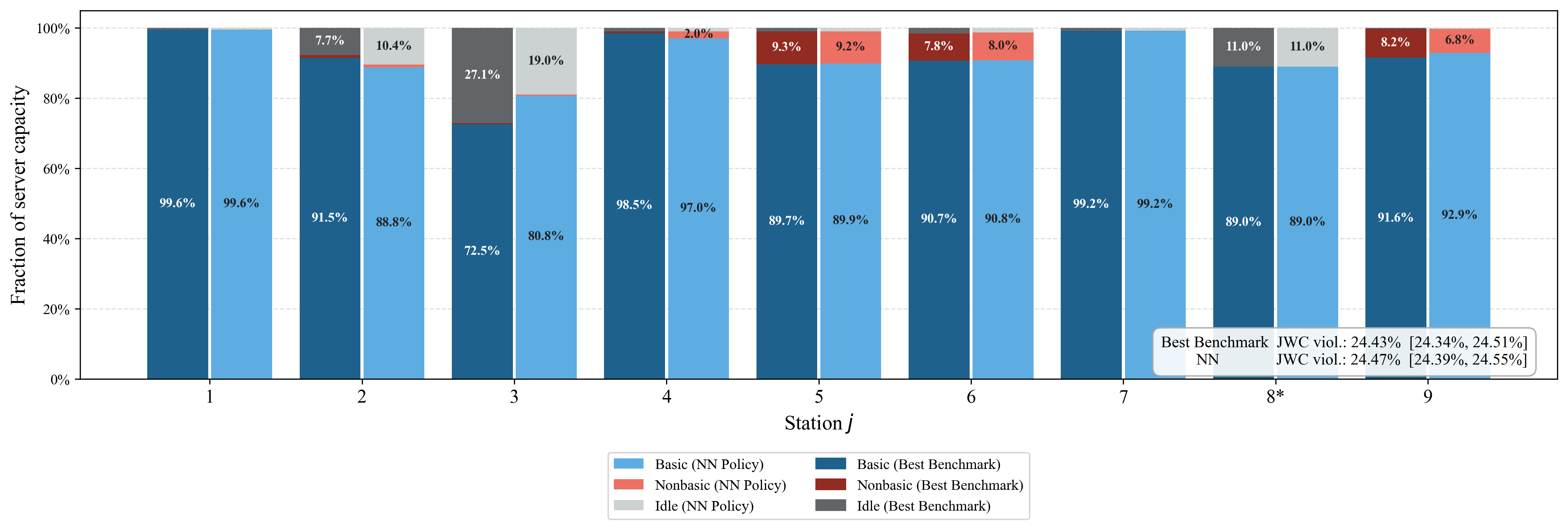}
    \caption{
    Capacity decomposition by station in the main test problem under the best benchmark policy and the proposed NN policy. Each bar reports the fraction of station capacity allocated to basic activities, nonbasic activities, and idleness. The figure also reports the percentage of time that each policy violates joint work conservation. The asterisk denotes the station for which all connected activities are basic.
    }
    \label{fig:main_capacity_comparison}
\end{figure}
\vspace{-3mm}
Figure \ref{fig:main_capacity_comparison} compares the fraction of server capacity allocated to basic activities, nonbasic activities, and idleness at each station in the main test problem under the best benchmark policy (FSF) and under our proposed policy. It also shows that nonbasic activities account for a substantial share of service capacity in several stations.  Figure~\ref{fig:high_capacity_comparison} in ~\ref{100_dim_graphs} shows the corresponding capacity decompositions for the 100-dimensional test problem under the same two policies. Together, these results suggest that restricting attention to policies that use only basic activities would exclude behavior that naturally arises in this data-driven system. Next, we consider joint work conservation. The main test problem has a sparse network structure: six of the thirteen classes can each be served by exactly one service station. As a result, when demand accumulates for one of these classes, the corresponding station can become a bottleneck even if agents at other stations remain idle. This explains why both policies violate JWC approximately 24\% of the time. On the other hand, in the 100-class, 70-station network, each class is connected to several stations, so the single-station bottlenecks present in the main test problem are largely absent. The best benchmark policy and our proposed policy differ substantially in their adherence to JWC. The best benchmark violates JWC about 3.15\% of the time, whereas our policy violates it only about 0.02\% of the time. Thus, although joint work conservation is not imposed by our formulation, the proposed policy recovers it almost exactly when the network topology makes it attainable. 

% ============================================================
%  BIBLIOGRAPHY
% ============================================================

\begingroup
\setstretch{0.92}
\setlength{\bibsep}{0pt}
\bibliographystyle{apalike}
\bibliography{bibfile}
\endgroup
\newpage
% ============================================================
%  APPENDICES
% ============================================================
\appendix
\renewcommand{\thesection}{Appendix \Alph{section}}
\section{Proof of Proposition 1}
\label{app:proof_proposition}
\begin{proof}
Applying Ito's formula to $e^{-\alpha t}\,V(\tilde{X}(t))$ on $[0,T]$ and using \eqref{eq:state_process_limit} yields
\begin{align}
    e^{-\alpha T}\,V(\tilde{X}(T)) - V(\tilde{X}(0)) &= \int_{0}^{T} e^{-\alpha t}\Big[\sum_{k = 1}^{K} \frac{\partial V(\tilde{X}(t))}{\partial x_{k}} \Big(\zeta_{k} - \theta_{k} \tilde{X}_{k}(t) + \sum_{j \in \mathcal{J}(k)} (\theta_{k} - \mu_{kj})\tilde{\psi}_{kj}(t)\Big)\Big]dt \nonumber\\
    & \quad + \int_{0}^{T} e^{-\alpha t}\Big[\sum_{k = 1}^{K} \lambda_{k} \frac{\partial^{2}V(\tilde{X}(t))}{\partial x_{k}^{2}} - \alpha V(\tilde{X}(t))\Big]dt \nonumber\\
    & \quad + \int_{0}^{T} e^{-\alpha t} \sum_{k = 1}^{K} \frac{\partial V(\tilde{X}(t))}{\partial x_{k}} \sqrt{2 \lambda_{k}}dB_{k}(t). \label{eq:ito_formula_value_func}
\end{align}
Multiplying both sides of the HJB equation (\ref{eqn:HJB_PDE}) by $e^{-\alpha t}$ and integrating over $[0,T]$ yields
\begin{align}
    \int_{0}^{T} e^{-\alpha t}\Big(\sum_{k = 1}^{K} \lambda_{k} &\frac{\partial^{2}V(\tilde{X}(t))}{\partial x_{k}^{2}} - \alpha V(\tilde{X}(t))\Big)dt =\nonumber \\
    &- \int_{0}^{T} e^{-\alpha t}\Big(\sum_{k = 1}^{K} c_{k}\tilde{X}_{k}(t) + (\zeta_{k} - \theta_{k} \tilde{X}_{k}(t))\frac{\partial V(\tilde{X}(t))}{\partial x_{k}}\Big)dt \nonumber \\
    & + \int_{0}^{T} e^{-\alpha t} \sup_{\psi \in \Psi(x)} \Big[\sum_{k = 1}^{K}\sum_{j \in \mathcal{J}(k)} \Big(c_{k} + (\mu_{kj} - \theta_{k})\frac{\partial V(\tilde{X}(t))}{\partial x_{k}}\Big)\psi_{kj}\Big]dt \label{eq:hjb_value_func}
\end{align}
Substituting Equation (\ref{eq:hjb_value_func}) into Equation (\ref{eq:ito_formula_value_func}), using the definition of $F$ (Equation (\ref{eq:f_function})), and writing the integrand of the stochastic integral term in (\ref{eq:ito_formula_value_func}) in vector notation gives Equation (\ref{eq:key_identity}). 
\end{proof}

% ────────────────────────────────────────────────────────────
%  Section: Approximating F(x,v)
% ────────────────────────────────────────────────────────────
\section{Approximating the auxiliary function \texorpdfstring{$F(x,v)$}{F(x,v)}} 
\label{sect:offline_approximation_auxiliary_function}
The auxiliary function $F(x,v),$ defined via Equation \eqref{eq:f_function}, can be rewritten as follows:
\begin{equation}
    F(x,v) = H(x,v) + \sum_{k=1}^{K}D_{k}(x)\,v_{k}-\sum_{k=1}^{K}c_{k}x_{k},
\end{equation}
where
\begin{equation}
    H(x,v) = \sup_{\psi \in \Psi(x)}\left[\sum_{k=1}^{K}\sum_{j\in \mathcal{J}(k)}(c_{k} + (\mu_{kj}-\theta_{k})\,v_{k})\,\psi_{kj}\right], \label{eq:h_function}
\end{equation}
\begin{equation}
    D_{k}(x) = \sum_{j \in \mathcal{J}(k)}(\theta_{k} - \mu_{kj})\,\tilde{\psi}_{kj}(x),\quad k = 1,\ldots,K. \label{eq:d_function}
\end{equation}
Also let $D(x) = (D_{1}(x),\ldots,D_{K}(x))^{\prime}$ for $x \in \mathbb{R}^{K}$.

As seen from Equation \eqref{eq:h_function}, computing $H(x,v)$ involves solving a linear program whose feasible set is $\Psi(x)$. Similarly, computing $D_{k}(x)$ also involves solving a linear program because our reference policy $\tilde{\psi}(\cdot)$ is defined via a linear program as explained below. Since these linear programs must be evaluated at every state visited along the sample path and at every iteration of the neural network training algorithm, solving them online is computationally expensive. To address this, we approximate the functions $H$ and $D$ with neural networks $\hat{H}$ and $\hat{D}$, respectively. These neural networks are trained on datasets generated by solving the corresponding linear programs over sampled inputs. Given the trained networks, the approximation of the auxiliary function $F(x,v)$, used in our computational method (see Algorithm \ref{main_algorithm}) is 
\begin{equation}
    \hat{F}(x, v) = \hat{H}(x, v) + \sum_{k=1}^{K}\hat{D}_{k}(x) \, v_{k} - \sum_{k=1}^{K} c_k\, x_k.
\end{equation}
Next, we describe further details of our offline approximations of $H$ and $D$.

% ── Supremum term ──
\paragraph{Approximating the $H$ function.}
As a preliminary to training the neural network $\hat{H}$ to approximate $H(x,v)$, we first note that $x$ corresponds to the (scaled) system state and $v$ corresponds to the gradient $\nabla V(x)$ of the value function at that state. For the neural network training, one needs to sample data points $\{(x^{(m)}, v^{(m)}): m=1,\ldots,M_{H}\}$. Then for each data point $(x^{(m)}, v^{(m)})$, we compute $H(x^{(m)}, v^{(m)})$ to use as the ground truth for the neural network training. 

Let us first describe how we sample the data for training. We set $M_{H}$ = 1 million and uniformly sample $(x,v)$ pairs from the compact set
\begin{equation*}
    \prod_{k=1}^{K}\left[\ushort{x}_{k}, \bar{x}_{k}\right] \times \prod_{k=1}^{K}\left[\ushort{v}_{k}, \bar{v}_{k}\right] \subset \mathbb{R}^{2K},
\end{equation*}
where the bounds $\ushort{x}_{k}$, $\bar{x}_{k}$, $\ushort{v}_{k}$, and $\bar{v}_{k}$ $(k=1,\ldots,K)$ are determined as follows. First, in order to set the bounds $\ushort{x}_{k}$ and $\bar{x}_{k}$ $(k=1,\ldots,K)$, we simulate the (prelimit) system under the first-come-first-served (FCFS) rule and let $\bar{X}^{r}$ denote the maximum total number of customers ever observed in the system. Thus, we have $0 \leq X^{r}_{k}(t) \leq \bar{X}^{r}$ for all $k,t$.

Then applying the diffusion scaling in Equation \eqref{eq:scaled_state} to these (prelimit) inequalities yields the following bounds for the (limiting) state:
\begin{equation*}
    \ushort{x}_{k} = \frac{0 - rx_{k}^{*}}{\sqrt{r}} = -\sqrt{r}x_{k}^{*}, \quad \text{and} \quad \bar{x}_{k} = \frac{\bar{X}^{r} - rx_{k}^{*}}{\sqrt{r}}, \quad k = 1,\ldots,K.
\end{equation*}
Second, to set the bounds $\ushort{v}_{k}$ and $\bar{v}_{k}$ on the gradient of the value function $V$, we rely on the structural properties of $V$. One expects the value function to be nondecreasing in each of its arguments, i.e., $\partial V/\partial x_{k} \geq 0$ for $k = 1,\ldots, K$. Thus, we set $\ushort{v}_{k} = 0,$ for $k = 1,\ldots,K$.

Intuitively, $\partial V/\partial x_{k}$ corresponds to the rate of change in the optimal objective as class $k$ queue length increases. Consider adding a customer to class $k$ queue. If this job is never served, it eventually abandons. The corresponding additional (discounted) cost is given as 
\begin{equation*}
    \frac{c_{k}}{(\theta_{k} + \alpha)} \leq \frac{c_{k}}{\theta_{k}}, \quad k =1,\ldots,K.
\end{equation*}
Thus, one expects $\partial V/\partial x_{k} \leq c_{k}/\theta_{k}$ for $k = 1,\ldots, K$. As such, we set $\bar{v}_{k} = c_{k}/\theta_{k}$ for $k = 1,\ldots,K$.
%\begin{equation*}
%    \bar{v}_{k} = \frac{c_{k}}{\theta_{k}}, \quad k = 1,\ldots, K.
%\end{equation*}
Third, given a data point $(x^{(m)}, v^{(m)})$, we solve the following linear program to compute $H(x^{(m)}, v^{(m)}):$ Choose $\psi \in \mathbb{R}^{|\mathcal{E}|}$ so as to
\begin{align}
&\text{Maximize } \sum_{k = 1}^{K} \sum_{j \in \mathcal{J}(k)} \Big(c_{k} + (\mu_{kj} - \theta_{k})\,v^{(m)}_{k}\Big)\psi_{kj} \label{eq:objective_off_lp}\\
&\text{subject to} \nonumber \\
&\sum_{j \in \mathcal{J}(k)} \psi_{kj} \leq x^{(m)}_{k}, \quad \,  k = 1,\ldots,K, \label{eq:lp1_off_lp}\\
&\sum_{k \in \mathcal{K}(j)} \psi_{kj} \leq 0, \quad \quad \,\, j = 1, \ldots,J, \label{eq:lp2_off_lp}\\
&\psi_{kj} \geq 0, \quad \quad \quad \quad \,\,\,\,\,\,\, (k,j) \in \mathscr{N}, \label{eq:lp3_off_lp}\\
&\psi_{kj} \geq -\sqrt{r}\psi_{kj}^{*}  \quad \,\,\,\,\,\,\,\,\, (k,j) \in \mathscr{B}. \label{eq:lp4_off_lp}
\end{align}

Note that constraints \eqref{eq:lp1_off_lp}--\eqref{eq:lp3_off_lp} follow from the definition of the constraint set $\Psi(x)$; see Equation (\ref{eq:set_feasible_policies}). Crucially, the approximation that underlies our approach is that we approximate the $r^{\text{th}}$ system by the Brownian control problem. With this in mind, for the $r^{\text{th}}$ system, it follows from Equation (\ref{eq:service_basic1}) and from the natural requirement $\psi^{r}(t) \geq 0$ that $\hat{\psi}_{kj}^{r}(t) \geq -\sqrt{r}\psi_{kj}^{*}$ for $t \geq 0$. 
%\begin{equation*}
%    \hat{\psi}_{kj}^{r}(t) \geq -\sqrt{r}\psi_{kj}^{*}, \quad t \geq 0.
%\end{equation*}
Motivated by this, we impose constraint \eqref{eq:lp4_off_lp} on the (limiting) controls for the basic activities. We then denote the optimal objective of \eqref{eq:objective_off_lp}--\eqref{eq:lp4_off_lp} by $H(x^{(m)},v^{(m)})$.
\begin{remark}
Because we do not impose any restrictions on the relative magnitudes of $\mu_{kj}$ and $\theta_{k}$, the objective coefficient $c_{k} + (\mu_{kj} - \theta_{k})\,v_{k}$ can be negative for some $(k,j) \in \mathcal{E}$. In these cases, the objective function in \eqref{eq:objective_off_lp} can be made arbitrarily large by taking $\psi_{kj} \rightarrow -\infty$. Constraints \eqref{eq:lp3_off_lp}--\eqref{eq:lp4_off_lp} prevent this.
\end{remark}
Lastly, we proceed with the neural network training as described in Subroutine \ref{subroutine:h_approx}. The hyperparameters used to train the deep neural network $\hat{H}_{\omega}(\cdot,\cdot)$ are shown in Table \ref{tab:hnet-driftnet-hyperparams} in \ref{app:offline_hyperparams}.

% ── Subroutine: H(x,v) approximation ──
\begin{algorithm}[!htb]
\floatname{algorithm}{Subroutine}
\caption{Offline approximation of the supremum term $H(x,v)$.}
\label{subroutine:h_approx}
\textbf{Input:} The number of training samples $M_H$, the sampling domains $\prod_{k=1}^{K}[\ushort{x}_k, \bar{x}_k]$ and $\prod_{k=1}^{K}[\ushort{v}_k, \bar{v}_k]$, the number of epochs $E_H$, a batch size $S_H$, a learning rate schedule $(\text{lr}_0, \gamma, \text{milestones})$ (We use PyTorch's \texttt{MultiStepLR} scheduler: the learning rate starts at $\text{lr}_0$ and is multiplied by $\gamma \in (0,1)$ at each epoch listed in $\text{milestones}$), and neural network architecture hyperparameters (number of layers, neurons per layer, activation function).\\
\textbf{Output:} A trained neural network $\hat{H}(\cdot\,, \cdot)$ approximating $H(x,v)$.
\begin{algorithmic}[1]
\State \textbf{Data generation:}
\For{$m = 1, \ldots, M_H$}
    \State Sample $x^{(m)}_k \sim \text{Uniform}([\ushort{x}_k, \bar{x}_k])$ and $v^{(m)}_k \sim \text{Uniform}([\ushort{v}_k, \bar{v}_k])$ independently for $k = 1, \ldots, K$.
    \State Solve the linear program \eqref{eq:objective_off_lp}--\eqref{eq:lp4_off_lp} with $(x, v) = (x^{(m)}, v^{(m)})$ to obtain $H^{(m)} = H(x^{(m)}, v^{(m)})$.
\EndFor
\State \textbf{Network training:}
\State Initialize a feedforward neural network $\hat{H}_\omega: \mathbb{R}^{2K} \to \mathbb{R}$ with parameter vector $\omega$.
\State Split the dataset $\{(x^{(m)}, v^{(m)}), H^{(m)}\}_{m=1}^{M_H}$ into training (80\%) and validation (20\%) sets.
\For{epoch $= 1, \ldots, E_H$}
    \For{each mini-batch $\{(x^{(m)}, v^{(m)}), H^{(m)}\}$ of size $S_H$}
        \State Compute the loss: $\ell(\omega) = \frac{1}{S_H} \sum_{m} \big(\hat{H}_\omega(x^{(m)}, v^{(m)}) - H^{(m)}\big)^2$.
        \State Update $\omega$ using Adam optimizer.
    \EndFor
    \State Update learning rate according to schedule.
\EndFor
\State \Return $\hat{H}(\cdot\,,\cdot) \equiv \hat{H}_\omega(\cdot\,,\cdot)$.
\end{algorithmic}
\end{algorithm}
\color{black}
% ── Reference policy term ──
\paragraph{Approximating the $D$ function.}
Recall from Section~\ref{sect:equivalent_characterization} that our computational method requires a reference policy $\tilde{\psi}$ to (i)~simulate sample paths of the reference process $\tilde{X}$ via Equation~\eqref{eq:training_sde}, and (ii)~evaluate the auxiliary function $F(x,v)$ defined in Equation~\eqref{eq:f_function}. In both expressions, the reference policy enters through the term $\sum_{j \in \mathcal{J}(k)} (\theta_{k} - \mu_{kj})\,\tilde{\psi}_{kj}(x)$. To facilitate our analysis, recall the mapping $D: \mathbb{R}^{K} \rightarrow \mathbb{R}^{K}$, defined above by $D_k(x) = \sum_{j \in \mathcal{J}(k)} (\theta_k - \mu_{kj})\,\tilde{\psi}_{kj}(x),$ for $k = 1,\ldots, K$.

In the analysis below, we consider three reference policies drawn from the literature: the $c\mu$ rule \citep{cox1961queues}, the $c\mu/\theta$ rule \citep{atar2010cmu}, and the fastest-server-first (FSF) rule \citep{armony2005dynamic}. Each is obtained by solving a linear program over the feasible set defined by constraints \eqref{eq:lp1_off_lp}--\eqref{eq:lp4_off_lp}, with the objective function
\vspace{-3mm}
\begin{equation}
    \operatorname{Maximize} \sum_{k = 1}^{K} \sum_{j \in \mathcal{J}(k)} w_{kj}\,\psi_{kj},
\end{equation}
where the activity weights $w_{kj}$ are set to $c_{k}\mu_{kj}$ for the $c\mu$ rule, $c_{k}\mu_{kj}/\theta_{k}$ for the $c\mu/\theta$ rule, and $\mu_{kj}$ for the FSF rule.

Since constraint~\eqref{eq:lp1_off_lp} depends on the current state $x$, the optimal solution $\tilde{\psi}(x)$ must be recomputed at every state visited along the sample path. Solving this linear program online at each time step of the Euler discretization (Subroutine~\ref{subroutine:euler}) and at each evaluation of $F$ during training (Algorithm~\ref{main_algorithm}) is computationally demanding. To address this, we approximate $D(x)$ using a neural network $\hat{D}$ trained offline; see Subroutine~\ref{subroutine:d_approx}. The state sampling domain used to generate the training data for approximating the function $D$ is the same as that used to approximate the function $H$, i.e., $\prod_{k=1}^{K}\left[\ushort{x}_{k}, \bar{x}_{k}\right],$
%\begin{equation*}
%    \prod_{k=1}^{K}\left[\ushort{x}_{k}, \bar{x}_{k}\right], 
%\end{equation*}
that is defined above. The hyperparameters used to train the deep neural network $\hat{D}_{\phi}(\cdot)$ are shown in Table \ref{tab:hnet-driftnet-hyperparams} in \ref{app:offline_hyperparams}.
% ── Subroutine: D(x) approximation ──
\begin{algorithm}[!htb]
\floatname{algorithm}{Subroutine}
\caption{Offline approximation of the reference policy term $D(x)$.}
\label{subroutine:d_approx}
\textbf{Input:} A reference policy rule (e.g., $c\mu$, $c\mu/\theta$, or FSF), the number of training samples $M_D$, the sampling domain $\prod_{k=1}^{K}[\ushort{x}_k, \bar{x}_k]$, the number of epochs $E_D$, a batch size $S_D$, a learning rate schedule $(\text{lr}_0, \gamma, \text{milestones})$, and neural network architecture hyperparameters (number of layers, neurons per layer, activation function).\\
\textbf{Output:} A trained neural network $\hat{D}(\cdot)$ approximating $D(x) = \big(D_1(x), \ldots, D_K(x)\big)$, where $D_k(x) = \sum_{j \in \mathcal{J}(k)} (\theta_k - \mu_{kj})\,\tilde{\psi}_{kj}(x)$ for $k = 1, \ldots, K$.
\begin{algorithmic}[1]
\State \textbf{Data generation:}
\For{$m = 1, \ldots, M_D$}
    \State Sample $x^{(m)}_k \sim \text{Uniform}([\ushort{x}_k, \bar{x}_k])$ independently for $k = 1, \ldots, K$.
    \State Solve the linear program corresponding to the chosen reference policy rule with state $x = x^{(m)}$ subject to constraints \eqref{eq:lp1_off_lp}--\eqref{eq:lp4_off_lp} to obtain $\tilde{\psi}^{(m)}$.
    \State Compute $D_k^{(m)} = \sum_{j \in \mathcal{J}(k)} (\theta_k - \mu_{kj})\,\tilde{\psi}_{kj}^{(m)}$ for $k = 1, \ldots, K$.
\EndFor
\State \textbf{Network training:}
\State Initialize a feedforward neural network $\hat{D}_\phi: \mathbb{R}^{K} \to \mathbb{R}^{K}$ with parameter vector $\phi$.
\State Split the dataset $\{x^{(m)}, D^{(m)}\}_{m=1}^{M_D}$ into training (80\%) and validation (20\%) sets.
\For{epoch $= 1, \ldots, E_D$}
    \For{each mini-batch $\{x^{(m)}, D^{(m)}\}$ of size $S_D$}
        \State Compute the loss: $\ell(\phi) = \frac{1}{S_D} \sum_{m} \big\|\hat{D}_\phi(x^{(m)}) - D^{(m)}\big\|^2$.
        \State Update $\phi$ using Adam optimizer.
    \EndFor
    \State Update learning rate according to schedule.
\EndFor
\State \Return $\hat{D}(\cdot) \equiv \hat{D}_\phi(\cdot)$.
\end{algorithmic}
\end{algorithm}
\vspace{-8mm}
\subsection{Hyperparameters of the neural networks for \texorpdfstring{$H$}{H} and \texorpdfstring{$D$}{D} functions}
\label{app:offline_hyperparams}
\begin{table}[H]
    \centering
    \setlength\tabcolsep{4pt} % default value: 6pt
    {\small
        \scalebox{0.73}{
        \begin{tabular}{lllllll}
            \toprule
            Hyperparameters && $\hat{H}_{\omega}(\cdot,\cdot)$ network && $\hat{D}_{\phi}(\cdot)$ network\\
            \midrule
            Number of hidden layers && 4 && 4 \\
            \\
            Number of neurons per layer && 150 && 150 \\
            \\
            Input batch normalization && Yes && Yes \\
            \\
            Hidden-layer batch normalization && Yes && Yes \\
            \\
            Activation function \& Weight initialization  && ELU \& Kaiming uniform && Leaky ReLU \& Kaiming uniform \\
            \\
            Optimizer && Adam && Adam \\
            \\
            Batch size && 4096 && 4096 \\
            \\
            Number of epochs && 5000 && 1000 \\
            \\
            Learning rate (epoch range) && 1e-2 \;\,(0, 150)     && 1e-2 \;\,(0, 150) \\
                                        && 5e-3 \;\,(150, 350)   && 1e-3 \;\,(150, 350) \\
                                        && 2.5e-3 (350, 5000)    && 1e-4 \;\,(350, 1000) \\
            \\
            Train/validation split && 80/20 && 80/20 \\
            \bottomrule
        \end{tabular}
        }
    }
    \caption{Summary of the hyperparameters used for the $\hat{H}$ and the $\hat{D}$ networks.}
    \label{tab:hnet-driftnet-hyperparams}
\end{table}

% ────────────────────────────────────────────────────────────
\section{Derivation of the proposed policy for the prelimit system}
\label{sec:derivation_proposed_policy}

We derive the prelimit linear program \eqref{eq:prelimit_obj}--\eqref{eq:prelimit_c3} by substituting the scaling 
relations \eqref{eq:scaled_state}--\eqref{eq:scaled_control} into the 
limiting linear program \eqref{eq:lp1_off_lp}--\eqref{eq:lp4_off_lp} and simplifying each 
component in turn.

\paragraph{Objective.}
Substituting Equation~\eqref{eq:scaled_control} 
into the objective of~\eqref{eq:objective_off_lp} gives
\begin{align*}
    &\sum_{k,\, j} \Big(c_k + (\mu_{kj} - \theta_k)\, 
    G_k^{\nu^*}(\hat{X}(t))\Big)\, \hat{\psi}_{kj} \\
    &\quad= \sum_{k,\, j} \Big(c_k + (\mu_{kj} - \theta_k)\, 
    G_k^{\nu^*}(\hat{X}(t))\Big) 
    \frac{\psi_{kj}^r(t) - r\psi_{kj}^*}{\sqrt{r}} \\
    &\quad= \frac{1}{\sqrt{r}}\sum_{k,\, j} \Big(c_k + (\mu_{kj} - \theta_k)\, 
    G_k^{\nu^*}(\hat{X}(t))\Big)\, \psi_{kj}^r(t) 
    - \frac{1}{\sqrt{r}}\sum_{k,\, j} \Big(c_k + (\mu_{kj} - \theta_k)\, 
    G_k^{\nu^*}(\hat{X}(t))\Big)\, r\psi_{kj}^*.
\end{align*}
Since $1/\sqrt{r} > 0$ and the second sum is constant with respect to 
the decision variables $\psi_{kj}^r(t)$, maximizing this expression is 
equivalent to maximizing
\begin{equation*}
    \sum_{k,\, j} \Big(c_k + (\mu_{kj} - \theta_k)\, 
    G_k^{\nu^*}(\hat{X}(t))\Big)\, \psi_{kj}^r(t),
\end{equation*}
which yields \eqref{eq:prelimit_obj}.

\paragraph{Constraint \eqref{eq:lp1_off_lp} $\to$ \eqref{eq:prelimit_c1}.}
Substituting the scaling 
relations~\eqref{eq:scaled_state}--\eqref{eq:scaled_control} into the 
class constraint $\sum_{j \in \mathcal{J}(k)} \hat{\psi}_{kj} \leq 
\hat{X}_k$ gives
\begin{equation*}
    \sum_{j \in \mathcal{J}(k)} 
    \frac{\psi_{kj}^r(t) - r\psi_{kj}^*}{\sqrt{r}}
    \leq \frac{X_k^r(t) - r\, x_k^*}{\sqrt{r}}.
\end{equation*}
Multiplying through by $\sqrt{r}$ gives
\begin{equation*}
    \sum_{j \in \mathcal{J}(k)} \psi_{kj}^r(t) 
    - r\!\sum_{j \in \mathcal{J}(k)} \psi_{kj}^* 
    \leq X_k^r(t) - r\, x_k^*.
\end{equation*}
By Equations~\eqref{eq:nominal_no_system}--\eqref{eq:nominal_no_service}, 
the fluid solution satisfies $\sum_{j \in \mathcal{J}(k)} 
\psi_{kj}^* = x_k^*$, so the constant terms on both sides 
cancel:
\begin{equation*}
    \sum_{j \in \mathcal{J}(k)} \psi_{kj}^r(t) 
    \leq X_k^r(t),
\end{equation*}
which yields \eqref{eq:prelimit_c1}.

\paragraph{Constraint \eqref{eq:lp2_off_lp} $\to$ \eqref{eq:prelimit_c2}.}
Substituting~\eqref{eq:scaled_control} into the station constraint 
$\sum_{k \in \mathcal{K}(j)} \hat{\psi}_{kj} \leq 0$ gives
\begin{equation*}
    \sum_{k \in \mathcal{K}(j)} 
    \frac{\psi_{kj}^r(t) - r\psi_{kj}^*}{\sqrt{r}} \leq 0.
\end{equation*}
Multiplying by $\sqrt{r}$ and rearranging gives
\begin{equation*}
    \sum_{k \in \mathcal{K}(j)} \psi_{kj}^r(t) 
    \leq r \sum_{k \in \mathcal{K}(j)} \psi_{kj}^*.
\end{equation*}
By Equations~\eqref{assumption:heavy_traffic} 
and~\eqref{eq:nominal_no_service}, the fluid solution satisfies 
$\sum_{k \in \mathcal{K}(j)} \psi_{kj}^* = \nu_j$. Together 
with $\nu_{j} = N_{j}^{r}/r$ for $j = 1,\ldots, J$, the right-hand side therefore 
equals $r\nu_j = N_j^r$, giving
\begin{equation*}
    \sum_{k \in \mathcal{K}(j)} \psi_{kj}^r(t) \leq N_j^r,
\end{equation*}
which yields \eqref{eq:prelimit_c2}.

\paragraph{Constraints \eqref{eq:lp3_off_lp}--\eqref{eq:lp4_off_lp} 
$\to$ \eqref{eq:prelimit_c3}.}
We consider basic and nonbasic activities separately. For nonbasic 
activities $(k,j) \in \mathscr{N}$, the fluid solution has 
$\psi_{kj}^* = 0$, so substituting into~\eqref{eq:scaled_control} 
gives $\hat{\psi}_{kj} = \psi_{kj}^r(t)/\sqrt{r}$, and the 
nonnegativity constraint $\hat{\psi}_{kj} \geq 0$ translates directly 
to $\psi_{kj}^r(t) \geq 0$. For basic activities $(k,j) \in 
\mathscr{B}$, substituting~\eqref{eq:scaled_control} into the lower 
bound $\hat{\psi}_{kj} \geq -\sqrt{r}\,\psi_{kj}^*$ gives
\begin{equation*}
    \frac{\psi_{kj}^r(t) - r\psi_{kj}^*}{\sqrt{r}} 
    \geq -\sqrt{r}\,\psi_{kj}^*,
\end{equation*}
and multiplying through by $\sqrt{r}$ and simplifying yields 
$\psi_{kj}^r(t) \geq 0$. In both cases we obtain 
\eqref{eq:prelimit_c3}.

% ────────────────────────────────────────────────────────────
\section{Data used for the test problems}
% ────────────────────────────────────────────────────────────
\subsection{Agent and service rate data}
\label{app:supply_data}
\begin{table}[H]
    \centering
    \setlength\tabcolsep{8pt}
    \renewcommand{\arraystretch}{1.1}
    \scalebox{0.65}{
        \begin{tabular}{ccc@{\hskip 24pt}ccc}
            \toprule
            Group & \# Agents & Service Type & Group & \# Agents & Service Type \\
            \midrule
            1     & 93 & Retail         & 31 & 31 & Consumer Loans   \\
            5     & 89 & Retail         & 33 & 19 & Online Banking   \\
            9     & 28 & EBO            & 34 & 45 & Telesales        \\
            15/16 & 12 & Retail         & 38 & 3  & Case Quality     \\
            19/20 & 15 & Premier        & 39 & 1  & Case Quality     \\
            26    & 15 & Business       & 40 & 4  & Priority Service \\
            28    & 3  & Business       &    &    &                  \\
            30    & 9  & Business       &    &    &                  \\
            \bottomrule
        \end{tabular}
    }
    \caption{Average number of agents and main service type for each agent group code.}
    \label{table:group_code_main_service}
\end{table}

\begin{table}[H]
    \centering
    \setlength\tabcolsep{6pt}
    \renewcommand{\arraystretch}{1.1}
    {\footnotesize
    \scalebox{0.7}{
    \begin{tabular}{ll*{13}{c}}
        \toprule
        Group && Retail & Retail & Retail & Premier & Business & Platinum & Consumer & Online & EBO & Telesales & Subanco & Case & Priority \\
        Code && (Node: 1) & (Node: 2) & (Node: 3) & & & & Loans & Banking & & & & Quality & Service \\
        \midrule
        1 && 62.35 & 23.50 & 13.92 & 0.06 & 0 & 0 & 0.04 & 0 & 0 & 0.12 & 0 & 0 & 0 \\
        5 && 0 & 62.85 & 36.63 & 0.17 & 0.33 & 0 & 0 & 0.01 & 0 & 0 & 0 & 0 & 0 \\
        9 && 58.59 & 25.51 & 0 & 0.29 & 0 & 0 & 0 & 0.07 & 15.53 & 0.01 & 0 & 0 & 0 \\
        15/16 && 29.47 & 0 & 47.43 & 2.15 & 0 & 0 & 0 & 0 & 0 & 0.02 & 20.93 & 0 & 0 \\
        19/20 && 0 & 9.36 & 9.17 & 77.28 & 2.37 & 0 & 1.82 & 0 & 0 & 0 & 0 & 0 & 0 \\
        26 && 0 & 1.99 & 4.28 & 0.05 & 90.96 & 1.80 & 0 & 0 & 0 & 0.92 & 0 & 0 & 0 \\
        28 && 0 & 8.48 & 0 & 0 & 86.96 & 3.05 & 0 & 0 & 0 & 1.51 & 0 & 0 & 0 \\
        30 && 0 & 0.69 & 0 & 0 & 79.72 & 19.55 & 0 & 0 & 0.005 & 0.03 & 0 & 0 & 0 \\
        31 && 3.77 & 0 & 9.66 & 0 & 0 & 0 & 86.57 & 0 & 0 & 0.01 & 0 & 0 & 0 \\
        33 && 3.12 & 0 & 10.44 & 0.15 & 0.42 & 0 & 0.27 & 85.59 & 0.02 & 0 & 0 & 0 & 0 \\
        34 && 0 & 0 & 0.24 & 0 & 0.05 & 0 & 0.06 & 0.01 & 0 & 99.63 & 0 & 0 & 0 \\
        38 && 0.31 & 0.39 & 0 & 0 & 0 & 0 & 0 & 0 & 0 & 0 & 0.15 & 92.71 & 6.44 \\
        39 && 0.22 & 0 & 0 & 0 & 0 & 0 & 0 & 0 & 0 & 0 & 0.65 & 64.44 & 34.70 \\
        40 && 0.58 & 0.12 & 0 & 0 & 0 & 0 & 0 & 0 & 0 & 0 & 0 & 7.01 & 92.29 \\
        \bottomrule
    \end{tabular}
    }}
    \caption{Percentage (\%) of each service type among all calls handled by each agent group code. Each row represents the conditional distribution of service types for the corresponding group.}
\label{table:num_obs_2002}
\end{table}

\begin{table}[H]
	\centering
	\setlength\tabcolsep{4pt}
	{\small 
             \scalebox{0.75}{
		\begin{tabular}{llllccll}
			\toprule
			Service Station && Combined Codes &&  \# of Agents && Service Types Offered\\
			\midrule
			1 && 1 && 93 && Retail (Node: 1, 2, 3), Telesales\\
                2 && 5 && 89 && Retail (Node: 2, 3), Premier, Business \\
                3 && 9 && 28 && Retail (Node: 1, 2), Premier, EBO\\
			4 && 26, 28, 30 && 27 && Retail (Node: 2, 3), Business, Platinum, Telesales\\			
                5 && 19/20 && 15 && Retail (Node: 2, 3), Premier, Business, Consumer Loans\\
			6 && 31 && 31 && Retail (Node: 1, 3), Consumer Loans\\
			7 && 33 && 19 && Retail (Node: 1, 3), Premier, Business, Consumer Loans, Online Banking\\
			8 && 34 && 45 && Retail (Node: 3), Telesales\\
			9 && 15/16, 38, 39, 40 && 20 && Retail (Node: 1, 2, 3), Premier, Subanco, Case Quality, Priority Service\\
			\bottomrule 
		\end{tabular}
	}
 }
        \caption{The codes of agent groups combined, the average number of agents in each service station and the service types that the agents in each service station can serve.}
	\label{table:code_assignments}
\end{table}

\begin{table}[!htb]
    \setlength\tabcolsep{3pt}
	{\small 
             \scalebox{0.79}{
		\begin{tabular}{llccccccccccccc}
			\toprule
			Service & Retail & Retail & Retail & Premier & Business & Platinum & Consumer & Online & EBO & Telesales & Subanco & Case & Priority \\
			Station & (Node: 1) & (Node: 2) & (Node: 3) & & & & Loans & Banking & & & & Quality & Service\\
			\midrule
                1 & 16.47 & 16.72 & 15.91 & - & - & - & - & - & - & 17.57 & - & - & -\\
                2 & - & 16.36 & 17.28 & 13.21 & 15.61 & - & - & - & - & - & - & - & -\\
                3 & 16.05 & 14.82 & - & 11.90 & - & - & - & - & 9.13 & - & - & - & -\\
                4 & -	& 16.65 & 16.74 & - & 15.45 & 15.37 & - & - & - & 15.87 & - & - & - \\
                5 & - & 16.14 & 16.47 & 13.58 & 17.01 & - & 14.12 & - & - & - & - & - & - \\
                6 & 16.07 & - & 18.13 & - & - & - & 15.18 & - & - & - & - & - & - \\
                7 & 16.29 & - & 15.19 & 12.49 & 14.73 & - & 13.43 & 10.86 & - & - & - & - & - \\
                8 & - & - & 26.71 & - & - & - & - & - & - & 9.63 & - & - & - \\
                9 & 16.79 & 17.26 & 16.72 & 15.04 & - & - & - & - & - & - & 10.86 & 11.36 & 11.21\\
			\bottomrule 
		\end{tabular}
	}
 }
        \caption{The hourly service rates for each service station and customer class pair for the main test problem and its variant.}
	\label{table:service_rates_2002}
\end{table}
\subsection{Data used for the main test problem and its variant}\label{data_main_test}
\begin{table}[H]
    \centering
    \vspace{0mm}
    \setlength\tabcolsep{4pt}
    {\small 
         \scalebox{0.8}{
        \begin{tabular}{lccccccccc}
            \toprule
            Class & Arrival & $\tilde{\lambda}^{r}$ & $\theta$ & $p$ & $h$ & $c$\\
            \rule{0pt}{3.5ex} & percentage (\%)  & (per hr) & (per hr) & (per job) & (per hr) & (per hr)\\
            \midrule
            Retail (Node: 1)  & 23.80 & 1398.65 & 7.01  & \$1.667 & \$25.00 & \$36.69\\
            Retail (Node: 2)  & 26.68 & 1568.26 & 7.74  & \$1.667 & \$25.00 & \$37.90\\
            Retail (Node: 3)  & 17.94 & 1054.34 & 7.74  & \$1.667 & \$25.00 & \$37.90\\
            Premier           &  3.11 &  182.68 & 36.26 & \$1.800 & \$27.00 & \$92.27\\
            Business          &  6.86 &  403.41 &  5.73 & \$2.000 & \$30.00 & \$41.46\\
            Platinum          &  0.60 &   35.08 &  6.12 & \$2.200 & \$33.00 & \$46.46\\
            Consumer Loans    &  7.76 &  455.84 &  4.57 & \$1.533 & \$23.00 & \$30.01\\
            Online Banking    &  3.17 &  186.23 &  8.25 & \$1.533 & \$23.00 & \$35.65\\
            EBO               &  0.86 &   50.42 &  7.38 & \$1.333 & \$20.00 & \$29.84\\
            Telesales         &  7.27 &  427.03 &  9.78 & \$1.533 & \$23.00 & \$37.99\\
            Subanco           &  0.71 &   41.95 &  7.62 & \$1.333 & \$20.00 & \$30.16\\
            Case Quality      &  0.53 &   31.11 & 13.37 & \$1.333 & \$20.00 & \$37.82\\
            Priority Service  &  0.73 &   42.79 & 17.54 & \$2.200 & \$33.00 & \$71.59\\
            \bottomrule 
        \end{tabular}
    }
    }
    \caption{Summary statistics for the data used in the main test problem.}
    \label{stats_main}
\end{table}
\begin{table}[H]
    \centering
    \vspace{0mm}
    \setlength\tabcolsep{4pt}
    {\small 
         \scalebox{0.8}{
        \begin{tabular}{lccccccccc}
            \toprule
            Class & Arrival & $\tilde{\lambda}^{r}$ & $\theta$ & $p$ & $h$ & $c$\\
            \rule{0pt}{3.5ex} & percentage (\%)  & (per hr) & (per hr) & (per job) & (per hr) & (per hr)\\
            \midrule
            Retail (Node: 1)  & 23.80 & 1398.65 &  7.01 & \$1.667 & \$25.00 & \$36.69\\
            Retail (Node: 2)  & 26.68 & 1568.26 &  7.74 & \$1.667 & \$25.00 & \$37.90\\
            Retail (Node: 3)  & 17.94 & 1054.34 &  7.74 & \$1.667 & \$25.00 & \$37.90\\
            Premier           &  3.11 &  182.68 & 36.26 & \$1.260 & \$18.90 & \$64.59\\
            Business          &  6.86 &  403.41 &  5.73 & \$1.400 & \$21.00 & \$29.02\\
            Platinum          &  0.60 &   35.08 &  6.12 & \$2.860 & \$42.90 & \$60.40\\
            Consumer Loans    &  7.76 &  455.84 &  4.57 & \$1.073 & \$16.10 & \$21.01\\
            Online Banking    &  3.17 &  186.23 &  8.25 & \$1.993 & \$29.90 & \$46.35\\
            EBO               &  0.86 &   50.42 &  7.38 & \$1.733 & \$26.00 & \$38.79\\
            Telesales         &  7.27 &  427.03 &  9.78 & \$1.073 & \$16.10 & \$26.59\\
            Subanco           &  0.71 &   41.95 &  7.62 & \$1.733 & \$26.00 & \$39.21\\
            Case Quality      &  0.53 &   31.11 & 13.37 & \$1.733 & \$26.00 & \$49.17\\
            Priority Service  &  0.73 &   42.79 & 17.54 & \$2.860 & \$42.90 & \$93.07\\
            \bottomrule 
        \end{tabular}
    }
    }
    \caption{Summary statistics for the data used in the variant test problem.}
    \label{stats_variant}
\end{table}

\subsection{Data used for the low-dimensional test problems}
\subsubsection{The first 2-dimensional test problem}
\label{app:first_two_dim_data}
\begin{table}[H]
	\centering
	\setlength\tabcolsep{4pt}
	{\footnotesize 
             \scalebox{0.9}{
		\begin{tabular}{llllcccccc}
			\toprule
			Class &  &  Names of the Combined Classes \\
			\midrule
			1 & & Retail (Node: 1, 2, 3)\\
			2 & & Premier, Business, Platinum, Consumer Loans, Online Banking, EBO, \\
            & & Telesales, Subanco, Case Quality, Priority Service\\
			\bottomrule 
		\end{tabular}
	}
 }
        \caption{The combination of original classes into two new classes.}
	\label{class_division_2dim}
\end{table}

\begin{table}[H]
	\centering
	\setlength\tabcolsep{4pt}
	{\footnotesize 
             \scalebox{0.9}{
		\begin{tabular}{lccccccccccccccccc}
			\toprule
			Class && Arrival && $\tilde{\lambda}^{r}$ && $\theta$ && $p$ && $h$ && $c$ \\
			\rule{0pt}{3.5ex} && percentage (\%) && (per hr) && (per hr) && (per job) && (per hr) && (per hr)\\
			\midrule
			 1 && 68.41 && 3450.24 && 7.40 && \$1.67 && \$25.00 && \$37.40\\
          2 && 31.59 && 1592.91 && 6.53 && \$1.68 && \$25.13 && \$36.10\\
			\bottomrule 
		\end{tabular}
	}
 }
        \caption{Summary statistics for the first two-dimensional test problem.}
	\label{table_2dim}
\end{table}

\begin{table}[!htb]
	\centering
	\setlength\tabcolsep{4pt}
	{\small 
             \scalebox{0.9}{
		\begin{tabular}{llllccccccccccccc}
			\toprule
    	Service Station &&& Combined Agent Codes &&&  \# of Agents\\
			\midrule
			1 &&& 1, 5, 15/16 &&& 194\\
                2 &&& 9, 19/20, 26, 28, 30, 31, 33, 34, 38, 39, 40 &&& 173\\
			\bottomrule 
		\end{tabular}
	}
 }
        \caption{The codes of agent groups combined, the number of agents in each service station for the first two-dimensional test problem where $K = 2$ and $J = 2$.}
	\label{table:code_assignments_2dim}
\end{table}
\begin{table}[!htb]
	\centering
	\setlength\tabcolsep{4pt}
	{\small 
             \scalebox{0.9}{
		\begin{tabular}{lccccccccccccccccc}
			\toprule
    	Classes  &&& Service Station 1 &&&  Service Station 2\\
			\midrule
			1 &&& 16.50 &&& 12.35\\
                2 &&& 16.20 &&& 12.14\\
			\bottomrule 
		\end{tabular}
	}
 }
        \caption{The hourly service rates for each customer class and service station pair for the first two-dimensional test problem.}
	\label{table:service_rates_2dim}
\end{table}

\subsubsection{The second 2-dimensional test problem}
\label{app:second_two_dim_data}
\begin{table}[H]
	\centering
	\setlength\tabcolsep{4pt}
	{\footnotesize 
             \scalebox{0.9}{
		\begin{tabular}{lccccccccccccccccc}
			\toprule
			Class && Arrival && $\tilde{\lambda}^{r}$ && $\theta$ && $p$ && $h$ && $c$ \\
			\rule{0pt}{3.5ex} && percentage (\%) && (per hr) && (per hr) && (per job) && (per hr) && (per hr) \\
			\midrule
			 1 && 67.86 && 1805 && 10.00 && \$2.00 && \$30.00 && \$50.00 \\
             2 && 32.14 && 855 && 5.00 && \$1.33 && \$20.00 && \$26.67 \\
			\bottomrule 
		\end{tabular}
	}
        }
        \caption{Summary statistics for the second two-dimensional test problem.}
	\label{table_2dim_variant}
\end{table}

\begin{table}[!htb]
	\centering
	\setlength\tabcolsep{4pt}
	{\small 
             \scalebox{0.9}{
		\begin{tabular}{lccccccccccccccccc}
			\toprule
    	Classes &&& Service Station 1 &&& Service Station 2 \\
			\midrule
			1 &&& 15.00 &&& 10.00 \\
            2 &&& --- &&& 15.00 \\
			\bottomrule 
		\end{tabular}
	}
        }
        \caption{The hourly service rates for each customer class and service station pair for the second two-dimensional test problem.}
	\label{table:service_rates_2dim_variant}
\end{table}
\newpage
\subsection{Optimal solution to the static planning problem}
\label{appendix:optimal_solution}

Throughout this appendix, $\xi_{kj}^{*}$ denotes the optimal solution to the static allocation problem defined by~\eqref{lp_obj}--\eqref{lp_constraint3}, computed with the parameters of the corresponding test instance.

\subsubsection{Main test problem}
\begin{equation}
\scalebox{0.75}{$
\xi_{kj}^{*} = \begin{pmatrix}
0.7252 & 0.0 & 0.7925 & 0.0 & 0.0 & 0.0 & 0.015 & 0.0 & 0.0 \\
0.0203 & 0.8303 & 0.0 & 0.9111 & 0.0 & 0.0 & 0.0 & 0.0 & 0.0 \\
0.0 & 0.0 & 0.0 & 0.0 & 0.0 & 0.0 & 0.0  & 0.9229  & 0.0\\
0.0 & 0.0479 & 0.0 & 0.0 & 0.0 & 0.0 & 0.0 & 0.0 & 0.4519\\
0.0 & 0.1219 & 0.0 & 0.0 & 1.0 &  0.0 & 0.0 & 0.0 & 0.0 \\
0.0 & 0.0 & 0.0 & 0.0889 & 0.0 & 0.0 & 0.0 & 0.0 & 0.0\\
0.0 & 0.0 & 0.0 & 0.0 & 0.0 & 1.0 & 0.0354 & 0.0 & 0.0\\
0.0 & 0.0 & 0.0 & 0.0 & 0.0 & 0.0 & 0.9496 & 0.0 & 0.0\\
0.0 & 0.0 & 0.2075 & 0.0 & 0.0 & 0.0 & 0.0 & 0.0 & 0.0\\
0.2545 & 0.0 & 0.0 & 0.0 & 0.0 & 0.0 & 0.0 & 0.0771 & 0.0\\
0.0 & 0.0 & 0.0 & 0.0 & 0.0 & 0.0 & 0.0 & 0.0 & 0.2032\\
0.0 & 0.0 & 0.0 & 0.0 & 0.0 & 0.0 & 0.0 & 0.0 & 0.1441\\
0.0 & 0.0 & 0.0 & 0.0 & 0.0 & 0.0 & 0.0 & 0.0 & 0.2008
\end{pmatrix}$}
\label{optimal_solution_matrix}
\end{equation}
\subsubsection{The first 2-dimensional test problem}
\begin{equation}
\scalebox{0.8}{$
\xi_{kj}^{*} = \begin{pmatrix}
1.0 & 0.2016\\
0.0 & 0.7984
\end{pmatrix}$}.
\label{optimal_solution_matrix_2dim}
\end{equation}
In the prelimit, the system is an $X$-network (Figure~\ref{fig:two_dim_x_network}); in the fluid limit, it reduces to an $N$-network (Figure~\ref{fig:two_dim_n_network}) with only basic activities.

\section{Construction of the high-dimensional test problem}
\label{app:high-dim-construction}
We build on the main test problem of Section~\ref{section:main_test} using a
scaling procedure that preserves the heavy traffic assumption by construction.

\paragraph{System parameter.}
Recall that $r = 100$ for the main test problem. Also recall that the
total number of agents for the main test problem is
$\sum_{j = 1}^{J} N_{j}^{r} = 367$. For the 100-dimensional test
problem, we set
\begin{equation*}\label{eq:rtilde}
    \tilde{r} = \left\lceil r \times
    \frac{\tilde{N}_{\mathrm{total}}}
         {\sum_{j=1}^{J} N^{r}_{j}} \right\rceil
    = \left\lceil 100 \times \frac{2{,}500}{367} \right\rceil = 682.
\end{equation*}
The rationale behind this choice is to ensure
\begin{equation}
    \sum_{j=1}^{\tilde{J}} \tilde{\nu}_{j}
    = \frac{\tilde{N}_{\mathrm{total}}}{\tilde{r}}
    \approx \frac{1}{r} \sum_{j=1}^{J} N_{j}^{r}
    = \sum_{j=1}^{J} \nu_{j},
\end{equation}
so that the limiting quantities $\sum_{j=1}^{\tilde{J}}\tilde{\nu}_{j}$
and $\sum_{j=1}^{J}\nu_{j}$ are close and that $\tilde{\nu}_{j}$ and
$\nu_{j}$ are of order 1 ($j=1,\ldots,\tilde{J}$).

\paragraph{Target utilization.}
The main test problem has utilization $\rho^{r} = 0.95$. For the
100-dimensional test problem, we set
\begin{equation}
    \tilde{\rho}^{\tilde{r}} = 1 - \frac{1 - \rho^{r}}{\sqrt{\tilde{r}/r}}
    \approx 0.98. \label{eq:utilization_large}
\end{equation}
The rationale for this choice stems from the statistical economies of
scale, i.e., the larger the system is, the larger the utilization can
be without sacrificing system performance. Under
\eqref{eq:utilization_large}, one sees that
$\sqrt{r}\,(1-\rho^{r}) = \sqrt{\tilde{r}}\,(1 - \tilde{\rho}^{\tilde{r}})$.
Because the drift terms $\zeta_{k}$ ($k=1,\ldots,K$) and
$\tilde{\zeta}_{k}$ ($k=1,\ldots,\tilde{K}$) are proportional to
$\sqrt{r}\,(1-\rho^r)$ and $\sqrt{\tilde{r}}\,(1-\tilde{\rho}^{\tilde{r}})$,
respectively, Equation \eqref{eq:utilization_large} leads to drift terms of
similar magnitudes in the approximating Brownian control problems for
the two systems; see Section 6.4 of \citet{ata2025dynamic} for a
similar discussion on designing high-dimensional test problems.

\paragraph{Staffing.}
The $\tilde{N}_{\mathrm{total}} = 2{,}500$ agents are distributed across
the $\tilde{J}$ stations as follows. First, we assign
$\tilde{N}_{\min} = 25$ agents to each station as a minimum base. Next,
we distribute the remaining
$\tilde{N}_{\mathrm{total}} - \tilde{J} \cdot \tilde{N}_{\min}$ agents in
a certain proportional manner. To explain this, we first mention that
each station consists of agents with identical skill sets that are
bootstrapped from those agents in the main test problem. Viewing the
agent types in the main test problem as templates, we let $\tau(j)$
denote the template station in the main test problem of station $j$
(in the 100-dimensional test problem). Station $j$ inherits its
parameters from the template station $\tau(j)$ of the main test
problem (see step (i) below) and the remaining agents are assigned to
station $j$ proportionally to $N_{\tau(j)}^{r}$. That is, we set
\begin{equation}
    \tilde{N}_{j} = \tilde{N}_{\min}
    + \left\lceil (\tilde{N}_{\mathrm{total}} - \tilde{J}\cdot\tilde{N}_{\min})
    \frac{N_{\tau(j)}^{r}}{\sum_{j'=1}^{\tilde{J}} N_{\tau(j')}^{r}}\right\rceil.
\end{equation}
The limiting staffing levels are then
$\tilde{\nu}_{j} = \tilde{N}_{j}/\tilde{r}$ for $j = 1,\ldots,\tilde{J}$.

\subsection{An algorithm for building a larger test problem}

Let $\mathcal{T}$ denote the tree of basic activities in the original
system that has $13 + 9 - 1 = 21$ edges (see
Figure~\ref{fig:optimal_solution_2002_basic}). We first describe the
general algorithm and then illustrate it on a small example.

\paragraph{Step~(i): Growing the tree.}
We expand $\mathcal{T}$ by attaching $\tilde{K} - K = 87$ new customer
classes and $\tilde{J} - J = 61$ new service stations as leaves to the bipartite graph of server pools and buffers, e.g., see Figure \ref{figure_model}. For each new customer class $k = K + 1, \ldots, \tilde{K}$, we draw it uniformly at random from the customer classes $\{1,\ldots, K\}$ of the main test problem. We denote the resulting class type as $\kappa(k) \in \{1,\ldots,K\}$ for $k = K+ 1, \ldots, \tilde{K}$. Recall that customer classes $1,\ldots,K$ of the new test problem are taken directly from the main test problem. Similarly, for each new server pool $j = J+1,\ldots, \tilde{J}$, we draw its type from the service pools $\{1,\ldots,J\}$ of the main test problem uniformly at random. We denote the resulting server pool type as $\tau(j)$ for $j = J+1, \ldots, \tilde{J}$. For notational convenience, we set $\kappa(k) = k$ for $k = 1,\ldots, K$ and $\tau(j) = j$ for $j = 1,\ldots, J$. New nodes are processed sequentially following the
order given in a random permutation, attaching them to a node already in
the tree chosen uniformly at random among eligible neighbors, where
eligibility requires that the pair
$(\kappa(k), \tau(j))$ corresponds to an edge in $\mathcal{E}$. Since
each step adds exactly one node and one edge, the resulting tree
$\tilde{\mathcal{T}}$ has $\tilde{K} + \tilde{J} - 1 = 169$ edges.
Additionally, the service and abandonment rates are given as $\tilde{\mu}_{kj} = \mu_{\kappa(k),\tau(j)}$ and $\tilde{\theta}_{k} = \theta_{\kappa(k)}$.

\paragraph{Step~(ii): Allocating capacity and setting the arrival rates.}
For each server pool $j$, the service fractions across its tree neighbors
$\mathcal{K}_{\tilde{\mathcal{T}}}(j) = \{k : (k,j) \in \tilde{\mathcal{T}}\}$
are drawn from a symmetric Dirichlet distribution  
\begin{equation}
\bigl(\tilde{\xi}^*_{kj}\bigr)_{k \in \mathcal{K}_{\tilde{\mathcal{T}}}(j)}
    \sim \mathrm{Dirichlet}(\mathbf{1}), \label{eq:dirichlet_high}
\end{equation}
which ensures $\tilde{\xi}^*_{kj} > 0$ on every tree edge and that
every station is fully utilized.

Given the limiting staffing levels $\tilde{\nu}_{j}$ and the service
fractions $\tilde{\xi}^*_{kj}$, the limiting arrival rates are
determined by the demand constraint of the static planning problem
(SPP), which requires that the total service capacity allocated to
each class meets its demand:
\begin{equation}\label{eq:limiting_staffing}
    \tilde{\lambda}_{k}
    = \sum_{j=1}^{\tilde{J}} \tilde{\nu}_{j}\,\tilde{\mu}_{kj}\,\tilde{\xi}^*_{kj},
    \qquad k = 1, \ldots, \tilde{K}.
\end{equation}
The prelimit arrival rates are then scaled to achieve the target
utilization $\tilde{\rho}^{\tilde{r}}$, and the second-order terms
follow from Equation~\eqref{eq:def:scaled_lambda}:
\begin{align}
    \tilde{\lambda}_{k}^{\tilde{r}}
    &= \tilde{\rho}^{\tilde{r}}\,\tilde{r}\,\tilde{\lambda}_{k}, \quad \,\,\,\, \qquad \qquad k = 1, \ldots, \tilde{K},
    \label{eq:prelimit_arrival_high}\\
    \tilde{\zeta}_{k}
    &= \frac{1}{\sqrt{\tilde{r}}}
       \bigl(\tilde{\lambda}^{\tilde{r}}_{k} - \tilde{r}\,\tilde{\lambda}_{k}\bigr), \qquad k = 1, \ldots, \tilde{K}.
    \label{eq:zeta_high}
\end{align}

\paragraph{Step~(iii): Nonbasic activities.}
We refer to the edges of $\tilde{\mathcal{T}}$ as the basic activities
of the high-dimensional system. However, because the server pools and buffers are drawn from the main test problem randomly, there may be additional edges that are not in $\tilde{\mathcal{T}}$. We set the service rates to those edges so that their corresponding activities remain nonbasic and that the solution to the static planning problem is unique. More specifically, let
$(\tilde{\alpha}^*, \tilde{\beta}^*)$ denote the optimal dual variables
of the SPP, with $\tilde{\alpha}^*_{k}$ associated with the demand
constraint of class $k$ and $\tilde{\beta}^*_{j}$ with the capacity
constraint of station $j$. By complementary slackness on basic
activities (cf.\ Equation~(2.15) of \citet{harrison1999heavy}), the following must hold:
\begin{equation}\label{eq:cs_tree_high}
    \tilde{\nu}_{j}\,\tilde{\mu}_{kj}\,\tilde{\alpha}^*_{k}
    = \tilde{\beta}^*_{j},
    \qquad (k,j) \in \tilde{\mathcal{T}}.
\end{equation}
Because $\tilde{\mathcal{T}}$ is a tree spanning every class and every
station, fixing a dual variable determines all remaining dual variables
through \eqref{eq:cs_tree_high} and iterating along $\tilde{\mathcal{T}}$ assigns a value to
each $\tilde{\alpha}^*_{k}$ and $\tilde{\beta}^*_{j}$.

By the last statement of Proposition~2 in \citet{harrison1999heavy},
the optimal dual solution must satisfy strict complementary slackness
for every nonbasic activity. That is, 
\begin{equation*}
    \tilde{\mu}_{kj}
    < \frac{\tilde{\beta}^*_{j}}{\tilde{\nu}_{j}\,\tilde{\alpha}^*_{k}},
    \qquad (k,j) \in \tilde{\mathcal{E}} \setminus \tilde{\mathcal{T}}.
\end{equation*}
If the rate $\mu_{\kappa(k),\tau(j)}$ assigned in Step~(i) already
satisfies this bound, we use it directly. Otherwise we set
\begin{equation}\label{eq:nontree_high}
    \tilde{\mu}_{kj}
    = (1 - \delta) \cdot
    \frac{\tilde{\beta}^*_{j}}{\tilde{\nu}_{j}\,\tilde{\alpha}^*_{k}},
    \qquad \delta = 0.01.
\end{equation}

\begin{proposition}\label{prop:scaling}
The system constructed by Steps~(i)--(iii) satisfies the heavy traffic
assumption. The SPP has a unique optimal solution with $\rho^* = 1$,
and the basic activities correspond to the edges of the tree $\tilde{\mathcal{T}}$.
\end{proposition}

\begin{proof}
The pair $(\tilde{\xi}^*, 1)$ is primal feasible. The demand
constraint holds by~\eqref{eq:limiting_staffing} and every capacity
constraint binds by~\eqref{eq:dirichlet_high}, so $\rho^* = 1$. The
dual variables satisfy~\eqref{eq:cs_tree_high} on basic activities and
strict complementary slackness on nonbasic activities
by~\eqref{eq:nontree_high}, so by Proposition~3 of
\citet{harrison1999heavy} the optimal SPP basis is unique and consists
of the edges of $\tilde{\mathcal{T}}$.
\end{proof}

\subsection{Cost parameters}

For each class $k = 1, \ldots, \tilde{K}$, we draw a holding cost rate
$\tilde{h}_{k}$ uniformly from $[\$15, \$35]$. This range spans the
full set of holding cost rates observed in the main test problem. Its
midpoint of $\$25$ corresponds to the Retail class. Following the same
logic as in Section~\ref{sect:data_section}, we set the abandonment
penalty to $\tilde{p}_{k} = \tilde{h}_{k}/15$, reflecting the value of
one service interaction for an agent handling approximately fifteen
calls per hour. The effective cost rate is then
$\tilde{c}_{k} = \tilde{h}_{k} + \tilde{\theta}_{k}\,\tilde{p}_{k}$.
\newpage
\section{100-dimensional graphs}
\label{100_dim_graphs}
\begin{figure}[!htb]
    \centering
    \begin{subfigure}{1.0\textwidth}
        \centering
        \includegraphics[width=\textwidth]{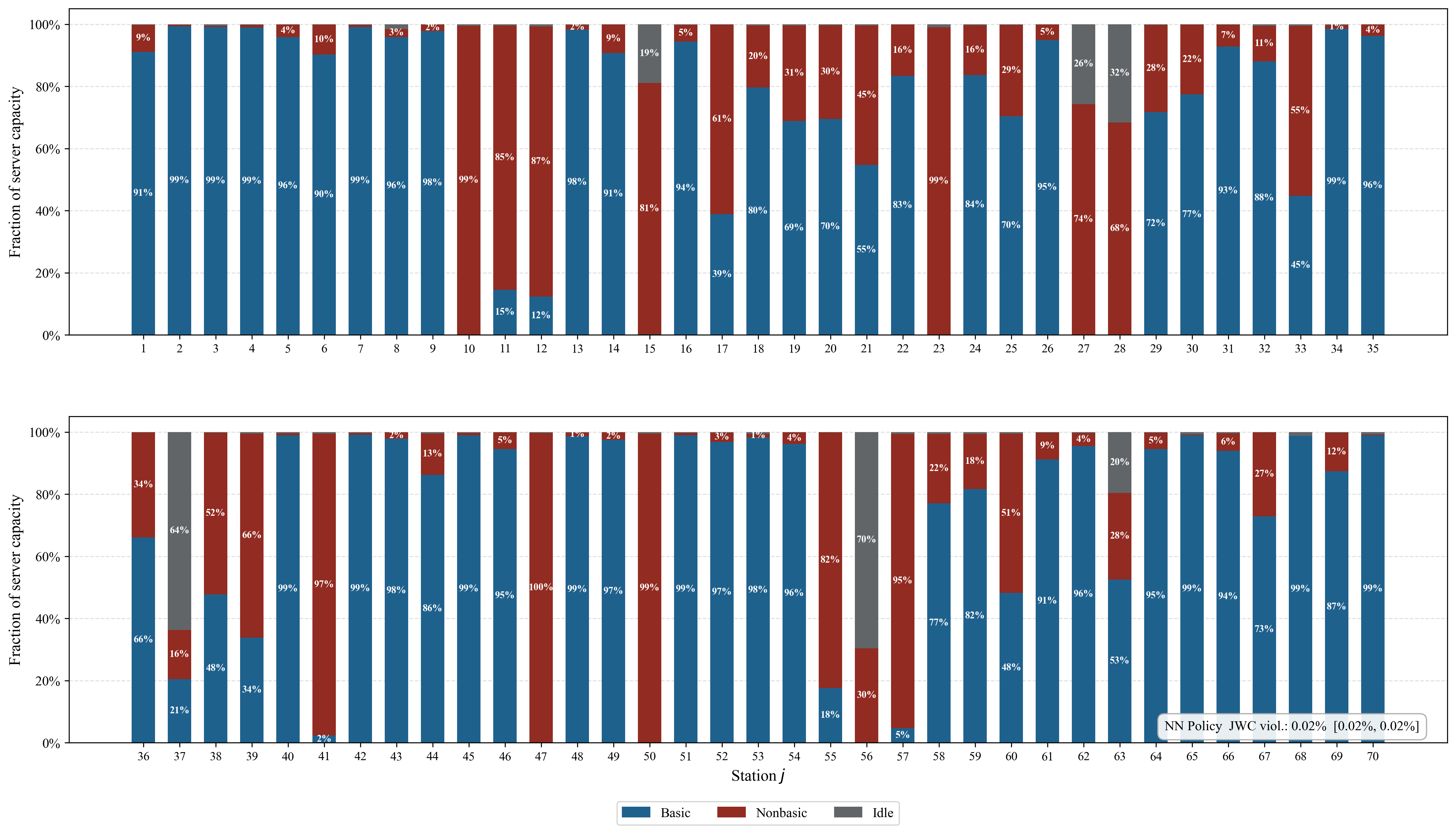}
        \caption{Proposed NN policy.}
        \label{fig:high_capacity_nn}
    \end{subfigure}

    \vspace{0.5cm}

    \begin{subfigure}{1.0\textwidth}
        \centering
        \includegraphics[width=\textwidth]{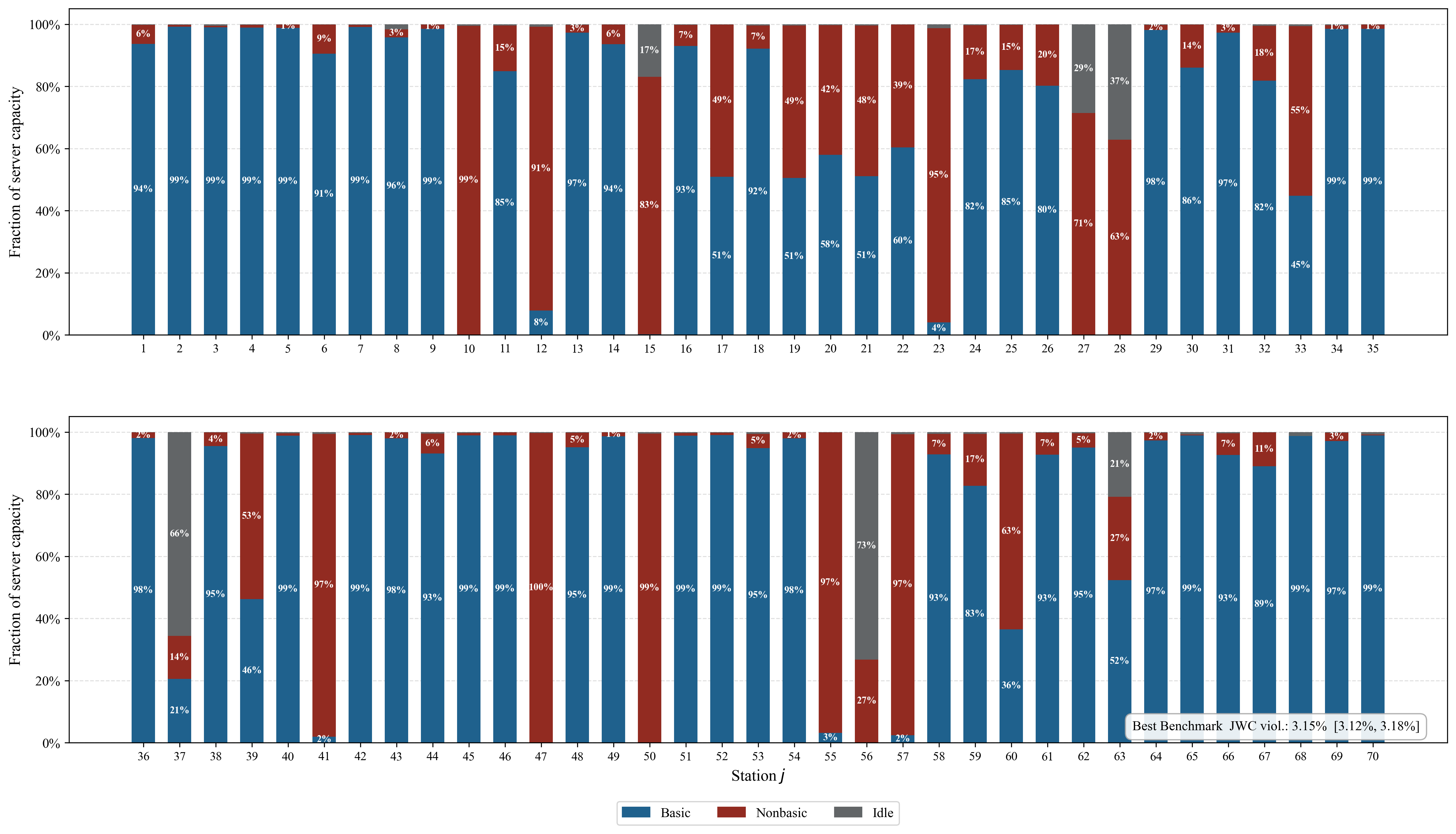}
        \caption{Best benchmark policy.}
        \label{fig:high_capacity_benchmark}
    \end{subfigure}

    \caption{
    Capacity decomposition by station in the 100-dimensional test problem. Each bar reports the
    fraction of station capacity allocated to basic activities, nonbasic activities, and idleness.
    Although the aggregate allocation patterns are broadly similar, the proposed NN policy violates
    joint work conservation substantially less often than the best benchmark policy.
    }
    \label{fig:high_capacity_comparison}
\end{figure}

\newpage
\section{A worked example for building a larger test problem}\label{sec:scaling_example}

For concreteness, we illustrate the algorithm by scaling a system with
$K = 2$ classes and $J = 2$ stations to $\tilde{K} = 4$ classes and
$\tilde{J} = 4$ stations. The original system corresponds to an
$X$-model with edge set
$\mathcal{E} = \{(1,1),\, (1,2),\, (2,1),\, (2,2)\}$, service rates
\[
\mu = \begin{pmatrix} 2 & 3 \\ 1 & 4 \end{pmatrix},
\]
and staffing $N_{1} = N_{2} = 100$. The tree of basic activities is
$\mathcal{T} = \{(1,1),\, (1,2),\, (2,2)\}$, which has $K + J - 1 = 3$
edges; see Figure~\ref{fig:orig_small_example}.

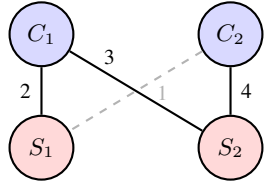
\begin{figure}[ht]
\centering
\begin{tikzpicture}[
    class/.style={circle, draw, thick, fill=blue!15, minimum size=24pt, inner sep=0pt, font=\footnotesize},
    station/.style={circle, draw, thick, fill=red!15, minimum size=24pt, inner sep=0pt, font=\footnotesize},
    treeedge/.style={thick},
    nontreeedge/.style={thick, dashed, gray!60},
]
\node[class]   (C1) at (0,1.5)    {$C_1$};
\node[class]   (C2) at (2.5,1.5)  {$C_2$};
\node[station] (S1) at (0,0)      {$S_1$};
\node[station] (S2) at (2.5,0)    {$S_2$};
\draw[treeedge] (C1) -- node[left, font=\scriptsize] {2} (S1);
\draw[treeedge] (C1) -- node[pos=0.3, above, font=\scriptsize] {3} (S2);
\draw[treeedge] (C2) -- node[right, font=\scriptsize] {4} (S2);
\draw[nontreeedge] (C2) -- node[pos=0.3, below, font=\scriptsize] {1} (S1);
\end{tikzpicture}
\caption{The $X$-model with $K = J = 2$. Solid edges denote basic
activities; dashed edges denote nonbasic activities. Edge labels are
the service rates $\mu_{kj}$.}
\label{fig:orig_small_example}
\end{figure}
\vspace{-3mm}

\paragraph{Step (i): Growing the tree.}
We add $\tilde{K} - K = 2$ new classes (indexed $3, 4$) and
$\tilde{J} - J = 2$ new stations (indexed $3, 4$). For each new node,
we draw its parent class $\kappa(k) \in \{1, 2\}$ or template station
$\tau(j) \in \{1, 2\}$ uniformly at random:

\begin{table}[ht]
\centering
\label{tab:template_assignments_example}
\begin{tabular}{lcc}
\toprule
New node & Draw & Realization \\
\midrule
Class 3   & $\kappa(3) \sim \mathrm{Unif}\{1, 2\}$ & $\kappa(3) = 1$ \\
Class 4   & $\kappa(4) \sim \mathrm{Unif}\{1, 2\}$ & $\kappa(4) = 2$ \\
Station 3 & $\tau(3) \sim \mathrm{Unif}\{1, 2\}$   & $\tau(3) = 1$ \\
Station 4 & $\tau(4) \sim \mathrm{Unif}\{1, 2\}$   & $\tau(4) = 2$ \\
\bottomrule
\end{tabular}
\caption{Template assignments for the new classes and stations in the worked example.}
\end{table}

\noindent We generate a random ordering of the new nodes as follows: 
$[\text{station 3},\; \text{class 3},\; \text{class 4},\; \text{station 4}]$,
and attach them to the tree one at a time. At each step, the new node
is connected to a node already in the tree, drawn uniformly from the
eligible set. A class $k$ and a station $j$ are eligible to be
connected if $(\kappa(k), \tau(j)) \in \mathcal{E}$. The added edge
$(k, j)$ receives the service rate
$\tilde{\mu}_{kj} = \mu_{\kappa(k),\tau(j)}$.

\noindent\textit{Step 1: attach station 3.} Since $\tau(3) = 1$,
eligible classes are $\{k \text{ in tree} : (\kappa(k), 1) \in \mathcal{E}\}
= \{1, 2\}$. Draw uniformly and assuming the outcome is 2, we attach it to class 2. Add edge $(2, 3)$
with $\tilde{\mu}_{2,3} = \mu_{2,1} = 1$.

\begin{figure}[H]
\centering
\begin{tikzpicture}[
    class/.style={circle, draw, thick, fill=blue!15, minimum size=24pt, inner sep=0pt, font=\footnotesize},
    station/.style={circle, draw, thick, fill=red!15, minimum size=24pt, inner sep=0pt, font=\footnotesize},
    newnode/.style={draw, very thick, dashed},
    treeedge/.style={thick},
    newedge/.style={very thick, green!50!black},
]
\node[class]   (C1) at (0,1.5)    {$C_1$};
\node[class]   (C2) at (2.5,1.5)  {$C_2$};
\node[station] (S1) at (0,0)      {$S_1$};
\node[station] (S2) at (2.5,0)    {$S_2$};
\node[station, newnode] (S3) at (5,0) {$S_3$};
\node[below=0pt, font=\scriptsize, gray] at (S3.south) {$(S_1)$};
\draw[treeedge] (C1) -- node[left, font=\scriptsize] {2} (S1);
\draw[treeedge] (C1) -- node[above, font=\scriptsize, pos=0.3] {3} (S2);
\draw[treeedge] (C2) -- node[right, font=\scriptsize] {4} (S2);
\draw[newedge]  (C2) -- node[above, font=\scriptsize] {1} (S3);
\end{tikzpicture}
\caption{Step 1 of the tree-growth procedure: station~3 is added to the tree and attached to class~2. The dashed node indicates the newly added station, the gray label identifies its template station, and the highlighted edge denotes the newly created basic activity.}

\end{figure}
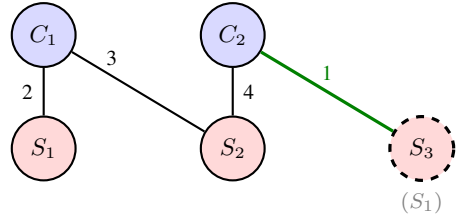

\medskip
\noindent\textit{Step 2: attach class 3.} Since $\kappa(3) = 1$,
eligible stations are $\{j \text{ in tree} : (1, \tau(j)) \in \mathcal{E}\}
= \{1, 2, 3\}$. Draw uniformly and assuming the outcome is 1, we attach it to station 1. We then add edge
$(3, 1)$ with $\tilde{\mu}_{3,1} = \mu_{1,1} = 2$.

\begin{figure}[H]
\small
\centering
\begin{tikzpicture}[
    class/.style={circle, draw, thick, fill=blue!15, minimum size=24pt, inner sep=0pt, font=\footnotesize},
    station/.style={circle, draw, thick, fill=red!15, minimum size=24pt, inner sep=0pt, font=\footnotesize},
    newnode/.style={draw, very thick, dashed},
    treeedge/.style={thick},
    newedge/.style={very thick, green!50!black},
]
\node[class, newnode]   (C3) at (-2.5,1.5)  {$C_3$};
\node[font=\scriptsize, gray] at (C3.north) [above] {$(C_1)$};
\node[class]   (C1) at (0,1.5)    {$C_1$};
\node[class]   (C2) at (2.5,1.5)  {$C_2$};
\node[station] (S1) at (0,0)      {$S_1$};
\node[station] (S2) at (2.5,0)    {$S_2$};
\node[station, newnode] (S3) at (5,0) {$S_3$};
\node[below=0pt, font=\scriptsize, gray] at (S3.south) {$(S_1)$};
\draw[treeedge] (C1) -- node[left, font=\scriptsize] {2} (S1);
\draw[treeedge] (C1) -- node[above, font=\scriptsize, pos=0.3] {3} (S2);
\draw[treeedge] (C2) -- node[right, font=\scriptsize] {4} (S2);
\draw[treeedge] (C2) -- node[above, font=\scriptsize] {1} (S3);
\draw[newedge]  (C3) -- node[left, font=\scriptsize] {2} (S1);
\end{tikzpicture}
\caption{Step 2 of the tree-growth procedure: class~3 is added to the tree and attached to station~1. The dashed node indicates the newly added class, the gray label identifies its template class, and the highlighted edge denotes the newly added basic activity.}
\end{figure}
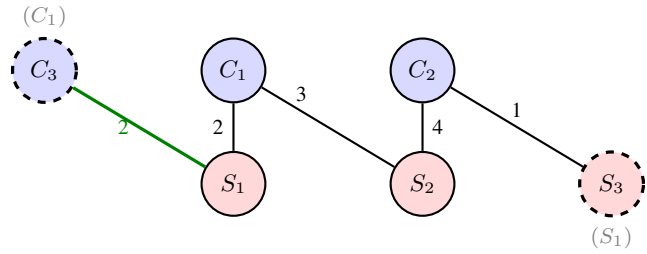
\vspace{-4mm}
\noindent\textit{Step 3: attach class 4.} Since $\kappa(4) = 2$,
eligible stations are $\{j \text{ in tree} : (2, \tau(j)) \in \mathcal{E}\}
= \{1, 2, 3\}$. Draw uniformly and assuming the outcome is 3, we attach it to station 3. We then add edge
$(4, 3)$ with $\tilde{\mu}_{4,3} = \mu_{2,1} = 1$.

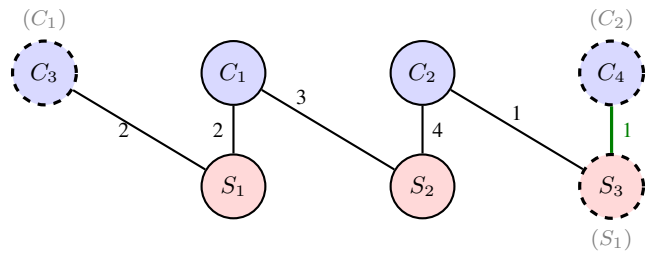
\begin{figure}[H]
\small
\centering
\begin{tikzpicture}[
    class/.style={circle, draw, thick, fill=blue!15, minimum size=24pt, inner sep=0pt, font=\footnotesize},
    station/.style={circle, draw, thick, fill=red!15, minimum size=24pt, inner sep=0pt, font=\footnotesize},
    newnode/.style={draw, very thick, dashed},
    treeedge/.style={thick},
    newedge/.style={very thick, green!50!black},
]
\node[class, newnode]   (C3) at (-2.5,1.5)  {$C_3$};
\node[font=\scriptsize, gray] at (C3.north) [above] {$(C_1)$};
\node[class]   (C1) at (0,1.5)    {$C_1$};
\node[class]   (C2) at (2.5,1.5)  {$C_2$};
\node[class, newnode]   (C4) at (5,1.5)  {$C_4$};
\node[font=\scriptsize, gray] at (C4.north) [above] {$(C_2)$};
\node[station] (S1) at (0,0)      {$S_1$};
\node[station] (S2) at (2.5,0)    {$S_2$};
\node[station, newnode] (S3) at (5,0) {$S_3$};
\node[below=0pt, font=\scriptsize, gray] at (S3.south) {$(S_1)$};
\draw[treeedge] (C1) -- node[left, font=\scriptsize] {2} (S1);
\draw[treeedge] (C1) -- node[above, font=\scriptsize, pos=0.3] {3} (S2);
\draw[treeedge] (C2) -- node[right, font=\scriptsize] {4} (S2);
\draw[treeedge] (C2) -- node[above, font=\scriptsize] {1} (S3);
\draw[treeedge] (C3) -- node[left, font=\scriptsize] {2} (S1);
\draw[newedge]  (C4) -- node[right, font=\scriptsize] {1} (S3);
\end{tikzpicture}
\caption{Step 3 of the tree-growth procedure: class~4 is added to the tree and attached to station~3. The dashed node indicates the newly added class, the gray label identifies its template class, and the highlighted edge denotes the newly added basic activity.}
\end{figure}

\medskip
\noindent\textit{Step 4: attach station 4.} Since $\tau(4) = 2$,
eligible classes are $\{k \text{ in tree} : (\kappa(k), 2) \in \mathcal{E}\}
= \{1, 2, 3, 4\}$. Draw uniformly and assuming the outcome is 4, we attach it to class 4. We then add edge
$(4, 4)$ with $\tilde{\mu}_{4,4} = \mu_{2,2} = 4$.

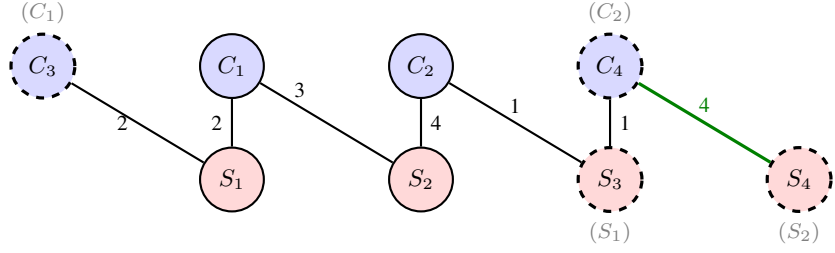
\begin{figure}
\centering
\small
\begin{tikzpicture}[
    class/.style={circle, draw, thick, fill=blue!15, minimum size=24pt, inner sep=0pt, font=\footnotesize},
    station/.style={circle, draw, thick, fill=red!15, minimum size=24pt, inner sep=0pt, font=\footnotesize},
    newnode/.style={draw, very thick, dashed},
    treeedge/.style={thick},
    newedge/.style={very thick, green!50!black},
]
\node[class, newnode]   (C3) at (-2.5,1.5)  {$C_3$};
\node[font=\scriptsize, gray] at (C3.north) [above] {$(C_1)$};
\node[class]   (C1) at (0,1.5)    {$C_1$};
\node[class]   (C2) at (2.5,1.5)  {$C_2$};
\node[class, newnode]   (C4) at (5,1.5)  {$C_4$};
\node[font=\scriptsize, gray] at (C4.north) [above] {$(C_2)$};
\node[station] (S1) at (0,0)      {$S_1$};
\node[station] (S2) at (2.5,0)    {$S_2$};
\node[station, newnode] (S3) at (5,0) {$S_3$};
\node[below=0pt, font=\scriptsize, gray] at (S3.south) {$(S_1)$};
\node[station, newnode] (S4) at (7.5,0) {$S_4$};
\node[below=0pt, font=\scriptsize, gray] at (S4.south) {$(S_2)$};
\draw[treeedge] (C1) -- node[left, font=\scriptsize] {2} (S1);
\draw[treeedge] (C1) -- node[above, font=\scriptsize, pos=0.3] {3} (S2);
\draw[treeedge] (C2) -- node[right, font=\scriptsize] {4} (S2);
\draw[treeedge] (C2) -- node[above, font=\scriptsize] {1} (S3);
\draw[treeedge] (C3) -- node[left, font=\scriptsize] {2} (S1);
\draw[treeedge] (C4) -- node[right, font=\scriptsize] {1} (S3);
\draw[newedge]  (C4) -- node[above, font=\scriptsize] {4} (S4);
\end{tikzpicture}
\caption{Step 4 of the tree-growth procedure: station~4 is added to the tree and attached to class~4. The rightmost dashed node indicates the newly added station, the gray label identifies its template station, and the highlighted edge denotes the newly added basic activity.}
\end{figure}

\paragraph{Final tree.} The resulting tree has
$\tilde{K} + \tilde{J} - 1 = 7$ edges,
\[
\tilde{\mathcal{T}} = \{(1,1),\, (1,2),\, (2,2),\, (2,3),\, (3,1),\,
(4,3),\, (4,4)\}.
\]

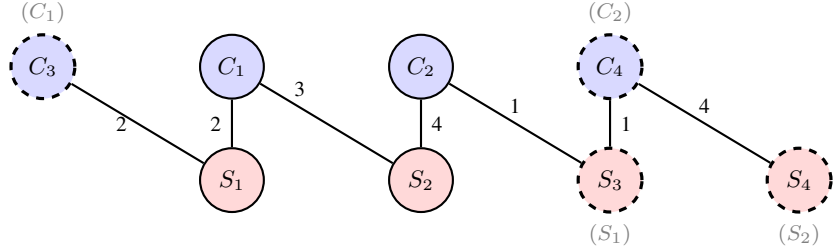
\begin{figure}[H]
\centering
\begin{tikzpicture}[
    class/.style={circle, draw, thick, fill=blue!15, minimum size=24pt, inner sep=0pt, font=\footnotesize},
    station/.style={circle, draw, thick, fill=red!15, minimum size=24pt, inner sep=0pt, font=\footnotesize},
    newnode/.style={draw, very thick, dashed},
    treeedge/.style={thick},
]
\node[class, newnode]   (C3) at (-2.5,1.5)   {$C_3$};
\node[font=\scriptsize, gray] at (C3.north) [above] {$(C_1)$};
\node[class]   (C1) at (0,1.5)      {$C_1$};
\node[class]   (C2) at (2.5,1.5)    {$C_2$};
\node[class, newnode]   (C4) at (5,1.5)      {$C_4$};
\node[font=\scriptsize, gray] at (C4.north) [above] {$(C_2)$};
\node[station] (S1) at (0,0)        {$S_1$};
\node[station] (S2) at (2.5,0)      {$S_2$};
\node[station, newnode] (S3) at (5,0)        {$S_3$};
\node[below=0pt, font=\scriptsize, gray] at (S3.south) {$(S_1)$};
\node[station, newnode] (S4) at (7.5,0)      {$S_4$};
\node[below=0pt, font=\scriptsize, gray] at (S4.south) {$(S_2)$};
\draw[treeedge] (C1) -- node[left, font=\scriptsize] {2} (S1);
\draw[treeedge] (C1) -- node[above, font=\scriptsize, pos=0.3] {3} (S2);
\draw[treeedge] (C2) -- node[right, font=\scriptsize] {4} (S2);
\draw[treeedge] (C2) -- node[above, font=\scriptsize] {1} (S3);
\draw[treeedge] (C3) -- node[left, font=\scriptsize] {2} (S1);
\draw[treeedge] (C4) -- node[right, font=\scriptsize] {1} (S3);
\draw[treeedge] (C4) -- node[above, font=\scriptsize] {4} (S4);
\end{tikzpicture}
\caption{Final tree obtained by expanding the original $2$-class, $2$-station system to a $4$-class, $4$-station system. Dashed nodes denote the classes and stations added during the expansion, gray labels identify their templates in the original system, and edge labels denote the service rates $\tilde{\mu}_{kj}$ on the basic activities.}
\label{fig:final_tree_example}
\end{figure}

\noindent\textit{Staffing.} We set $\tilde{N}_{\mathrm{total}} = 2000$
with $\tilde{N}_{\min} = 100$. Both templates carry $N_{\tau(j)} = 100$
agents in the original system, so the proportional weights are
$N_{\tau(j)}/\sum_{j'} N_{\tau(j')} = 1/4$ for every $j$, giving
\[
    \tilde{N}_{j} = \tilde{N}_{\min}
    + \left\lceil (\tilde{N}_{\mathrm{total}} - \tilde{J}\,\tilde{N}_{\min})
      \cdot \tfrac{1}{4} \right\rceil
    = 100 + 400 = 500,
    \qquad j = 1, \ldots, 4.
\]
For ease of exposition, we assume the system parameter is $\tilde{r} = 100$ and then we set the limiting staffing levels as $\tilde{\nu}_{j} = \tilde{N}_{j}/\tilde{r} = 5$ for $j = 1,\ldots, \tilde{J}$ in this example.

\paragraph{Step~(ii): Allocating capacity.} For each station $j$, we
draw the service fractions across its tree neighbors
$\mathcal{K}_{\tilde{\mathcal{T}}}(j)$ from a symmetric Dirichlet distribution. Assuming the displayed realizations, we arrive at the following table:
\vspace{-3mm}
\begin{table}[!htb]
\centering
\small
\caption{Realized capacity fractions on the tree activities in the worked example.}
\label{tab:capacity_fractions_example}
\begin{tabular}{lll}
\toprule
Station & Neighbors in $\tilde{\mathcal{T}}$ & Draw \\
\midrule
$j = 1$ & $\{1, 3\}$ & $(\tilde{\xi}^*_{1,1},\, \tilde{\xi}^*_{3,1}) = (0.7,\, 0.3)$ \\
$j = 2$ & $\{1, 2\}$ & $(\tilde{\xi}^*_{1,2},\, \tilde{\xi}^*_{2,2}) = (0.4,\, 0.6)$ \\
$j = 3$ & $\{2, 4\}$ & $(\tilde{\xi}^*_{2,3},\, \tilde{\xi}^*_{4,3}) = (0.8,\, 0.2)$ \\
$j = 4$ & $\{4\}$    & $\tilde{\xi}^*_{4,4} = 1$ \\
\bottomrule
\end{tabular}
\end{table}

\noindent\textit{Arrival rates.} The limiting arrival rates follow
from the demand constraint~\eqref{eq:limiting_staffing},
$\tilde{\lambda}_{k} = \sum_{j : (k,j) \in \tilde{\mathcal{T}}}
\tilde{\nu}_{j}\,\tilde{\mu}_{kj}\,\tilde{\xi}^*_{kj}$:
\begin{align*}
    \tilde{\lambda}_{1}
    &= 5 \cdot 2 \cdot 0.7 + 5 \cdot 3 \cdot 0.4 = 13, \\
    \tilde{\lambda}_{2}
    &= 5 \cdot 4 \cdot 0.6 + 5 \cdot 1 \cdot 0.8 = 16, \\
    \tilde{\lambda}_{3}
    &= 5 \cdot 2 \cdot 0.3 = 3, \\
    \tilde{\lambda}_{4}
    &= 5 \cdot 1 \cdot 0.2 + 5 \cdot 4 \cdot 1 = 21.
\end{align*}

\paragraph{Step~(iii): Nonbasic activities.} Recall
from~\eqref{eq:cs_tree_high} that the optimal dual variables satisfy
$\tilde{\nu}_{j}\,\tilde{\mu}_{kj}\,\tilde{\alpha}^*_{k}
= \tilde{\beta}^*_{j}$ on every basic activity. Because
$\tilde{\mathcal{T}}$ is a tree, fixing $\tilde{\alpha}^*_{1}$
determines the remaining duals through these relations as multipliers of $\tilde{\alpha}_{1}^{*}$
\begin{alignat*}{3}
    (1,1):&\quad \tilde{\beta}^*_{1}
    &&= \tilde{\nu}_{1}\,\tilde{\mu}_{1,1}\,\tilde{\alpha}^*_{1}
    &&= 10\tilde{\alpha}_{1}^{*}, \\
    (1,2):&\quad \tilde{\beta}^*_{2}
    &&= \tilde{\nu}_{2}\,\tilde{\mu}_{1,2}\,\tilde{\alpha}^*_{1}
    &&= 15\tilde{\alpha}_{1}^{*}, \\
    (2,2):&\quad \tilde{\alpha}^*_{2}
    &&= \tilde{\beta}^*_{2}/(\tilde{\nu}_{2}\,\tilde{\mu}_{2,2})
    &&= 0.75\tilde{\alpha}_{1}^{*}, \\
    (2,3):&\quad \tilde{\beta}^*_{3}
    &&= \tilde{\nu}_{3}\,\tilde{\mu}_{2,3}\,\tilde{\alpha}^*_{2}
    &&= 3.75\tilde{\alpha}_{1}^{*}, \\
    (3,1):&\quad \tilde{\alpha}^*_{3}
    &&= \tilde{\beta}^*_{1}/(\tilde{\nu}_{1}\,\tilde{\mu}_{3,1})
    &&= \tilde{\alpha}_{1}^{*}, \\
    (4,3):&\quad \tilde{\alpha}^*_{4}
    &&= \tilde{\beta}^*_{3}/(\tilde{\nu}_{3}\,\tilde{\mu}_{4,3})
    &&= 0.75\tilde{\alpha}_{1}^{*}, \\
    (4,4):&\quad \tilde{\beta}^*_{4}
    &&= \tilde{\nu}_{4}\,\tilde{\mu}_{4,4}\,\tilde{\alpha}^*_{4} &&= 15\tilde{\alpha}_{1}^{*}.
\end{alignat*}
For each
$(k, j) \in \tilde{\mathcal{E}} \setminus \tilde{\mathcal{T}}$,
strict complementary slackness requires
$\tilde{\mu}_{kj} < \tilde{\beta}^*_{j}/(\tilde{\nu}_{j}\,\tilde{\alpha}^*_{k})$.
If the rate $\mu_{\kappa(k),\tau(j)}$ from Step~(i) already satisfies
this bound, we use it directly. Otherwise we set
$\tilde{\mu}_{kj} = (1 - \delta)\,\tilde{\beta}^*_{j}/(\tilde{\nu}_{j}\,\tilde{\alpha}^*_{k})$
with $\delta = 0.01$:

\begin{table}[H]
\centering
\label{tab:nonbasic_service_rates_example}
\begin{tabular}{ccccc}
\toprule
Edge & Template rate & Bound & $\tilde{\mu}_{kj}$ & $\bar{c}_{kj}$ \\
$(k,j)$ & $\mu_{\kappa(k),\tau(j)}$ 
& $\tilde{\beta}^*_{j}/(\tilde{\nu}_{j}\tilde{\alpha}^*_{k})$ & & \\
\midrule
$(2,1)$ & 1 & 2.67 & 1    & 6.25 \\
$(3,2)$ & 3 & 3    & 2.97 & 0.15 \\
$(3,3)$ & 2 & 0.75 & 0.74 & 0.04 \\
$(3,4)$ & 3 & 3    & 2.97 & 0.15 \\
$(4,1)$ & 1 & 2.67 & 1    & 6.25 \\
$(4,2)$ & 4 & 4    & 3.96 & 0.15 \\
$(1,3)$ & 2 & 0.75 & 0.74 & 0.04 \\
$(1,4)$ & 3 & 3    & 2.97 & 0.15 \\
$(2,4)$ & 4 & 4    & 3.96 & 0.15 \\
\bottomrule
\end{tabular}
\caption{Construction of service rates for nonbasic activities. The bound is imposed to ensure strict complementary slackness, and $\bar c_{kj}$ denotes the resulting reduced cost.}
\end{table}

\medskip
\noindent The service rates for edges $(2,1)$ and $(4,1)$ are below
the bound and used directly; the service rates for the remaining
edges are set according to~\eqref{eq:nontree_high}. All nonbasic
activities have strictly positive reduced cost. By construction, the SPP admits a unique optimal solution with
$\rho^* = 1$, and the tree $\tilde{\mathcal{T}}$ corresponds to the
unique optimal basis.
\color{black}
\section{Computational Benchmarks}

\subsection{Optimal policies for MDP formulations of the low-dimensional test problems}\label{appendix_mdp_low_dim}
The state process $X(t)$ is a continuous-time Markov chain (CTMC) on $\mathbb{Z}_{+}^{K}$.  The control is the $K \times J$-dimensional process $\psi(x) \in \Psi(x)$, where the set of admissible controls is given by:
\begin{equation*}
\Psi(x) = \left\{ \psi \in \mathbb{R}_{+}^{|\mathcal{E}|} : 
\begin{array}{l}
\sum_{j \in \mathcal{J}(k)} \psi_{kj} \leq x_{k},\, \forall k \in \mathcal{K}, \quad \sum_{k \in \mathcal{K}(j)} \psi_{kj} \leq N_{j}, \, \forall j \in \mathcal{J}
\end{array}
\right\}.
\end{equation*}
The control $\psi_{kj}$ is the number of customers of class $k$ in service at station $j$. The transition rate matrix $Q^{\psi} = (Q^{\psi}(x, y))$ under policy $\psi$ is defined as follows: For $k = 1, \ldots, K$,
\begin{align}
&Q^{\psi}(x, x + e_k) = \lambda_k, \label{transition_arrival}\\
&Q^{\psi}(x, x - e_k) = \sum_{j \in \mathcal{J}(k)} \mu_{kj} \psi_{kj} + \theta_k \left(x_k - \sum_{j \in \mathcal{J}(k)} \psi_{kj}\right), \\
&Q^{\psi}(x, x) = - \sum_{k=1}^K \left[ \lambda_k + \sum_{j \in \mathcal{J}(k)} \mu_{kj} \psi_{kj} + \theta_k \left(x_k - \sum_{j \in \mathcal{J}(k)} \psi_{kj}\right) \right]. \label{transition_remaining}
\end{align}
Then, we define the optimal value function for the infinite-horizon discounted-cost problem as follows:
\begin{equation*}
        \tilde{V}(x) = \inf_{\psi \in \Psi(x)} \mathbb{E}_{x}^{\psi}\left\{ \int_{0}^{\infty} e^{-\alpha s} \sum_{k=1}^{K}c_{k}(x_{k} - \sum_{j \in \mathcal{J}(k)}\psi_{kj})ds\right\}.
\end{equation*}
The associated Bellman equation which helps us characterize the value function $\tilde{V}$ and the corresponding policy is as follows: For $x \in \mathbb{Z}_{+}^{K}$
\begin{equation}
    \alpha \tilde{V}(x) = \inf_{\psi \in \Psi(x)} \left\{\sum_{k=1}^{K} c_{k} (x_{k} - \sum_{j \in \mathcal{J}(k)} \psi_{kj}) + Q^{\psi}\tilde{V}(x)\right\}. \label{eq:bellman_equation_ctmc}
\end{equation}
Substituting the definition of $Q^{\psi}$ given in Equations (\ref{transition_arrival})--(\ref{transition_remaining}) into the Bellman equation (\ref{eq:bellman_equation_ctmc}) gives the following explicit form:
\begin{align}
        \alpha \tilde{V}(x) = \sum_{k=1}^{K} c_{k}x_{k} &+ \sum_{k=1}^{K}\lambda_{k}\Delta_{k}^{-}(x+e_{k}) - \sum_{k=1}^{K}\theta_{k}x_{k}\Delta_{k}^{-}(x)\\
        &-\sup_{\psi \in \Psi(x)}\left\{\sum_{k=1}^{K} \sum_{j \in \mathcal{J}(k)} \left(c_{k} + (\mu_{kj} - \theta_{k})\Delta_{k}^{-}(x)\right)\psi_{kj}\right\}, \label{eq:hjb_ctmc}
\end{align}
where for $k = 1,\ldots, K$, and $x \in \mathbb{Z}_{+}^{K}$, we have
$$\Delta_{k}^{-}(x) = \tilde{V}(x) - \tilde{V}(x-e_{k}).$$
The supremum term in Equation (\ref{eq:hjb_ctmc}) characterizes the optimal policy $\psi^{*}(x)$ as solution of the linear program over the feasible set $\Psi(x)$.

\noindent \textbf{Computational Method}. To numerically solve the Bellman equation, we use the policy iteration algorithm as shown in Algorithm \ref{algorithm_policy_iteration}.

\begin{algorithm}[!htb]    
\caption{Policy Iteration Algorithm for Low-Dimensional Test Problems}
\label{algorithm_policy_iteration}
\begin{minipage}{\textwidth}
\begin{algorithmic}[1]
\Statex \textbf{Input:} Discount rate $\alpha$, state space $S_{\bar{x}}$.
\Statex \textbf{Output:} Optimal value function $\tilde{V}(x)$ and optimal policy $\psi^{*}(x)$ for all $x \in S_{\bar{x}}$
\State Initialize a feasible policy $\psi^0(x) \in \Psi(x)$ for all $x \in S_{\bar{x}}$. Set $n \leftarrow 0$.
\Repeat
    \State Solve the linear system: 
    \[
    \alpha \tilde{V}^{\psi^n}(x) = \sum_{k=1}^K c_k\left(x_k - \sum_{j \in \mathcal{J}(k)} \psi^n_{kj}(x)\right) + Q^{\psi^n} \tilde{V}^{\psi^n}(x), \quad \forall x \in S_{\bar{x}}.
    \]
    \State For each $x \in S_{\bar{x}}$ and $k = 1,\ldots,K$, compute 
    \[
    \Delta_k^-(x) = \tilde{V}^{\psi^n}(x) - \tilde{V}^{\psi^n}(x - e_k).
    \]
    \State Update policy:
    \[
    \psi^{n+1}(x) \in \arg\max_{\psi \in \Psi(x)} \sum_{k=1}^K \sum_{j \in \mathcal{J}(k)} (c_k + (\mu_{kj} - \theta_k)\Delta_k^-(x)) \psi_{kj}.
    \]
    \State $n \leftarrow n+1$
\Until{\textit{convergence}: $\psi^{n+1}(x) = \psi^n(x)$ for all $x \in S_{\bar{x}}$}
\State \Return $\tilde{V}(x)$ and $\psi^{*}(x)$
\end{algorithmic}
\end{minipage}
\end{algorithm}

\noindent \textbf{Truncating the state space.} For computational feasibility, we truncate the state space by replacing it with $S$ defined as follows:
\begin{equation*}
    S = \{x \in \mathbb{Z}^{K}_{+}: 0 \leq x_{k} \leq \bar{x}_{k} \quad \text{for} \,\, k = 1, \ldots, K\}.
\end{equation*}
To define the behavior of the Markov chain in the boundary states, we modify the transition rate matrix $Q^{\psi}$. For $K = 2$, we define the vertical boundary of the state space as $E_{1} = \{(\bar{x}_{1}, x_{2}): 0 \leq x_{2} < \bar{x}_{2}\}$. Similarly, we define the horizontal boundary of the state space as $E_{2} = \{(x_{1}, \bar{x}_{2}): 0 \leq x_{1} < \bar{x}_{1}\}$.
Then, we set 
$$ 
\lambda(x) = \begin{cases}
(0,\lambda_{2})^{\prime}, & \text{if $x \in E_{1}$},\\
(\lambda_{1},0)^{\prime}, & \text{if $x \in E_{2}$},\\
(0,0)^{\prime}, & \text{if $x = (\bar{x}_{1}, \bar{x}_{2})$},\\
(\lambda_{1}, \lambda_{2})^{\prime}, & \text{otherwise}.
\end{cases}
$$
\vspace{-5mm}
\section{Implementation details of our computational method}\label{appendix_computational_method}

We implement our method using a fully connected deep neural network. We utilize Leaky ReLU, ELU, and SiLU as choices for the activation function, adapting code from the work of \citet{han2018solving} to our setting. The implementation is carried out in Python using the PyTorch package.

\noindent\textbf{Optimizer.} We use the Adam optimizer across all test problems. At each milestone listed in Tables~\ref{hyperparameter_low} and~\ref{hyperparameter_main}, the learning rate is multiplied by the decay factor $\gamma$.

\noindent\textbf{Reference policy.} Each test problem uses a drift network trained separately to approximate either the $c\mu$ rule or FSF rule. We denote the corresponding reference policy as ``$c\mu$'' or ``$\mu$'' in the tables below.

\noindent\textbf{Enforcement of positive gradient approximations.} Where the theory implies the learned gradient should be nonnegative, we enforce this constraint in one of two ways, depending on the choice of activation function in the hidden layers. For test problems using Leaky ReLU or ELU activations, we add a negative gradient penalty term $\Lambda > 0$ to the loss function to prevent negative gradient approximations. For test problems using SiLU activations, we instead apply a softplus function at the output layer, which guarantees nonnegativity structurally; in this case the penalty $\Lambda$ is unnecessary. Both approaches produced comparable results in preliminary experiments, and we report the configuration we used for each problem.
\vspace{-5mm}
\subsection{Hyperparameters used for the test problems}\label{hyperparameters}

\begin{table}[H]
\centering
\setlength\tabcolsep{6pt}
\renewcommand{\arraystretch}{1.2}
{\small
\scalebox{0.75}{
\begin{tabular}{lcc}
\toprule
\textbf{Hyperparameters} & \textbf{2D} & \textbf{2D Variant} \\
\midrule
Number of hidden layers             & 2                            & 4 \\
Number of neurons per layer         & 50                           & 100 \\
Time discretization steps $N$       & 200                          & 200 \\
Rolling horizon $T$                 & 1                            & 0.1 \\
Time step $\Delta T = T/N$          & $1/200$                      & $1/2000$ \\
Batch size                          & 512                          & 256 \\
\midrule
Total iterations                    & 5{,}000                      & 6{,}000 \\
Initial learning rate               & $1\text{e-}2$                & $1\text{e-}3$ \\
Milestones                          & $[1000,\,3000]$              & $[2000,\,4000,\,5000]$ \\
Learning rate decay factor $\gamma$ & 0.1                          & 0.2 \\
\midrule
Reference policy                    & $\mu$                        & $c\mu$ \\
Activation function                 & Leaky ReLU ($\alpha = 0.1$)  & SiLU \\
Negative gradient penalty $\Lambda$ & 0.6                          & --- \\
Initialization                      & Kaiming                      & Kaiming \\
Optimizer                           & Adam                         & Adam \\
\bottomrule
\end{tabular}
}}
\caption{Summary of the hyperparameters used for low-dimensional test problems.}
\label{hyperparameter_low}
\end{table}

\begin{table}[H]
\centering
\setlength\tabcolsep{6pt}
\renewcommand{\arraystretch}{1.2}
{\small
\scalebox{0.75}{
\begin{tabular}{lccc}
\toprule
\textbf{Hyperparameters} & \textbf{Main} & \textbf{Variant} & \textbf{100D} \\
\midrule
Number of hidden layers             & 4                        & 4                        & 4 \\
Number of neurons per layer         & 100                      & 100                      & 100 \\
Time discretization steps $N$       & 200                      & 200                      & 200 \\
Rolling horizon $T$                 & 1                        & 0.1                      & 1 \\
Time step $\Delta T = T/N$          & $1/200$                  & $1/2000$                 & $1/200$ \\
Batch size                          & 256                      & 768                      & 1024 \\
\midrule
Total iterations                    & 15{,}000                 & 7{,}000                  & 6{,}000 \\
Initial learning rate               & $1\text{e-}2$            & $1\text{e-}3$            & $1\text{e-}3$ \\
Milestones                          & $[1000,\,5000]$          & $[2000,\,4000,\,5000]$   & $[2000,\,4000,\,5000]$ \\
Learning rate decay factor $\gamma$ & 0.1                      & 0.2                      & 0.2 \\
\midrule
Reference policy                    & $\mu$                    & $c\mu$                   & $\mu$ \\
Activation function                 & ELU                      & SiLU                     & SiLU \\
Negative gradient penalty $\Lambda$ & 0.5                      & ---                      & --- \\
Initialization                      & Kaiming                  & Kaiming                  & Kaiming \\
Optimizer                           & Adam                     & Adam                     & Adam \\
\bottomrule
\end{tabular}
}}
\caption{Summary of the hyperparameters used for the main test problem, its variant, and the 100-dimensional test problem.}
\label{hyperparameter_main}
\end{table}

\end{document}